\documentclass[sigconf]{acmart}
\def\BibTeX{{\rm B\kern-.05em{\sc i\kern-.025em b}\kern-.08emT\kern-.1667em\lower.7ex\hbox{E}\kern-.125emX}}

\usepackage[T1]{fontenc}
\usepackage[colorinlistoftodos]{todonotes}
\usepackage{amsmath}
\usepackage{nicefrac}
\usepackage{flushend}
\usepackage{multirow}
\usepackage{soul}
\newcommand{\nfrac}{\nicefrac}

\usepackage{hyperref}

\definecolor{minttext}{RGB}{0,204,171}

\hypersetup{
    colorlinks=true,
    citecolor=minttext,
    linkcolor=minttext,
    urlcolor=minttext,
}

\usepackage[ruled, linesnumbered, noend]{algorithm2e}

\makeatletter
\newenvironment{algoproc}[1][ht]
  {
   \begin{algorithm}[#1]
   \long\def\@caption##1[##2]##3{%
     \par
     \begingroup\@parboxrestore
     \if@minipage\@setminipage\fi
     \normalsize \@makecaption{\AlCapSty{\AlCapFnt\algorithmcfname}}{\ignorespaces ##3}%
     \par\endgroup
   }}
  {\end{algorithm}}
\makeatother
\newenvironment{algorithm-hbox}{\begin{algoproc}}{\end{algoproc}}

\newtheorem{fact}{Fact}[section]
\newtheorem{definition2}{Definition}[section]

\renewcommand{\leq}{\leqslant}
\renewcommand{\geq}{\geqslant}
\newcommand{\down}[1]{\downarrow\mathclose(#1)}
\newcommand{\msf}[1]{\ensuremath{{\mathsf {#1}}}}

\newcommand{\abs}[1]{\left|#1\right|}   
\newcommand{\inparen}[1]{\left(#1\right)}             
\newcommand{\insquare}[1]{\left[#1\right]}             
\newcommand{\inangle}[1]{\left\langle#1\right\rangle} 
\newcommand{\expect}[1]{\mathbb{E}\insquare{#1}}

\newcommand{\chdag}{\ensuremath{\mbox{$\msf{ch}$-DAG}}}
\newcommand{\chrbc}{\msf{ch}-RBC}
\newcommand{\dagr}{\msf{DAG}-round}
\newcommand{\alef}{\msf{Aleph}}
\newcommand{\quickalef}{\msf{Quick Aleph}}
\newcommand{\asynchr}{\msf{async}-round}
\newcommand{\chdags}{\ensuremath{\mbox{$\msf{ch}$-DAGs}}}
\newcommand{\enc}{\mathrm{Enc}}
\newcommand{\dec}{\mathrm{Dec}}
\renewcommand{\H}{\mathrm{R}}
\newcommand{\hashg}{\mathrm{hash}}
\newcommand{\hash}{\mathrm{hash}}
\newcommand{\wt}{\widetilde}
\newcommand{\cD}{\mathcal{D}}
\newcommand{\cP}{\mathcal{P}}

\newcommand{\eps}{\varepsilon}

\newcommand\Z{\mathbb Z}
\newcommand\N{\mathbb N}

\newcommand\poly{\mathrm{poly}}
\newcommand{\ourskip}{\smallskip}

\renewcommand\footnotetextcopyrightpermission[1]{}

\begin{document}
\SetEndCharOfAlgoLine{}

\title[Aleph: Efficient Atomic Broadcast]{Aleph: Efficient Atomic Broadcast in Asynchronous Networks with Byzantine Nodes}

\author{Adam G\k{a}gol$^{1,2}$, Damian Le\'{s}niak$^{1,2}$, Damian Straszak$^{1}$, Micha\l{} \'{S}wi\k{e}tek$^{1,2}$}
\affiliation{\vspace{0.1cm}
\institution{$^1$ Aleph Zero Foundation}
\institution{$^2$ Jagiellonian University}
}

\date{June 2019}

\begin{abstract}
The spectacular success of Bitcoin and Blockchain Technology in recent years has provided enough evidence that a widespread adoption of a common cryptocurrency system is not merely a distant vision, but a scenario that might come true in the near future.
However, the presence of Bitcoin's obvious shortcomings such as excessive electricity consumption, unsatisfying transaction throughput, and large validation time (latency) makes it clear that a new, more efficient system is needed.

We propose a protocol in which a set of nodes maintains and updates a linear ordering of transactions that are being submitted by users.
Virtually every cryptocurrency system has such a protocol at its core, and it is the efficiency of this protocol that determines the overall throughput and latency of the system.
We develop our protocol on the grounds of the well-established field of Asynchronous Byzantine Fault Tolerant (ABFT) systems.
This allows us to formally reason about correctness, efficiency, and security in the strictest possible model, and thus convincingly prove the overall robustness of our solution. 

Our protocol improves upon the state-of-the-art HoneyBadgerBFT by Miller {\it et al}. by reducing the asymptotic latency while matching the optimal communication complexity.
Furthermore, in contrast to the above, our protocol does not require a trusted dealer thanks to a novel implementation of a trustless ABFT Randomness Beacon.

\end{abstract}


\maketitle

\section{Introduction}\label{sec:intro}

The introduction of Bitcoin and the Blockchain in the seminal paper of Satoshi Nakamoto~\cite{Nakamoto08} is already considered a pivotal point in the history of Finanancial Technologies.
While the rise of Bitcoin's popularity clearly shows that there is significant interest in a globally distributed currency system, the scalability issues have become a significant hurdle to achieve it.
Indeed, Bitcoin's latency of $30$ to $60$ minutes, the throughput of $7$ transactions per second, and the excessive power usage of the proof of work consensus protocol have motivated the search for alternatives.

At the core of virtually every cryptocurrency system lies a mechanism that collects transactions from users and constructs a total ordering of them, i.e., either explicitly or implicitly forming a blockchain of transactions.
This total ordering is then used to determine which transaction came first in case of double-spending attempts and thus to decide which transactions should be validated.
The protocol that guides the maintenance and growth of this total ordering is the heart of the whole system.
In Bitcoin, the protocol is Proof of Work, but there are also systems based on Proof of Stake~\cite{KRDO17, BG17} and modifications of these two basic paradigms~\cite{Popov16, KB18}.
Aside from efficiency, the primary concern when designing such protocols is their security.
While Bitcoin's security certainly has passed the test of time, numerous newly proposed designs claim security but fall short of providing convincing arguments.
In many such cases, serious vulnerabilities have been discovered, see~\cite{CV17, AGMAKT17}.

Given these examples, one may agree that for a new system to be trusted, strong mathematical foundations should guarantee its security.
What becomes important then are the assumptions under which the security claim is pursued -- in order to best imitate the highly adversarial execution environment of a typical permissionless blockchain system, one should work in the strictest possible model.
Such a model -- the Asynchronous Network model with Byzantine Faults -- has spawned a large volume of research within the field of Distributed Systems for the past four decades.
Protocols that are designed to work in this model are called Asynchronous Byzantine Fault Tolerant (ABFT) -- and are resistant to harsh network conditions: arbitrarily long delays on messages, node crashes, or even multiple nodes colluding in order to break the system.
Interestingly, even though these protocols seem to perfectly meet the robustness requirements for these kinds of applications, they have still not gained much recognition in the crypto-community.
This is perhaps because the ABFT model is often considered purely theoretical, and in fact, the literature might be hard to penetrate by an inexperienced reader due to heavy mathematical formalism.
Indeed, several of the most important results in this area~\cite{FLP85, Bracha87, CR93, BT83} have been developed in the '80s and '90s and were likely not meant for deployment at that time but rather to obtain best asymptotic guarantees.
Now, 30 years in the future, perhaps surprisingly, the ABFT model has become more practically relevant than ever, since the presence of bad actors in modern distributed ledger systems is inevitable, and their power ranges from blocking or taking over several nodes to even entirely shutting down large parts of the network.

In recent important work~\cite{MXCSS16}, Miller {\it et al}. presented the HoneyBadgerBFT (HBBFT) protocol, taking the first step towards practical ABFT systems.
HBBFT achieves optimal, constant communication overhead, and its validation time scales logarithmically with the number of nodes.
Moreover, importantly, HBBFT is rather simple to understand, and its efficiency has been also confirmed by running large-scale experiments.
Still, an unpleasant drawback of HBBFT, especially in the context of trustless applications, is that it requires a trusted dealer to initialize.

In this paper, we present a completely new ABFT protocol that keeps all of the good properties of HBBFT and improves upon it in two important aspects: tightening the complexity bounds on latency from logarithmic to constant and eliminating the need for a trusted dealer.
Furthermore even though being developed for the asynchronous setting, it matches the  optimistic-case performance of $3$-round validation time of state-of-the-art synchronous protocols~\cite{CL99}.
We believe that our protocol is simple to understand, due to its transparent structure that clearly separates the network layer from the protocol logic.
We also present a contribution that might be of independent interest: an efficient, completely trustless ABFT Randomness Beacon that generates common, unpredictable randomness. 
Since such a mechanism is necessary in many blockchain-based systems, we believe it might see future applications.
Finally, we believe that this paper, while offering a valuable theoretical improvement, also contributes to bridging the gap between the theory and practice of ABFT systems.

\section{Our Results}\label{sec:our-results}

The main goal of this paper\footnote{In order to make this text accessible also for readers with no background in Secure Distributed Systems, the narration of the paper focuses on providing intuitions and explaining the core ideas, at the cost of occasionally being slightly informal. At the same time, we stay mathematically rigorous when it comes to theorems and proofs.} is to design a distributed system that runs in a trustless network environment and whose purpose is to build a collective total ordering of messages submitted by users.
Apart from blockchain, such systems have several other applications, for example implementing state machine replication, where messages can be arbitrary operations to be executed by a state machine, and the purpose of the system is to keep the states of several copies of the machine consistent by executing the commands in the same order.

A major decision to make when designing systems of this kind is how to realistically model a network environment where such a system would run.
In the subsequent paragraphs, we introduce the model we work in and argue why we find it the most suitable for applications in distributed financial systems.
\ourskip

\noindent{\bf Nodes and Messages.}
The system consists of $N$ parties $\cP = \{\cP_1, \cP_2, \ldots, \cP_N\}$ that are called \emph{nodes}.
Each node $\cP_i$ identifies itself through its public key $pk_i$ for which it holds a private key $sk_i$ that allows it to sign messages if necessary.
Messages in the system are point-to-point, i.e., a node $\cP_i$ can send a message $m$ to another node $\cP_j$; the node $\cP_j$ is then convinced that the message came from $\cP_i$ because of the signature.
We assume signatures are unforgeable, so a node also cannot deny sending a particular message, since it was signed with its private key.
\ourskip

\noindent{\bf Network Model.}
The crucial part of the model are the assumptions about the message delivery and delays.
These assumptions are typically formulated by defining the abilities of an adversary, i.e., a powerful entity that watches the system and performs actions to slow it down or cause its malfunction.
The first assumption that is somewhat necessary is that the adversary cannot destroy messages that were sent, i.e., when a node sends a message, then it eventually reaches the recipient.
In practice, this assumption can be enforced by sending the same message multiple times if necessary.
Note also that an adversary with the ability to destroy messages would easily block every possible system, by just destroying all messages.
Given that, the most powerful ability an adversary could possibly have is to delay messages for an \emph{arbitrary} amount of time. 
This is what we assume and what is known in the literature as the \emph{Asynchronous Network Model}.
That means the adversary can watch messages being sent and schedule their delivery in an arbitrary order.

In contrast, another popular model is the \emph{Synchronous Network Model}\footnote{The Synchronous Model comes in several variants depending on whether the global bound $\Delta$ is known to the algorithm or not and whether there is an initial, finite period of asynchrony.}, where a global-bound $\Delta$ exists such that whenever a message is sent, it is delivered after time at most $\Delta$.
As one can imagine, this assumption certainly makes protocol design easier; however, the crucial question to address is: which of these models -- asynchronous or synchronous -- better fits the typical execution environment of a cryptocurrency system, i.e., the Internet.

While the asynchronous model may seem overly conservative, since no real-life adversary has full control over the network delays, there are mechanisms that may grant him partial control, such as timed DoS attacks.
Additionally, synchrony assumptions may be violated due to factors such as transient network partitions or a massive CPU load on several nodes preventing them from sending messages timely.

Finally, the key quality of any protocol meant for finance-related applications is its overall robustness, i.e., a very general notion of resilience against changing network conditions and against other unforeseeable factors.
The archetypical partially synchronous algorithm PBFT~\cite{CL99} (and its numerous variants~\cite{MGGR05, KADCW09, BKM18}) works in two modes: optimistic and pessimistic. 
The often-claimed simplicity of PBFT indeed manifests itself in the optimistic mode, but the pessimistic mode (that could as well become the default one under unfavorable network conditions) is in fact a completely separate algorithm that is objectively complex and thus prone to implementation errors.
Notably, Miller {\it et al.} in~\cite{MXCSS16} demonstrate an attack scenario on protocols from the PBFT family that completely blocks their operation.
In the same paper~\cite{MXCSS16}, it is also reasoned that asynchronous protocols are substantially more robust, as the model somewhat forces a single, homogeneous operation mode of the algorithm.
As such, we believe that the asynchronous model, even though it enforces stricter conditions, is the best way to describe the environment in which our system is going to operate.
\ourskip

\noindent{\bf Node Faults.}
For the kind of system we are trying to design, one cannot assume that all the nodes always proceed according to the protocol.
A node could simply crash, or go offline, and thus stop sending messages. 
Alternatively, a node or a set of nodes could act maliciously (be controlled by the adversary) and send arbitrary messages in order to confuse the remaining nodes and simply break the protocol.
These latter kinds of nodes are typically referred to in the literature as \emph{dishonest}, \emph{malicious}, \emph{faulty}, or \emph{Byzantine} nodes, and a protocol that solves a given problem in the presence of dishonest nodes is called Byzantine Fault Tolerant (BFT).
In cryptocurrency systems the presence of dishonest nodes is more than guaranteed as there will always be parties trying to take advantage of design flaws in order to gain financial benefits.
It is known that no asynchronous system can function properly (reach consensus) in the presence of $\nfrac{N}{3}$ or more dishonest nodes~\cite{BT83}; thus, we make the standard assumption that the total number of nodes is $N=3f+1$, and $f$ of them are dishonest.
\ourskip

\subsection{Our Contribution}
Before presenting our contributions, let us formalize the problem of building a total ordering of transactions, which in the literature is known under the name of \emph{Atomic Broadcast}.
\begin{definition2}[Atomic Broadcast]
Atomic Broadcast is a problem in which a set of nodes commonly constructs a total ordering of a set of transactions, where the transactions arrive at nodes in an on-line fashion, i.e., might not be given all at once. In a protocol that is meant to solve Atomic Broadcast, the following primitives are defined for every node:
\begin{enumerate}
    \item $\msf{Input}(tx)$ is called whenever a new transaction $tx$ is received by the node,
    \item $\msf{Output}(pos, tx)$ is called when the node decides to place the transaction $tx$ at the position $pos\in \N$.
\end{enumerate}
We say that such a protocol implements Atomic Broadcast if it meets all the requirements listed below:
\begin{enumerate}
    \item {\bf Total Order.} Every node outputs a sequence of transactions in an incremental manner, i.e., before outputting a transaction at position $pos\in \N$ it must have before output transactions at all positions $<pos$, and only one transaction can be output at a given position. 
    \item {\bf Agreement.} Whenever an honest node outputs a transaction $tx$ at position $pos$, every other honest node eventually outputs the same transaction at this position.
    \item {\bf Censorship Resilience.} Every transaction $tx$ that is input at some honest node is eventually output by all honest nodes.
\end{enumerate}
\end{definition2}
\noindent
The above definition formally describes the setting in which all nodes listen for transactions and commonly construct an ordering of them.
While the Total Order and Agreement properties ensure that all the honest nodes always produce the same ordering, the Censorship Resilience property is to guarantee that no transaction is lost due to censorship (especially important in financial applications) but also guarantees that the system makes progress and does not become stuck as long as new transactions are being received.

Let us briefly discuss how the performance of such an Atomic Broadcast protocol is measured.
The notion of time is not so straightforward when talking about asynchronous protocols, as the adversary has the power to arbitrarily delay messages between nodes.
For this reason, the running time of such a protocol is usually measured in the number of asynchronous rounds it takes to achieve a specific goal~\cite{CR93}.
Roughly speaking, the protocol advances to round number $r$ whenever all messages sent in round $r-2$ have been delivered; for a detailed description of asynchronous rounds, we refer the reader to Section~\ref{sec:async_round}.

The second performance measure that is typically used when evaluating such protocols is \emph{communication complexity}, i.e., how much data is being sent between the nodes (on average, by a single honest node).
To reduce the number of parameters required to state this result, we assume that a transaction, a digital signature, an index of a node, etc., all incur a constant communication overhead when sent; in other words, the communication complexity is measured in machine words, which are assumed to fit all the above objects.
Our first contribution is the \alef{} protocol for solving Atomic Brodcast whose properties are described in the following

\begin{theorem}[Atomic Broadcast]\label{thm:AB}\label{THM:AB}
The \alef{} protocol implements Atomic Broadcast over $N=3f+1$ nodes in asynchronous network, of which $f$ are dishonest and has the following properties:
\begin{enumerate}
    \item {\bf Latency.} For every transaction $tx$ that is input to at least $k$ honest nodes, the expected number of asynchronous rounds until it is output by every honest node is $O(\nfrac{N}{k})$.
    \item {\bf Communication Complexity.} The total communication complexity of the protocol in $R$ rounds\footnote{Here by a round we formally mean a \dagr{}s, as formally defined in Section~\ref{sec:asynch-as-dag}.} is $O(T+R\cdot N^2\log N)$ per node, where $T$ is the total number of transactions input to honest nodes during $R$ rounds.
\end{enumerate}
\end{theorem}
\noindent
We believe that the most important parameter of such a protocol and also the one that is hardest to optimize in practice\footnote{The bandwidth is potentially unbounded and can be improved in various ways, while the speed of light puts a hard limit on the network delay.} is the transaction latency.
This is why we mainly focus on achieving the optimal $O(1)$ latency\footnote{The $O(1)$ latency is achieved for transactions that are input to at least $k=\Omega(N)$ honest nodes. This is the same regime as in~\cite{MXCSS16} where it is assumed that $k\geq \nfrac{2}{3}N$.}.
In the Honey Badger BFT~\cite{MXCSS16} the latency is $\Omega(\log N)$ in the optimistic case, while it becomes $\Omega(\beta \log N)$ when there are roughly $\beta N^2 \log N$ unordered transactions in the system at the time when $tx$ is input.
In contrast, the latency of our protocol is $O(1)$ independently from the load.

Similarly to~\cite{MXCSS16} we need to make somewhat troublesome assumptions about the rate at which transactions arrive in the system to reason about the communication complexity, see Section~\ref{ssec:related_work} for comparison.
Still, in the regime of~\cite{MXCSS16} where a steady inflow of transactions is being assumed (i.e., roughly $N^2$ per round), we match the optimal $O(1)$ communication complexity per transaction of Honey Badger BFT~\cite{MXCSS16}.
As a practical note, we also mention that a variant of our protocol (Section~\ref{sec:practical}) achieves the optimal $3$-round validation latency in the ``optimistic case" akin to partially synchronous protocols from the PBFT family~\cite{CL99, BKM18} (see Section~\ref{subsec:synchronous} for a proof).
On top of that, our protocol satisfies the so-called {\it Responsiveness} property \cite{PassS17} which intuitively means that it makes progress at speed proportional to the instantaneous network throughput and is not slowed down by predefined timeouts.

Importantly, we believe that aside from achieving optimal latency and communication complexity, \alef{} is simple, clean, and easy to understand, which makes it well fit for a practical implementation.
This is a consequence of its modular structure, which separates entirely the network communication layer from the protocol logic.
More specifically, the only role of the network layer is to maintain a common (among nodes) data structure called a Communication History DAG, and the protocol logic is then stated entirely through combinatorial properties of this structure.
We introduce our protocol in Section~\ref{sec:ab} and formally prove its properties in Sections~\ref{sec:consensus_analysis},~\ref{sec:proofs-ab}  and \ref{ssec:proofs-coin}.

Another important property of our protocol is that unlike~\cite{CKS05, MXCSS16}, it does not require a trusted dealer.
The role of a trusted dealer is typically to distribute certain cryptographic keys among nodes before the protocol starts.
Clearly, in blockchain systems, no trusted entities can be assumed to exist, and thus a trusted setup is tricky if not impossible to achieve in real-world applications.

Our second contribution is an important, stand-alone component of our protocol that allows us to remove the trusted dealer assumption.
More specifically, it can be seen as a protocol for generating unpredictable, common randomness, or in other words, it implements an ABFT Randomness Beacon.
Such a source of randomness is indispensable in any ABFT protocol for Atomic Broadcast, since by the FLP-impossibility result~\cite{FLP85}, it is not possible to reach consensus in this model using a deterministic protocol.
Below we give a formalization of what it means for a protocol to implement such a randomness source.
The number $\lambda$ that appears below is the so-called security parameter (i.e., the length of a hash, digital signature, etc.). 
\begin{definition2}[Randomness Beacon]\label{def:randomnessbeacon}
We say that a pair of protocols $\inparen{\msf{Setup}, \msf{Toss}(m)}$ implements a Randomness Beacon if after running $\msf{Setup}$ once, each execution of $\msf{Toss}(m)$ (for any nonce $m\in \{0,1\}^\star$) results in $\lambda$ fresh random bits. More formally, we require
\begin{itemize}
    \item {\bf Termination.} All honest nodes correctly terminate after running either $\msf{Setup}$ or $\msf{Toss}(m)$,
    \item {\bf Correctness.} For a nonce $m\in \{0,1\}^\star$, $\msf{Toss}(m)$ results in all honest nodes outputting a common bitstring $\sigma_{m} \in\{0,1\}^{\lambda}$,
    \item {\bf Unpredictability.} No computationally bounded adversary can predict the outcome of $\msf{Toss}(m)$ with non-negligible probability.
\end{itemize}
\end{definition2}

\noindent 
In the literature such a source of randomness (especially the variant that outputs just a single bit) is often called a {\it Common Coin}~\cite{CKS05, MXCSS16, CKPS01} or a {\it Global Coin}~\cite{CR93}.
We note that once the $\msf{Toss}(m)$ protocol terminates, the value of $\sigma_m$ is known, and after revealing it, another execution of $\msf{Toss}(m)$ will not provide new random bits; thus, the Unpredictability property is meant to be satisfied only before $\msf{Toss}(m)$ is initiated.
As our second main contribution, in Section~\ref{sec:global_coin} we introduce the \msf{ABFT}-\msf{Beacon} protocol, and in Section~\ref{ssec:proofs-coin} we prove the following
\begin{theorem}[\msf{ABFT}-\msf{Beacon}]\label{thm:strongcoin}
The \msf{ABFT}-\msf{Beacon} protocol implements a Randomness Beacon $\inparen{\msf{Setup}, \msf{Toss}(m)}$ such that:
\begin{itemize}
    \item the $\msf{Setup}$ phase takes $O(1)$ asynchronous rounds to complete and has $O(N^2 \log N)$ communication complexity per node,
    \item each subsequent call to $\msf{Toss}(m)$ takes $1$ asynchronous round and has $O(N)$ communication complexity per node.
\end{itemize}
\end{theorem}
\noindent
We also remark that the \msf{ABFT}-\msf{Beacon} is relatively light when it comes to computational complexity, as the setup requires roughly $O(N^3)$ time per node, and each subsequent toss takes typically $O(N)$ time (under a slightly relaxed adversarial setting); see Section~\ref{sec:global_coin}.

As an addition to our positive results, in Section~\ref{sec:fork-bomb} we introduce the {\it Fork Bomb} -- a spam attack scenario that affects most known DAG-based protocols.
In this attack, malicious nodes force honest nodes to download exponential amounts of data and thus likely crash their machines.
This attack when attempted to prevent at the implementation layer by banning ``suspect nodes'' is likely to harm honest nodes as well.
Thus, we strongly believe that without a mechanism preventing this kind of attacks already at the protocol layer, liveness is not guaranteed.
The basic version of \alef{} is resistant against this attack through the use of reliable broadcast to disseminate information among nodes.
In Section~\ref{sec:practical} we also show a mechanism to defend against this attack for a gossip-based variant of \alef{}.

\subsection{Related Work}\label{ssec:related_work}
\ourskip

\noindent{\bf Atomic Broadcast.}
For an excellent introduction to the field of Distributed Computing and overview of Atomic Broadcast and Consensus protocols we refer the reader to the book~\cite{CGR11}.
A more recent work of~\cite{CV17} surveys existing consensus protocols in the context of cryptocurrency systems. 

The line of work on synchronous BFT protocols was initiated in~\cite{CL99} with the introduction of PBFT.
PBFT and its numerous variants~\cite{MGGR05, KADCW09, BKM18} tolerate byzantine faults, yet their efficiency relies on the good behavior of the network, and drops significantly when entering the (in some cases implicit) ``pessimistic mode''.
As thoroughly reasoned in~\cite{MXCSS16}, synchronous algorithms might not be well suited for blockchain-related applications, because of their lack of robustness and vulnerability to certain types of attacks.

In the realm of asynchronous BFT protocols, a large part of the literature focuses on the more classical Consensus problem, i.e., reaching binary agreement by all the honest nodes.
As one can imagine, ordering transactions can be reduced to a sequence of binary decisions, and indeed there are known generic reductions that solve Atomic Broadcast by running Consensus multiple times~\cite{CNV06, CKPS01,BE03, MHS11}.
However, all these reductions either increase the number of rounds by a super-constant factor or introduce a significant communication overhead.
Thus, even though Consensus protocols with optimal number of $O(1)$ rounds~\cite{CR93} and optimal communication complexity~\cite{CKS05, MMR15} were known early on, only in the recent work of~\cite{MXCSS16} the Honey Badger BFT (HBBFT) protocol with optimal $O(1)$ communication complexity per transaction was proposed.

The comparison of our protocol to HBBFT is not straightforward since the models of transaction arrivals differ.
Roughly, HBBFT assumes that at every round, every node has $\Omega(N^2)$ transactions in its buffer.
Under this assumption the communication complexity of HBBFT per epoch, per node is roughly $O(N^2)$, and also $\Omega(N^2)$ transactions are ordered in one epoch, hence the optimal $O(1)$ per transaction is achieved.
However, the design of HBBFT that is optimized towards low communication complexity has the negative effect that the latency might be large under high load.
More precisely, if $\beta N^2$ transactions are pending in the system\footnote{For the sake of this comparison we only consider transactions that have been input to $\Omega(N)$ honest nodes.} when $tx$ is being input, the latency of $tx$ is $\approx \beta $ epochs, thus $\approx \beta \log(N)$ rounds.
Our algorithm, when adjusted to this model would achieve $\approx \beta$ rounds of latency (thus $\log(N)$-factor improvement), while retaining the same, optimal communication complexity.

In this paper we propose a different assumption on the transaction buffers that allows us to better demonstrate the capabilities of our protocol when it comes to latency.
We assume that at every round the ratio between lengths of transaction buffers of any two honest nodes is at most a fixed constant.
In this model, our protocol achieves $O(1)$ latency, while a natural adaptation of HBBFT would achieve latency $O(\log(N))$, thus again a factor-$\log(N)$ improvement. 
A qualitative improvement over HBBFT that we achieve in this paper is that we completely get rid of the trusted dealer assumption.
We also note that the definitions of Atomic Broadcast slightly differ between this paper and~\cite{MXCSS16}: we achieve Censorship Resilience assuming that it was input to even a single honest node, while in~\cite{MXCSS16} it has to be input to $\Omega(N)$ nodes. 

The recent work\footnote{We would like to thank the anonymous reviewer for bringing this work to our attention.} of Abraham et al.~\cite{AMS19} studies a closely related Validated Asynchronous Byzantine Agreement (VABA) problem, which is, roughly speaking, the problem of picking one value out of $N$ proposed by the nodes.
The protocol in~\cite{AMS19} achieves $O(1)$ latency and has optimal communication complexity of $O(N)$ per node.
We believe that combining it with the ideas present in HoneyBadgerBFT can yield an algorithm with the same communication complexity as HoneyBadgerBFT but with latency improved by a factor of $\log N$.
However, when compared to ours, such a protocol still requires a trusted dealer and achieves weaker censorship resilience.

Finally, we remark that our algorithm is based on maintaining a DAG data structure representing the communication history of the protocol.
This can be seen as a realization of Lamport's ``happened-before'' relation~\cite{Lamport78} or the so-called Causal Order~\cite{HT94}.
To the best of our knowledge, the first instance of using DAGs to design asynchronous protocols is the work of~\cite{MM99}.
More recently DAGs gained more attention in the blockchain space~\cite{DH18, Baird16, Popov16}.
\ourskip

\noindent{\bf Common Randomness.} For a thorough overview of previous work on generating randomness in distributed systems and a discussion on the novelty of our solution we refer to Section~\ref{ssec:comparison_beacons}.


\section{Atomic Broadcast}\label{sec:ab}

This section is devoted to an outline and discussion of the \alef{} protocol. 
We start by sketching the overall idea of the algorithm and explaining how it works from a high-level perspective.
In this process we present the most basic variant of the \alef{} protocol, which already contains all the crucial ideas and gives a provably correct implementation of Atomic Broadcast.
On the theoretical side, however, this basic variant might suffer from a slightly sub-optimal communication complexity.
In Section~\ref{sec:proofs-ab} we describe a simple tweak to the protocol which allows us to finally achieve the communication complexity as claimed in Theorem~\ref{thm:AB}.
We refer to Sections~\ref{sec:consensus_analysis} and~\ref{sec:proofs-ab} for proofs of correctness and efficiency of \alef{}. 

\subsection{Asynchronous communication as a DAG}\label{sec:asynch-as-dag}
The concept of a ``communication round'' as explained in the preliminaries, is rather natural in the synchronous model, but might be hard to grasp when talking about asynchronous settings. This is one of the reasons why asynchronous models are, generally speaking, harder to work with than their (partially) synchronous counterparts, especially when it comes to proving properties of such protocols.  
\ourskip

\noindent{\bf Units and DAGs.}
To overcome the above issue, we present a general framework for constructing and analyzing asynchronous protocols that is based on maintaining (by all the nodes) the common ``communication history'' in the form of an append-only structure: a DAG (Directed Acyclic Graph).

The DAGs that we consider in this paper originate from the following idea: we would like to divide the execution of the algorithm into virtual rounds so that in every round $r$ every node emits exactly one \emph{Unit} that should be thought of as a message broadcast to all the other nodes. Moreover, every such unit should have ``pointers'' to a large enough number of units from the previous round, emitted by other nodes. Such pointers can be realized by including hashes of the corresponding units, which, assuming that our hash function is collision-free, allows to uniquely determine the ``parent units''. Formally, every unit has the following fields:
\begin{itemize}
    \item {\bf Creator.} Index and a signature of unit's creator.
    \item {\bf Parents.} A list of units' hashes. 
    \item {\bf Data.} Additional data to be included in the unit. 
\end{itemize}
The fact that a unit $U$ has another unit $V$ included as its parent signifies that the information piece carried by $V$ was known to $U$'s creator at the time of constructing $U$, i.e., $V$ causally preceeds $U$.
All the nodes are expected to generate such units in accordance to some initially agreed upon rules (defined by the protocol) and maintain their local copies of the common DAG, to which new units are continuously being added. 
\ourskip

\noindent{\bf Communication History DAGs.}
To define the basic rules of creating units note first that this DAG structure induces a partial ordering on the set of units. To emphasize this fact, we often write $V \leq U$ if either $V$ is a parent of $U$, or more generally (transitive closure) that $V$ can be reached from $U$ by taking the ``parent pointer'' several times.
This also gives rise to the notion of \emph{\dagr{}} of a unit $U$ that is defined to be the maximum length of a downward chain starting at $U$. In other words, recursively, a unit with no parents has \dagr{} $0$ and otherwise a unit has \dagr{} equal to the maximum of \dagr{}s of its parents plus one.
We denote the \dagr{} of a unit $U$ by $\H(U)$.
Usually we just write ``round'' instead of \dagr{} except for parts where the distinction between \dagr{} and \asynchr{} (as defined in Section~\ref{sec:async_round}) is relevant (i.e., mostly Section~\ref{sec:proofs-ab}).
We now proceed to define the notion of a \emph{\chdag{}} (communication history DAG) that serves as a backbone of the \alef{} protocol.

\begin{definition2}[\chdag]\label{def:chdag}
We say that a set of units $\cD$ created by $N=3f+1$ nodes forms a \chdag{}  if the parents of every unit in $\cD$ also belong to $\cD$ and additionally the following conditions hold true
\begin{enumerate}
    \item {\bf Chains.} For each honest node $\cP_i\in \cP$, the set of units in $\cD$ created by $\cP_i$ forms a chain.
    \item {\bf Dissemination.} Every round-$r$ unit in $\cD$  has at least $2f+1$ parents of round $r-1$.
    \item{\bf Diversity} Every unit in $\cD$ has parents created by pairwise distinct nodes.
\end{enumerate}
\end{definition2}
\noindent 
What the above definition tries to achieve is that, roughly, every node should create one unit in every round and it should do so after learning a large enough portion (i.e. at least $2f+1$) of units created in the previous round.
The purpose of the {\bf Chains} rule is to forbid \emph{forking}, i.e., a situation where a node creates more than one unit in a single round.
The second rule, {\bf Dissemination}, guarantees that a node creating a unit in round $r$ learned as much as possible from the previous round -- note that as there are only $N-f=2f+1$ honest nodes in the system, we cannot require that a node receives more than $2f+1$ units, as byzantine nodes might not have created them.
The unit may have additional parents, but the {\bf Diversity} rule ensures that they are created by different nodes - otherwise units could become progressively bigger as the \chdag{} grows, by linking to all the units in existence, hence increasing the communication complexity of the protocol.
\ourskip

\noindent{\bf Building DAGs.}
The pseudocode $\msf{DAG}-\msf{Grow}(\cD)$ provides a basic implementation of a node that takes part in maintaining a common DAG. Such a node initializes first $\cD$ to an empty DAG and then runs two procedures $\msf{CreateUnit}$ and $\msf{ReceiveUnits}$ in parallel. 
To create a unit at round $r$, we simply wait until $2f+1$ units of round $r-1$ are available and then, for every node, we pick its unit of highest round (i.e., of round at most $r-1$) and include all these $N$ (unless some nodes created no units at all, in which case $<N$) units as parents of the unit.
In other words, we wait just enough until we can advance to the next round, and attach to our round-$r$ unit everything we knew at that point in time.

\renewcommand{\algorithmcfname}{\msf{DAG}-\msf{Grow}$(\cD)$}
\begin{algorithm-hbox}
  \SetKwFunction{FMain}{$\msf{CreateUnit}(\msf{data})$}
  \SetKwProg{Fn}{ }{:}{}
  \Fn{\FMain}
    {
    \For{$r=0,1,2, \ldots$}
        {
        \If{r>0}
        {
            {\bf wait until} $|\{U\in \cD:\H(U)=r-1\}|\geq 2f+1$
        }
        $P \leftarrow \{$maximal $\cP_i$'s unit of round $<r$ in $\cD:\cP_i\in\cP\}$\\
        {\bf create} a new unit $U$ with $P$ as parents\\
        {\bf include} \msf{data} in $U$\\
        {\bf add} $U$ to $\cD$\\
        $\msf{RBC}(U)$
        }
    }
\SetKwFor{Loop}{loop forever}{}{}
  \SetKwFunction{FMain}{$\msf{ReceiveUnits}$}
  \SetKwProg{Fn}{ }{:}{}
  \Fn{\FMain}
    {
    \SetKwProg{Upon}{upon}{ do}{end}
    \Loop{}
    {
    \Upon{receiving a unit $U$ via \msf{RBC}}
        {
        {\bf add} $U$ to $\cD$
        }
    }
    }
  \caption{ }

\end{algorithm-hbox}

Both $\msf{CreateUnit}$ and $\msf{ReceiveUnits}$ make use of a primitive called \msf{RBC} that stands for \emph{Reliable Broadcast}. This is an asynchronous protocol that guarantees that every unit broadcast by an honest node is eventually received by all honest nodes.
More specifically we use the validated RBC protocol which also ensures that incorrect units (with incorrect signatures or incorrect number of parents, etc.) are never broadcast successfully.
Furthermore, our version of RBC forces every node to broadcast exactly one unit per round, thus effectively banning forks.
We refer to Section~\ref{sec:rbc} for a thorough discussion on Reliable Broadcast.

The validated RBC algorithm internally checks whether a certain $\msf{Valid}(U)$ predicate is satisfied when receiving a unit $U$.
This predicate makes sure that the requirements in Definition~\ref{def:chdag} are satisfied, as well as verifies that certain data, required by the protocol is included in $U$.
Consequently, only valid units are added to the local \chdags{} maintained by nodes.
Furthermore, as alluded above, there can be only a single copy of a unit created by a particular node in a particular round.
This guarantees that the local copies of \chdags{} maintained by different nodes always stay consistent.

Let us now describe more formally the desired properties of a protocol used to grow and maintain a common \chdag{}.
For this it is useful to introduce the following convention: we denote the local copy of the \chdag{} maintaned by the $i$th node by $\cD_i$.

\begin{definition2}\label{def:protocol-properties}
We distinguish the following properties of a protocol for constructing a common \chdag{}
\begin{enumerate}
    \item {\bf Reliable:} for every unit $U$ added to a local copy $\cD_i$ of an honest node $\cP_i$, $U$ is eventually added to the local copy $\cD_j$ of every honest node $\cP_j$.
    \item {\bf Ever-expanding:} for every honest node $\cP_i$ the local copy $\cD_i$ grows indefinitely, i.e., $\H(\cD_i):=\max\{r(U)\mid U\in \cD_i\}$ is unbounded.
    \item {\bf Fork-free:} whenever two honest nodes $i_1, i_2$ hold units $U_1 \in \cD_{i_1}$ and $U_2 \in \cD_{i_2}$ such that both $U_1$, $U_2$ have the same creator and the same round number, then $U_1=U_2$.
\end{enumerate}
\end{definition2}
\noindent 
Having these properties defined, we are ready to state the main theorem describing the process of constructing \chdag{} by the \msf{DAG}-\msf{Grow} protocol.

\begin{theorem}\label{thm:broadcast-to-dag}
The \msf{DAG}-\msf{Grow} protocol is reliable, ever-expanding, and fork-free
Additionally, during asynchronous round $r$ each honest node holds a local copy of $\cD$ of round at least $\Omega(r)$.
\end{theorem}
\noindent
For a proof we refer the reader to Section~\ref{sec:dag_grow_proof}.
\ourskip

\noindent{\bf Benefits of Using DAGs.}
After formally introducing the idea of a \chdag{} and explaining how it is constructed we are finally ready to discuss the significance of this concept.
First of all, \chdags{} allow for a clean and conceptually simple separation between the communication layer (sending and receiving messages between nodes) and the protocol logic (mainly deciding on relative order of transactions).
Specifically, the network layer is simply realized by running Reliable Broadcast in the background, and the protocol logic is implemented as running off-line computations on the local copy of the \chdag{}. 
One can think of the local copy of the \chdag{} as the \emph{state} of the corresponding node; all decisions of a node are based solely on its state.
One important consequence of this separation is that the network layer, being independent from the logic, can be as well implemented differently, for instance using ordinary broadcast or random gossip (see Section~\ref{sec:practical}).

In the protocol based on \chdags{} the concept of an adversary and his capabilities are arguably easier to understand.
The ability of the adversary to delay a message now translates into a unit being added to some node's local copy of the \chdag{} with a delay.
Nonetheless, every unit that has ever been created will still be eventually added to all the \chdags{} maintained by honest nodes.
A consequence of the above is that the adversary can (almost arbitrarily) manipulate the structure of the \chdag{}, or, in other words, he is able to force a given set of round-$(r-1)$ parents for a given round-$r$ unit.
But even then, it needs to pick at least $2f+1$ round-$(r-1)$ units, which enforces that more than a half of every unit's parents are created by honest nodes.

\subsection{Atomic broadcast via \chdag{}}\label{ssec:ab}
In this section we show how to build an Atomic Broadcast protocol based on the \chdag{} maintained locally by all the nodes.
Recall that nodes receive transactions in an on-line fashion and their goal is to construct a common linear ordering of these transactions. 
Every node thus gathers transactions in its local buffer and whenever it creates a new unit, all transactions from its buffer are included in the data field of the new unit and removed from the buffer.\footnote{This is the simplest possible strategy for including transactions in the \chdag{} and while it is provably correct it may not be optimal in terms of communication complexity. We show how to fix this in Section~\ref{sec:proofs-ab}.} Thus, to construct a common linear ordering on transactions it suffices to construct a linear ordering of units in the \chdag{} (the transactions within units can be ordered in an arbitrary yet fixed manner, for instance alphabetically). The ordering that we are going to construct also has the nice property that it \emph{extends} the ordering of units induced by the \chdag{} (i.e. the causal order). 

Let us remark at this point that all primitives that we describe in this section  take a local copy $\cD$ of the \chdag{} as one of their parameters and return either
\begin{itemize}
    \item a result (which might be a single bit, a unit, etc.), or
    \item $\bot$, signifying that the result is not yet available in $\cD$.
\end{itemize}
The latter means that in order to read off the result, the local copy $\cD$ needs to grow further.
We occasionally omit the $\cD$ argument, when it is clear from the context which local copy should be used.
\ourskip

\noindent{\bf Ordering Units.}
The main primitive that is used to order units in the \chdag, $\msf{OrderUnits}$, takes a local copy $\cD$ of the \chdag{} and outputs a list $\msf{linord}$ that contains a subset of units in $\cD$.
This list is a prefix of the global linear ordering that is commonly generated by all the nodes.
We note that $\msf{linord}$ will normally not contain all the units in $\cD$ but a certain subset of them.
More precisely, $\msf{linord}$ contains all units in $\cD$ except those created in the most recent (typically around $3$) rounds.
While these top units cannot be ordered yet, the structural information about the \chdag{}  they carry is used to order the units below them.
Intuitively, the algorithm that is run in the \chdag{} at round $r$ makes decisions regarding units that are several rounds deeper, thus the delay.

Note that different nodes might hold different versions of the \chdag{} at any specific point in time, but what we guarantee in the \chdag{} growing protocol is that all copies of \chdag{} are consistent, i.e., all the honest nodes always see exactly the same version of every unit ever created, and that every unit is eventually received by all honest nodes.
The function $\msf{OrderUnits}$ is designed in such a way that even when called on different versions of the \chdag{} $\cD_1$, $\cD_2$, as long as they are consistent, the respective outputs $\msf{linorder}_1$, $\msf{linorder}_2$ also agree, i.e., one of them is a prefix of the other.

\renewcommand{\algorithmcfname}{\alef{}}
\begin{algorithm-hbox}
  \SetKwFunction{FMain}{$\msf{OrderUnits}$}
  \SetKwProg{Fn}{ }{:}{}
  \Fn{\FMain{$\cD$}}
    {
        $\msf{linord} \leftarrow []$\\
        \For{$r=0,1,\dots, \H(\cD)$}
            {
            $V_r \leftarrow \msf{ChooseHead}(r, \cD)$\\
            \lIf {$V_r=\bot$}
            {
                {\bf break}
            }
            $\msf{batch}\leftarrow \{U\in \cD:U\leq V_r$ {\bf and} $U\notin \msf{linord}$\}\\
            {\bf order} $\msf{batch}$ deterministically\\
            {\bf append} $\msf{batch}$ to \msf{linord}\\
            }
        
        {\bf output} \msf{linord}
    }
  \SetKwFunction{FMain}{$\msf{ChooseHead}$}
  \SetKwProg{Fn}{ }{:}{}
  \Fn{\FMain{$r, \cD$}}
    {
        $\pi_r \leftarrow \msf{GeneratePermutation}(r,\cD)$\\
        \lIf{$\pi_r = \bot$}
        {{\bf output } $\bot$}
        \Else
        {
            $(U_1, U_2, \ldots, U_k)\leftarrow \pi_r$\\
            \For{$i=1,2, \ldots, k$}
            {
                \If{$\msf{Decide}(U_i, \cD)=1$}
                {
                    {\bf output} $U_i$ \\
                }
                \ElseIf{\msf{Decide}$(U_i, \cD)=\bot$}
                {
                    {\bf output} $\bot$ \\
                }
            }
            {\bf output} $\bot$
        }
    }

\caption{ }
\end{algorithm-hbox}

The $\msf{OrderUnits}$ primitive is rather straightforward.
At every round $r$, one unit $V_r$ from among units of round $r$ is chosen to be a ``head'' of this round, as implemented in $\msf{ChooseHead}$.
Next, all the units in $\cD$ that are less than $V_r$, but are not less than any of $V_0, V_1, \ldots, V_{r-1}$ form the $r$th batch of units.
The batches are sorted by their round numbers and units within batches are sorted topologically breaking ties using the units' hashes.
\ourskip

\noindent{\bf Choosing Heads.}
Perhaps surprisingly, the only nontrivial part of the protocol is choosing a head unit for each round.
It is not hard to see that simple strategies for choosing a head fail in an asynchronous network.
For instance, one could try picking always the unit created by the first node to be the head: this does not work because the first node might be byzantine and never create any unit.
To get around this issue, one could try another tempting strategy: to choose a head for round $r$, every node waits till round $r+10$, and declares as the head the unit of round $r$ in its copy of the \chdag{} that has the smallest hash.
This strategy is also doomed to fail, as it might cause inconsistent choices of heads between nodes: this can happen when some of the nodes see a unit with a very small hash in their \chdag{} while the remaining ones did not receive it yet, which might have either happened just by accident or was forced by actions of the adversary. 
Note that under asynchrony, one can never be sure whether missing a unit from some rounds back means that there is a huge delay in receiving it or it was never created (the creator is byzantine).
More generally, this also justifies that in any asynchronous BFT protocol it is never correct to wait for one fixed node to send a particular message.

Our strategy for choosing a head in round $r$ is quite simple: pick the first unit (i.e., with lowest creator id) that is visible by every node.
The obvious gap in the above is how do we decide that a particular unit is visible?
As observed in the example above, waiting a fixed number of rounds is not sufficient, as seeing a unit locally does not imply that all other nodes see it as well.
Instead, we need to solve an instance of \emph{Binary Consensus} (also called \emph{Binary Agreement}).
In the pseudocode this is represented by a $\msf{Decide}(U)$ function that outputs $0$ or $1$; we discuss it in the subsequent paragraph.

There is another minor adjustment to the above scheme that aims at decreasing the worst case latency, which in the just introduced version is $O(\log N)$ rounds\footnote{The reason is that the adversary could cause the first $\Omega(N)$ units to be decided $0$. Since the delay of each such decision is a geometric random variable with expectation $\theta(1)$, the maximum of $\Omega(N)$ of them is $\Omega(\log N)$.}.
When using a random permutation (unpredictable by the adversary) instead of the order given by the units creator indices, the latency provably goes down to $O(1)$.
Such an unpredictable random permutation is realized by the $\msf{GeneratePermutation}$ function.

\ourskip

\noindent{\bf Consensus.}
For a round-$(r+1)$ unit $U$, by $\down{U}$ we denote the set of all round-$r$ parents of $U$.
Consider now a unit $U_0$ in round $r$; we would like the result of $\msf{Decide}(U_0, \cD)$ to ``answer'' the question whether all nodes can see the unit $U_0$.
This is done through voting: starting from round $r+1$ every unit casts a ``virtual" vote\footnote{The idea of virtual voting was used for the first time in~\cite{MM99}.} on $U_0$.
These votes are called virtual because they are never really broadcast to other nodes, but they are computed from the \chdag{}.
For instance, at round $r+1$, a unit $U$ is said to vote $1$ if $U_0<U$ and $0$ otherwise, which directly corresponds to the intuition that the nodes are trying to figure out whether $U_0$ is visible or not.

\renewcommand{\algorithmcfname}{\alef{}-\msf{Consensus}$(\cD)$}
\begin{algorithm-hbox}
  \SetKwFunction{FMain}{$\msf{Vote}$}
  \SetKwProg{Fn}{ }{:}{}
  \Fn{\FMain{$U_0,U, \cD$}}
    {
    \lIf{$\H(U) \leq \H(U_0)+1$}{{\bf output}\footnotemark  $~[U_0 < U]$}
    \Else{
    $A \leftarrow \{\msf{Vote}(U_0, V, \cD) :V \in \down{U} \}$\\
    {\bf if} $A = \{\sigma\}$ {\bf then output} $\sigma$\\
    {\bf else output} $\msf{CommonVote}(U_0,\H(U), \cD)$
    }
    }
  \SetKwFunction{FMain}{$\msf{UnitDecide}$}
  \SetKwProg{Fn}{ }{:}{}
  \Fn{\FMain{$U_0,U, \cD$}}
    {
    \lIf{$\H(U)<\H(U_0)+2$}{{\bf output} $\bot$}
    $\msf{v}\leftarrow\msf{CommonVote}(U_0,U)$\\
    \lIf{$|\{V\!\!\in \down{U}\!\!:\!\msf{Vote}(U_0,\! V)\! =\! \msf{v} \}|\!\!\geq\!\! 2f\!\!+\!1$}
        {{\bf output} \msf{v}}
    \lElse
        {{\bf output} $\bot$}
    }

  \SetKwFunction{FMain}{$\msf{Decide}$}
  \SetKwProg{Fn}{ }{:}{}
  \Fn{\FMain{$U_0, \cD$}}
    {
    \If{$\exists_{U\in \cD} \msf{UnitDecide}(U_0,U, \cD)=\sigma\in\{0,1\}$}
        {
        {\bf output} $\sigma$
        }
    \lElse
        {{\bf output} $\bot$}}
  \caption{ }
\end{algorithm-hbox}
\footnotetext{In the expression $[U_0<U]$ we use the Iverson bracket notation, i.e., $[U_0<U]=1$ if $U_0<U$ and it is $0$ otherwise.}

Starting from round $r+2$ every unit can either make a final decision on a unit or simply vote again.
This process is guided by the function $\msf{CommonVote}(U_0, r', \cD)$ that provides a common bit $\in \{0,1\}$ for every round $r'\geq r+2$.
Suppose now that $U$ is of round $r'\geq r+2$ and at least $2f+1$ of its round-$(r'-1)$ parents in the \chdag{} voted $1$ for $U_0$, then if $\msf{CommonVote}(U_0, r', \cD)=1$, then unit $U$ is declared to decide $U_0$ as $1$.
Otherwise, if either there is no supermajority vote among parents (i.e., at least $2f+1$ matching votes) or the supermajority vote does not agree with the $\msf{CommonVote}$ for this round, the decision is not made yet.
In this case, the unit $U$ revotes using either the vote suggested by its parents (in case it was unanimous) or using the default $\msf{CommonVote}$.
Whenever any of the units $U \in \cD$ decides $U_0$ then it is considered decided with that particular decision bit.

Crucially, the process is designed in such a way that when some unit $U$ decides $\sigma \in \{0,1\}$ on some unit $U_0$ then we prove that no unit ever decides $\overline{\sigma}$ (the negation of $\sigma$) on $U_0$ and also that every unit of high enough round decides $\sigma$ on $U_0$ as well.
At this point it is already not hard to see that if a unit $U$ makes decision $\sigma$ on $U_0$ then it follows that {\bf all} the units of round $\H(U)$ (and any round higher than that) vote $\sigma$ on $U_0$. To prove that observe that if a unit $V$ of round $r$ votes $\sigma'$ then either:
\begin{itemize}
    \item all its parents voted $\sigma'$ and hence $\sigma'=\sigma$, because $U$ and $V$ have at least $f+1$ parents in common, or
    \item $\sigma'=\msf{CommonVote}(U_0, \H(V), \cD)$, but since $U$ decided $\sigma$ for $U_0$  then $\msf{CommonVote}(U_0, \H(U), \cD)=\sigma$ and thus $\sigma = \sigma'$ because $\H(U) = \H(V)$.
\end{itemize}
The above gives a sketch of ``safety'' proof of the protocol, i.e., that there will never be inconsistent decisions regarding a unit $U_0$.

Another property that is desirable for a consensus protocol is that it always terminates, i.e., eventually outputs a consistent decision.
In the view of the FLP Theorem~\cite{FLP85} this cannot be achieved in the absence of randomness in the protocol.
This is the reason why we need to inject random bits to the protocol and this is done in $\msf{CommonVote}$.
We show that, roughly, if $\msf{CommonVote}$ provides a random bit (that cannot be predicted by the adversary well in advance) then at every round the decision process terminates with probability at least $\nfrac{1}{2}$.
Thus, the expected number of rounds until termination is $O(1)$.

We provide formal proofs of properties of the above protocol in Section~\ref{sec:proofs-ab}.
In the next section, we show how the randomness in the protocol is generated to implement $\msf{CommonVote}$.

\subsection{Common Randomness}\label{ssec:randomness}
As already alluded to in Subsection~\ref{ssec:ab}, due to FLP-impossibility~\cite{FLP85} there is no way to achieve binary agreement in finite number of rounds when no randomness is present in the protocol.
We also note here that not any kind of randomness will do and the mere fact that a protocol is randomized does not ``protect'' it from FLP-impossibility.
It is crucial to keep the random bits hidden from the adversary till a certain point, intuitively until the adversary has committed to what decision to pursue for a given unit at a certain round.
When the random bit is revealed after this commitment, then there is a probability of $\nfrac{1}{2}$ that the adversary ``is caught'' and cannot delay the consensus decision further.

That means, in particular, that a source of randomness where the nodes are initialized with a common random seed and use a pseudo-random number generator to extract fresh bits is not sufficient, since such a strategy actually makes the algorithm deterministic.

Another issue when generating randomness for the protocol is that it should be \emph{common}, thus in other words, the random bit output by every honest node in a particular round $r$ should agree between nodes.
While this is strictly required for safety of the algorithm that is presented in the paragraph on Consensus, one can also consider simple variants thereof for which a relaxed version of commonness is sufficient.
More specifically, if the bit output at any round is common with probability $p\in (0,1]$ then one can achieve consensus in $O(\nfrac{1}{p})$ rounds.
Indeed, already in~\cite{Bracha87} Bracha constructed such a protocol and observed that if every node locally tosses a coin independently from the others than this becomes a randomness source that is common with probability $p\approx 2^{-N}$ and thus gives a protocol that takes $O(2^N)$ rounds to terminate.
We refer to Subsection~\ref{ssec:related_work} for an overview of previous work on common randomness.

In our protocol, randomness is injected via a single procedure $\msf{SecretBits}$ whose properties we formalize in the definition below
\begin{definition2}[Secret Bits]\label{def:secret-bits}
The $\msf{SecretBits}(i, revealRound)$ primitive takes an index $i\in \{1,2,\ldots,N\}$  and $ revealRound \in \N$ as parameters, outputs a $\lambda$-bits secret $s$, and has the following properties whenever initiated by the $\msf{ChooseHead}$ protocol
\begin{enumerate}
    \item no computationally bounded adversary can guess $s$ with non-negligible probability, as long as no honest node has yet created a unit of round $revealRound$,
    \item the secret $s$ can be extracted by every node that holds any unit of round $revealRound+1$.
\end{enumerate}
\end{definition2}
\noindent 
As one might observe, the first parameter $i$ of the $\msf{SecretBits}$ function seems redundant.
Indeed, given an implementation of $\msf{SecretBits}$, we could as well use $\msf{SecretBits}(1, \cdot)$ in place of $\msf{SecretBits}(i,\cdot)$ for any $i\in [N]$ (here and in the remaining part of the paper $[N]$ stands for $\{1, \ldots, N\}$)  and it would seemingly still satisfy the definition above.
The issue here is very subtle and will become clear only in Section~\ref{sec:global_coin}, where we construct a $\msf{SecretBits}$ function whose computability is somewhat sensitive to what $i$ it is run with\footnote{More precisely, intuitively $\msf{SecretBits}(i, \cdot)$ is not expected to work properly if for instance node $\cP_i$ has produced no unit at all. Also, importantly, the $\msf{ChooseHead}$ algorithm will never call $\msf{SecretBits}(i, \cdot)$ in such a case.}.
In fact, we recommend the reader to ignore this auxiliary parameter, as our main implementation of $\msf{SecretBits}(i, revealRound)$ is anyway oblivious to the value of $i$ and thus essentially there is exactly one secret per round in the protocol.

The simplest attempt to implement $\msf{SecretBits}$ would be perhaps to use an external trusted party that observes the growth of the \chdag{} and emits one secret per round, whenever the time comes. 
Clearly, the correctness of the above relies crucially on the dealer being honest, which is an assumption we cannot make: indeed if there was such an honest entity, why do not we just let him order transactions instead of designing such a complicated protocol? 

Instead, our approach is to utilize a threshold secret scheme (see~\cite{Shamir79}).
In short, at round $r$ every node is instructed to include its \emph{share} of the secret (that is hidden from every node) in the unit it creates.
Then, in round $r+1$, every node collects $f+1$ different such shares included in the previous round and reconstructs the secret from them.
Crucially, any set of $f$ or less shares is not enough to derive the secret, thus the adversary controlling $f$ nodes still needs at least one share from an honest node.
While this allows the adversary to extract the secret one round earlier then the honest nodes, this advantage turns out to be irrelevant, as, intuitively, he needed to commit to certain actions several turns in advance (see Section~\ref{sec:consensus_analysis} for a detailed analysis).

Given the $\msf{SecretBits}$ primitive, we are ready to implement\\ $\msf{GeneratePermutation}$ and $\msf{CommonVote}$ (the pseudocode is given in the table \alef{}-\msf{CommonRandomness}).

\renewcommand{\algorithmcfname}{\alef{}-\msf{CommonRandomness}}
\begin{algorithm-hbox}
  \SetKwFunction{FMain}{$\msf{CommonVote}$}
  \SetKwProg{Fn}{ }{:}{}
  \Fn{\FMain{$U_0, r, \cD$}}
    {
    \lIf{$r\leq\H(U_0)+3$}
        {
        {\bf output} $1$
        }
    \lIf{$r=\H(U_0)+4$}
        {
        {\bf output} $0$
        }
    \Else
        {
        $i \leftarrow \mbox{the creator of $U_0$}$\\
        $x\leftarrow \msf{SecretBits}(i, r, \cD)$\\
        \lIf{$x=\bot$}{{\bf output} $\bot$}
        {\bf output} the first bit of $\hash(x)$\\
        }
    }
  \SetKwFunction{FMain}{$\msf{GeneratePermutation}$}
  \SetKwProg{Fn}{ }{:}{}
  \Fn{\FMain{$r, \cD$}}
    {
    \For{each unit $U$ of round $r$ in $\cD$}
    {
        $i \leftarrow \mbox{the creator of $U$}$\\
        $x\leftarrow \msf{SecretBits}(i, r+4, \cD)$\\
        \lIf{$x=\bot$}{{\bf output} $\bot$}
        {\bf assign} $\mathrm{priority}(U) \leftarrow \hash(x||U) \in \{0,1\}^\lambda$\\
    }
    {\bf let} $(U_1, U_2, \ldots, U_k)$ be the units in $\cD$ of round $r$ sorted by $\mathrm{priority}(\cdot)$\\
    {\bf output} $(U_1, U_2, \ldots, U_k)$\\
    }
  \caption{ }
\end{algorithm-hbox}

In the pseudocode, by $\hash$ we denote a hash function\footnote{We work in the standard Random Oracle Model.} that takes an input and outputs a bistring of length $\lambda$.
We remark that the $\msf{CommonVote}$ at round $R(U_0)+4$ being $0$ is to enforce that units that are invisible at this round will be quickly decided negatively (see Lemma~\ref{lemma:negative-decisions}).
To explain the intuition behind $\msf{GeneratePermutation}$, suppose for a second that $\msf{SecretBits}(i, r+4)$ outputs the same secret $x$ independently of $i$ (as it is the case for our main algorithm).
Then the above pseudocode assigns a pseudorandom priority $\hash(x||U)$ (where, for brevity $U$ denotes some serialization of $U$) to every unit $U$ or round $r$ which results in a random ordering of these units, as required.

At this point, the only remaining piece of the protocol is the $\msf{SecretBits}$ procedure.
We provide two implementations in the subsequent Section (see Lemma~\ref{lemma:secretbits}): one simple, but requiring a trusted dealer, and the second, more involved but completely trustless.

\section{Randomness Beacon}\label{sec:global_coin}

The goal of this section is to construct an ABFT Randomness Beacon, in order to provide an efficient implementation of $\msf{SecretBits}$ that is required by our Atomic Broadcast protocol.
We start by a detailed review of previous works and how do they compare to the result of this paper.
Subsequently we describe how to extract randomness from threshold signatures, which is a basic building block in our approach, and after that we proceed with the description of our randomness beacon.

\subsection{Comparison to Previous Work on Distributed Randomness Beacons}\label{ssec:comparison_beacons}

We distinguish two main approaches for solving the problem of generating randomness in distributed systems in the presence of byzantine nodes: using Verifiable Secret Sharing (VSS) and via Threshold Signatures.
We proceed to reviewing these two approaches and subsequently explain what does our work bring to the field.
This discussion is succinctly summarized in Table~\ref{table:beacon}.

\noindent {\bf Verifiable Secret Sharing.}
The rough idea of VSS can be explained as follows: a dealer first initializes a protocol to distribute shares $(x_1, x_2, \ldots, x_N)$ of a secret $x$ to the nodes $1,2, \ldots, N$ so that afterwards every set of $(f+1)$ nodes can extract $x$ by combining their shares, but any subset of $f$ nodes cannot do that\footnote{The thresholds $f$ and $(f+1)$ here are not the only possible, but are most relevant for our setting, see~\cite{CKLS02} for a more detailed treatment.}.
An honest dealer is requested to pick $x$ uniformly at random (and erase $x$ from memory); yet as one can imagine, this alone is not enough to build a reliable, distributed randomness source, with the main issue being: how to elect a dealer with the guarantee that he is honest?
In the seminal paper~\cite{CR93} Canetti and Rabin introduce the following trick: let all nodes in the network act as dealers and perform VSS with the intention to combine all these secrets into one (say, by xoring them) that would be random and unpredictable.
Turning this idea into an actual protocol that works under full asynchrony is especially tricky, yet~\cite{CR93} manage to do that and end up with a protocol that has $O(N^7)$ communication complexity.
This $O(N^7)$ essentially comes from executing an $O(N^5)$ communication VSS protocol $N^2$ times.
Later on this has been improved by Cachin et al.~\cite{CKLS02} to $O(N^4)$ by making the VSS protocol more efficient.
The issue with the AVSS approach is that the communication cost of generating one portion of randomness is rather high -- it requires to start everything from scratch when new random bits are requested.
This issue is not present in the approach based on threshold signatures, where a one-time setup is run and then multiple portions of random bits can be generated cheaply.

\noindent {\bf Threshold Signatures.}
For a technical introduction to this approach we refer to Section~\ref{ssec:threshold_sigs}.
The main idea is that if the nodes hold keys for threshold signatures with threshold $f+1$, then the threshold signature $\sigma_m \in \{0,1\}^\lambda$ of a message $m$ is unpredictable for the adversary and provides us with roughly $\lambda$ bits of entropy.

The idea of using threshold signatures as a source of randomness is not new and has been studied in the literature~\cite{CKS05}, especially in the context of blockchain~\cite{HMW18,MXCSS16}.
While this technique indeed produces unbiased randomness {\bf assuming} that the tossing keys have been correctly dealt to all nodes, the true challenge becomes: {\it How to deal secret keys without involving a trusted dealer?}
This problem is known in the literature under the name of Distributed Key Generation (DKG).
There has been a lot of prior work on DKG~\cite{Pedersen91,GJKR07,GJKR03,KHG12}, however none of the so far proposed protocols has been designed to run under full asynchrony.
Historically, the first implemenation of DKG was proposed in the work of Pedersen~\cite{Pedersen91}.
The network model is not formally specified in~\cite{Pedersen91}, yet one can implement it in the synchronous BFT model with $f$ out of $N=3f+1$ nodes being malicious.
Pedersen's protocol, as later shown by Gennaro et al. in~\cite{GJKR07}, does not guarantee a uniform distribution over private keys; in the same work~\cite{GJKR07} a fix is proposed that closes this gap, but also in a later work~\cite{GJKR03} by the same set of authors it is shown that Pedersen's protocol is still secure\footnote{We note that out protocol can be seen as an adaptation of Pedersen's and thus also requires an ad-hoc proof of security (provided in Lemma~\ref{lemma:multicoin-cannot-be-predicted2}), as the distribution of secret keys might be skewed.}.
The DFINITY randomness beacon~\cite{HMW18} runs DKG as setup and subsequently uses threshold signatures to generate randomness.

We note that the AVSS scheme by Cachin et al.~\cite{CKLS02} can be also used as a basis of DKG since it is polynomial-based and has the additional property of every node learning a certain "commitment" to the secret shares obtained by all the nodes, which in the setting of threshold signatures corresponds to {\it public keys} of the threshold scheme (see Section~\ref{ssec:threshold_sigs}).
This allows to construct a DKG protocol as follows: each node runs its own instance of AVSS and subsequently the nodes agree on a subset of $\geq f+1$ such dealt secrets to build the keys for threshold signatures.
This scheme looks fairly simple, yet it is extremely tricky to implement correctly.
One issue is security, i.e., making sure that the adversary is not "front-running" and cannot reveal any secret too early, but more fundamentally: this algorithm requires {\bf consensus} to chooses a subset of secrets.
ABFT consensus requires {\bf randomness} (by FLP theorem~\cite{FLP85}), which we are trying to obtain, and thus we fall into an infinite loop.
To obtain consensus without the help of a trusted dealer one is forced to use variants of the Canetti-Rabin protocol~\cite{CR93} (recall that we are working here under the hard constraint of having no trusted dealer, which is not satisfied by most protocols in the literature) that unfortunately has very high communication complexity.
Nevertheless using this idea, one can obtain an ABFT DKG protocol with $O(N^4)$ communication complexity.
The approach of~\cite{KHG12} uses the ideas of~\cite{CKLS02} and roughly follows the above outlined idea, but avoids the problem of consensus by working in a relaxed model: weak synchrony (which lies in between partial synchrony and full asynchrony).

\begin{table}
\begin{center}
    \begin{tabular}{ | c | c | c |c |}
    
    \hline
    \multirow{2}{*}{Protocol }& \multirow{2}{*}{Model} & \multicolumn{2}{ |c| }{Communication}   \\ \cline{3-4}
    & & Setup~~~~~ & Query\\ \hline \hline
     {\bf This Work} &  async. BFT & $O(N^2 \log N)$ &  $O(N)$\\ \hline
     Canetti, Rabin~\cite{CR93} &  async. BFT & - &  $O(N^7)$ \\ \hline
     Cachin et al.~\cite{CKLS02} & async. BFT & - & $O(N^4)$ \\ \hline
      Kate et al. DKG \cite{KHG12} &  w. sync. BFT & $O(N^3)$ &  - \\ \hline
       Pedersen DKG \cite{Pedersen91, GJKR03} & sync. BFT & $O(N^2)$ &  -  \\ \hline
      Gennaro et al. DKG \cite{GJKR07} & sync. BFT & $O(N^2)$ &  -  \\ \hline
      DFINITY~\cite{HMW18} & sync. BFT & $O(N^2)$ &  $O(N)$ \\ \hline
     
      RandShare~\cite{SJKGGKFF17} & sync. BFT & - &  $O(N^2)$   \\ \hline
      SCRAPE~\cite{CD17} & sync. BFT & - &  $O(N^2)$   \\ \hline
    
    \hline
    \end{tabular}
\end{center}
\caption{Comparison of randomness beacons. The \emph{Model} column specifies what are the necessary assumptions for the protocol to work.
The \emph{Communication} columns specify communication complexities (per node) of a one-time setup phase (run at the very beginning), and of one query for fresh random bits.
Some of them were designed for DKG and thus we do not specify the complexity of query (yet it can be made $O(N)$ in all cases). Some protocols also do not run any setup phase, in which case the complexity is omitted.}\label{table:beacon}
\end{table}

\noindent {\bf Our Approach.}
We base our randomness beacon on threshold signatures, yet make a crucial observation that {\bf full DKG is not necessary} for this purpose.
Indeed what we run instead as setup (key generation phase) of our randomness beacon is a weaker variant of DKG.
Very roughly: the end result of this variant is that the keys are dealt to a subset of at least $2f+1$ nodes (thus possibly up to $f$ nodes end up without keys).
This allows us to construct a protocol with setup that incurs only $\wt{O}(N^2)$ communication, whereas a variant with full DKG\footnote{We do not describe this variant in this paper; it can be obtained by replacing univariate by bivariate polynomials and adjusting the protocol accordingly (see~\cite{CKLS02}).} would require $\wt{O}(N^3)$.

To obtain such a setup in an asynchronous setting we need to deal with the problem (that was sketched in the previous paragraph) of reaching consensus without having a common source of randomness in place.
Intuitively a ``disentanglement" of these two problems seems hard: in one direction this can be made formal via FLP Theorem~\cite{FLP85} (consensus requires randomness); in the opposite direction, intuitively, generating a common coin toss requires all the nodes to agree on a particular random bit.

Canetti and Rabin present an interesting way to circumvent this problem in~\cite{CR93}: their consensus algorithm requires only a "weak coin", i.e., a source of randomness that is common (among nodes) with constant, non-zero probability (and is allowed to give different results for different nodes otherwise).
Then, via a beautiful construction they obtain such a randomness source that (crucially) does not require any consensus.

Our way to deal with this problem is different.
Roughly speaking: in the setup each node "constructs" its  unbiased source of randomness and reliably broadcasts it to the remaining nodes.
Thus after this step, the problem becomes to pick one out of $N$ randomness sources (consensus).
To this end we make a binary decision on each of these sources and pick as our randomness beacon the first that obtained a decision of ``$1$"; however each of these binary consensus instances uses its own particular source of randomness (the one it is making a decision about).
At the same time we make sure that the nodes that were late with showing their randomness sources to others will be always decided $0$.
While this is the basic idea behind our approach, there are several nontrivial gaps to be filled in order to obtain $\wt{O}(N^2)$ communication complexity and $O(1)$ latency.

\noindent {\bf Other Designs.}
More recently, in the work of~\cite{SJKGGKFF17} a randomness source RandShare has been introduced, which runs a certain variant of Pedersen's protocol to extract random bits.
The protocol is claimed by the authors to work in the asynchronous model, yet fails to achieve that goal, as any protocol that {\bf waits for} messages from {\bf all the nodes} to proceed, fails to achieve liveness under asynchrony (or even under partial synchrony).
In the Table~\ref{table:beacon} we list it as synchronous BFT, as after some adjustments it can be made to run in this model.
In the same work~\cite{SJKGGKFF17} two other randomness beacons are also proposed: RandHound and RandHerd, yet both of them rely on strong, non-standard assumptions about the network and the adversary and thus we do not include them in Table~\ref{table:beacon}.
The method used by SCRAPE~\cite{CD17} at a high level resembles Pedersen's protocol but requires access to a blockchain in order to have a total order on messages sent out during the protocol executation.

Finally, we mention an interesting line of work on generating randomness based on VDFs (Verifiable Delay Functions~\cite{LW17,BBBF18,Wesolowski19,Pietrzak19}).
Very roughly, the idea is to turn a biasable randomness (such as the hash of a Bitcoin block) into unbiasable randomness via a VDF, i.e., a function $f: \{0,1\}^\lambda \to \{0,1\}^\lambda$ that cannot be computed quickly and whose computation cannot be parallelized, yet it is possible to prove that $f(x)=y$ for a given $y$ much faster than actually computing $f(x)$.
The security of this approach relies on the assumption that the adversary cannot evaluate $f$ at random inputs much faster than honest participants.

\subsection{Randomness from Threshold Signatures}\label{ssec:threshold_sigs}
In this subsection we present the main cryptographic component of our construction: generating randomness using Threshold Signatures. 
When using a trusted dealer, this component is already enough to implement $\msf{SecretBits}$, but the bulk of this section is devoted to proving that we can eliminate the need for a trusted dealer.

\noindent{\bf Randomness through Signatures.}
The main idea for generating randomness is as follows: suppose that there is a key pair $(tk, vk)$ of private key and public key, such that $tk$ is unknown, while $vk$ is public, and $tk$ allows to sign messages for some public-key crypto-system that is deterministic\footnote{A system that for a given message $m$ there exists only one correct signature for key pair $(tk, vk)$.}. (We refer to $tk$ as to the ``tossing key" while $vk$ stands for ``verification key"; we use these names to distinguish from the regular private-public key pairs held by the nodes.)
Then, for any message $m$, its digital signature $\sigma_m$ generated with respect to $tk$ cannot be guessed, but can be verified using $vk$, thus $\hash(\sigma_m)\in \{0,1\}^{\lambda}$ provides us with $\lambda$ random bits!
There seems to be a contradiction here though: how can the tossing key be secret and at the same time we are able to sign messages with it? 
Surprisingly, this is possible using {\it Threshold Cryptography}: the tossing key $tk$ is ``cut into pieces" and distributed among $N$ nodes so that they can jointly sign messages using this key but no node (or a group of dishonest nodes) can learn $tk$.

More specifically, we use a threshold signature scheme built upon BLS signatures~\cite{BLS04}.
Such a scheme works over a GDH group $G$, i.e., a group in which the computational Diffie-Hellman problem is hard (i.e. computing $g^{xy}$ given $g^x, g^y \in G$) but the decisional Diffie-Hellman problem is easy (i.e. verifying that $z=xy$ given $g^x, g^y, g^z \in G$). For more details and constructions of such groups we refer the reader to~\cite{BLS04}.
We assume from now on that $G$ is a fixed, cyclic, GDH group generated by $g\in G$ and that the order of $G$ is a large prime $q$.
A tossing key $tk$ in BLS is generated as a random element in $\Z_q$ and the public key is $y=g^{tk}$.
A signature of a message $m$ is simply $\sigma_m = \wt{m}^{tk}$, where $\wt{m} \in G$ is the hash (see~\cite{BLS04} for a construction of such hash functions) of the message being a random element of $G$.

\noindent{\bf Distributing the Tossing Key.} 
To distribute the secret tossing key among all nodes, the Shamir's Secret Sharing Scheme~\cite{Shamir79} is employed. 
A trusted dealer generates a random tossing key $tk$ along with a random polynomial $A$ of degree $f$ over $\Z_q$ such that $A(0)=tk$ and privately sends $tk_i =A(i)$ to every node $i=1,2, \ldots, N$.
The dealer also publishes verification keys, i.e., $VK=(g^{tk_1}, \ldots, g^{tk_N})$.
Now, whenever a signature of a message $m$ needs to be generated, the nodes generate {\it shares} of the signature by signing $m$ with their tossing keys, i.e., each node $i$ multicasts $\wt{m}^{A(i)}$.
The main observation to make is that now $\wt{m}^{A(0)}$ can be computed given only $\wt{m}^{A(j)}$ for $f+1$ different $j$'s (by interpolation) and thus it is enough for a node $\cP_j$ to multicast $\wt{m}^{A(j)}$ as its share and collect $f$ such shares from different nodes to recover $\sigma_m$.
On the other hand, any collection of at most $f$ shares is not enough to do that, therefore the adversary cannot sign $m$ all by himself.
For details, we refer to the pseudocode of $\msf{ThresholdSignatures}$.

\renewcommand{\algorithmcfname}{\msf{ThresholdSignatures}}
\begin{algorithm-hbox}

\SetKwFunction{FMain}{$\msf{GenerateKeys}$}
  \SetKwProg{Fn}{ }{:}{}
  \Fn{\FMain{}}
  {
    Let $G=\inangle{g}$ be a GDH group of prime order $q$\\
    generate a random polynomial $A$ of degree $f$ over $\mathbb{Z}_q$\\
    let $TK=(tk_1, \ldots, tk_N)$ with $tk_i= A(i)$ for $i\in [N]$,\\
    let $VK =(vk_1,\dots,vk_N)$ with $vk_i = g^{tk_i}$ for $i\in [N]$\\
    {\bf output} $(TK, VK)$
    }
  \SetKwFunction{FMain}{$\msf{CreateShare}$}
  \SetKwProg{Fn}{ }{:}{}
  \Fn{\FMain{$m, tk_i$}}
    {
    $\wt{m} \leftarrow \hashg(m)$\\
    {\bf output} $\wt{m}^{tk_i}$
    }
  \SetKwFunction{FMain}{$\msf{VerifyShare}$}
  \SetKwProg{Fn}{ }{:}{}
  \Fn{\FMain{$m, s, i, VK$}}
    {
    $\wt{m} \leftarrow \hashg(m)$\\
        \tcc{can do the check below since $G$ is GDH}
        \If
        {$\log_g(vk_i) = \log_{\wt{m}}(s)$}
        {
        {\bf output} \msf{True}
        }
        {\bf else output} \msf{False}
    }
    \SetKwFunction{FMain}{$\msf{GenerateSignature}$}
  \SetKwProg{Fn}{ }{:}{}
  \Fn{\FMain{$m, S, VK$}}
    {
    \tcc{Let $S=\{(s_i, i)\}_{i\in P}$ where $|P| = f+1$}
    \tcc{Assume $\forall_{j}~~ \msf{VerifyShare}(m, s_j, i_{j})=\msf{True}$}
    {\bf interpolate} $A(0)$, i.e., find $l_1, l_2, \ldots, l_{f+1}\in \Z_q$ s.t.
    $$\textstyle{A(0)=\sum_{j=1}^{f+1} l_j A(i_j)}$$\\
    $\sigma_m \leftarrow \prod_{j=1}^{f+1} s_j^{l_j}$\\
    {\bf output} $\sigma_m$
    }
\caption{ }
\end{algorithm-hbox}

\ourskip

\noindent{\bf SecretBits Through Threshold Signatures.}
Given the just introduce machinery of threshold signatures, the $\msf{SecretBits}(i, r)$ primitive is straightforward to implement.
Moreover, as promised in Section~\ref{sec:ab}, we  give here an implementation that is oblivious to its first argument, i.e., it does only depend on $r$, but not on $i$.

First of all, there is a setup phase whose purpose is to deal keys for generating secrets to all nodes.
We start by giving a simple version in which an honest dealer is required for the setup.
Subsequently, we explain how can this be replaced by a trustless setup, to yield the final version of $\msf{SecretBits}$.

In the simpler case, the trusted dealer generates tossing keys and verification keys $(TK, VK)$ for all the nodes using the $\msf{GenerateKeys}$ procedure, and then openly broadcasts $VK$ to all nodes and to every node $i$ he secretly sends the tossing key $tk_i$.

Given such a setup, when $\msf{SecretBits}(j, r)$ is executed by the protocol, every node $i$ can simply ignore the $j$ argument and generate
$$s_i(m)\leftarrow \msf{CreateShare}(m, tk_i)$$
where $m$ is a nonce determined from the round number, say $m=\mbox{``$r$''}$.
Next, after creating its round-$(r+1)$ unit $U$, the $P_i$ collects all the shares $S$ included in $U$'s parents at round $r$ and computes:
\begin{align*}
\sigma_m &\leftarrow \msf{GenerateSignature}(m,S,VK)\\
    s(m) &\leftarrow \hash(\sigma_m)
\end{align*}
and $s(m)\in \{0,1\}^\lambda$ is meant as the output of $\msf{SecretBits}(\cdot, r)$.

Finally, to get rid of the trusted dealer, in Subsection~\ref{ssec:informal-setup} we describe a trustless protocol that performs the setup instead of a trusted dealer.
The crucial difference is that the keys are not generated by one entity, but jointly by all the nodes.
Moreover, in the asynchronous version of the setup, every honest node $i$ learns the verification key $VK$, some key $tk_i$ and a set of ``share dealers'' $T \subseteq [N]$ of size $2f+1$, such that every node $P_i$ with $i\in T$ has a correct tossing key $tk_i$.
This, while being slightly weaker than the outcome of the setup of trusted dealer, still allows to implement~$\msf{SecretBits}$ as demostrated in the below lemma.

\begin{lemma}[Secret Bits]\label{lemma:secretbits}
The above scheme in both versions (with and without trusted dealer) correctly implements $\msf{SecretBits}$. 
\end{lemma}
\noindent
We refer the reader to Subsection~\ref{ssec:proofs-coin} for a proof.

\subsection{Randomness Beacon with Trustless Setup} \label{ssec:informal-setup}
In Section~\ref{ssec:threshold_sigs} we have discussed how to implement a Randomness Beacon based on a trusted dealer.
Here, we devise a version of this protocol that has all the benefits of the previous one and at the same time is completely trustless.
For the sake of clarity, we first provide a high level perspective of the new ideas and how do they combine to give a trustless Randomness Beacon, next we provide a short summary of the protocol in the form of a very informal pseudocode, and finally we fill in the gaps by giving details on how the specific parts are implemented and glued together. 
\ourskip

\noindent{\bf Key Boxes.}
Since no trusted dealer is available, perhaps the most natural idea is to let all the nodes serve as dealers simultaneously. 
More precisely we let every node emit (via RBC, i.e., by placing it as data in a unit) a tossing \emph{Key Box} that is constructed by a node $k$ (acting as a dealer) as follows:
\begin{itemize}
\item as in $\msf{GenerateKeys()}$, sample a random polynomial $$A_k(x)=\sum_{j=0}^f a_{k,j} x^j \in \Z_q[x]$$ of degree $f$, 
\item compute a commitment to $A_k$ as
$$C_k = (g^{a_{k,0}}, g^{a_{k,1}}, \ldots, g^{a_{k,f}}).$$
\item define the tossing keys $TK_k$ and verification keys $VK_k$
\begin{align*}
    tk_{k,i} &:= A_k(i) & \mbox{ for } i=1,2, \ldots, N\\
    vk_{k,i} &:= g^{tk_i} & \mbox{ for } i=1,2, \ldots, N.\\
\end{align*}
Note that in particular each verification key $vk_{k,i}$ for $i\in [N]$ can be computed from the commitment $C_k$ as $vk_{k,i} = \prod_{j=0}^f C_{k,j}^{i^j}$.
\item encrypt the tossing keys for every node $i$ using the dedicated public key\footnote{We assume that as part of PKI setup each node $i$ is given exactly $N$ different key pairs for encryption: $(sk_{k \to i}, pk_{k \to i})$ for $k\in [N]$. The key $pk_{k \to i}$ is meant to be used by the $k$th node to encrypt a message whose recipient is $i$ (denoted as $\enc_{k \to i}(\cdot)$). This setup is merely for simplicity of arguments -- one could instead have one key per node if using verifiable encryption or double encryption with labels.} $pk_{k \to i}$ as
$$e_{k,i} := \enc_{k\to i}\inparen{tk_{k,_i}}$$
and let $E_k := (e_{k,1}, e_{k,2}, \ldots, e_{k,N})$. 
\item the $k$th key box is defined as $KB_k = (C_k, E_k)$.
\end{itemize}
In our protocol, every node $\cP_k$ generates its own key box $KB_k$ and places $(C_k, E_k)$ in his unit of round $0$.
We define the $k$th key set to be $KS_k = (VK_k, TK_k)$ and note that given the key box $KB_k = (C_k, E_k)$, every node can reconstruct $VK_k$ and moreover, every node $i$ can decrypt his tossing key $tk_{k,i}$ from the encrypted part $E_k$, but is not able to extract the remaining keys.
The encypted tossing keys $E_k$ can be seen as a certain way of emulating "authenticated channels".

\ourskip

\noindent{\bf Verifying Key Sets.} 
Since at least $\nfrac{2}{3}$ of nodes are honest, we also know that $\nfrac{2}{3}$ of all the key sets are safe to use, because they were produced by honest nodes who properly generated the key sets and the corresponding key boxes an erased the underlying polynomial (and thus all the tossing keys). 
Unfortunately, it is not possible to figure out which nodes cheated in this process (and kept the tossing keys that were supposed to be erased).

What is even harder to check, is whether a certain publicly known key box $KB_k$ was generated according to the instructions above.
Indeed, as a node $\cP_i$ we have access only to our tossing key $tk_{k,i}$ and while we can verify that this key agrees with the verification key (check if $g^{tk_{k,i}} = vk_{k,i}$), we cannot do that for the remaining keys that are held by other nodes.
The only way to perform verification of the key sets is to do that in collaboration with other nodes.
Thus, in the protocol, there is a round at which every node ``votes'' for correctness of all the key sets it has seen so far.
As will be explained in detail later, these votes cannot be ``faked'' in the following sense: if a node $\cP_i$ votes on a key set $KS_k$ being incorrect, it needs to provide a proof that its key $tk_{k,i}$ (decrypted from $e_{k,i}$) is invalid, which cannot be done if $\cP_k$ is an honest dealer.
Thus consequently, dishonest nodes cannot deceive the others that a valid key set is incorrect.
\ourskip

\noindent{\bf Choosing Trusted Key Sets.}
At a later round these votes are collected and summarized locally by every node $\cP_i$ and a trusted set of key sets $T_i \subseteq [N]$ is determined.
Intuitively, $T_i$ is the set of indices $k$ of key sets such that:
\begin{itemize}
    \item the node $\cP_i$ is familiar with $KB_k$,
    \item the node $\cP_i$ has seen enough votes on $KS_k$ correctness that it is certain that generating secrets from $KS_k$ will be successful,
\end{itemize}
%
%
The second condition is determined based solely on the \chdag{} structure below the $i$th node unit at a particular round.
What will be soon important is that, even though the sets $T_i$ do not necessarily coincide for different $i$, it is guaranteed that $|T_i|\geq f+1$ for every $i$, and thus crucially at least one honest key set is included in each of them.
\ourskip

\noindent{\bf Combining Tosses.}
We note that once the set $T_i$ is determined for some index $i\in [N]$, this set essentially defines a global common source of randomness that cannot be biased by an adversary.
Indeed, suppose we would like to extract the random bits corresponding to nonce $m$.
First, in a given round, say $r$, every node should include its share for nonce $m$ corresponding to every key set that it voted as being correct.  
In the next round, it is guaranteed that the shares included in round $r$ are enough to recover the random secret $\sigma_{m,k}=m^{A_k(0)} \in G$ (the threshold signature of $m$ generated using key set $KS_k$) for every $k\in T_i$.
Since up to $f$ out of these secrets might be generated by the adversary, we simply take
$$\tau_{m,i} := \prod_{k\in T_i} \sigma_{m,k} = m^{\sum_{k\in T_i} A_k(0)} \in G$$
to obtain a uniformly random element of $G$ and thus (by hashing it) a random bitstring of length $\lambda$ corresponding to node $\cP_i$, resulting from nonce $m$.
From now on we refer to the $i$th such source of randomness (i.e., corresponding to $\cP_i$)  as $\msf{MultiCoin}_i$.
\ourskip

\noindent{\bf Agreeing on Common MultiCoin.}
So far we have said that every node $\cP_i$ defines locally its strong source of randomness $\msf{MultiCoin}_i$.
Note however that paradoxically, this abundance of randomness sources is actually problematic: which one shall be used by all the nodes to have a ``truly common'' source of randomness?
This is nothing other than a problem of ``selecting a head'', i.e.,  choosing one unit from a round -- a problem that we already solved in Section~\ref{sec:ab}!
At this point however, an attentive reader is likely to object, as the algorithm from Section~\ref{sec:ab} only works provided a common source of randomness.
Therefore, the argument seems to be cyclic as we are trying to construct such a source of randomness from the algorithm from Section~\ref{sec:ab}.
Indeed, great care is required here: as explained in Section~\ref{sec:ab}, all what is required for the $\msf{ChooseHead}$ protocol to work is a primitive $\msf{SecretBits}(i, r)$ that is supposed to inject a secret of $\lambda$ bits at the $r$th round of the protocol.
Not going too deep into details, we can appropriately implement such a method by employing the MultiCoins we have at our disposal.
This part, along with the construction of MultiCoins, constitutes the technical core of the whole protocol.
\ourskip

\noindent{\bf Combining Key Sets.}
Finally, once the head is chosen to be $l\in [N]$, from now on one could use $\msf{MultiCoin}_l$ as the common source of randomness. 
If one does not care about savings in communication and computational complexity, then the construction is over.
Otherwise, observe that tossing the $\msf{MultiCoin}_l$ at a nonce $m$ requires in the worst case $N$ shares from every single node.
This is sub-optimal, and here is a simple way to reduce it to just $1$ share per node.
Recall that every Key Set $KS_k$ that contributes to $\msf{MultiCoin}_l$ is generated from a polynomial $A_k \in \Z_q[x]$ of degree $f$. 
It is not hard to see that by simple algebraic manipulations one can combine the tossing keys and all the verification keys of all key sets $KS_k$ for $k\in T_l$ so that the resulting Key Set corresponds to the sum of these polynomials
\[
    A(x) := \inparen{\sum_{k\in T_l} A_k(x)}\in \Z_q[x].
\]
This gives a source of randomness that requires one share per nonce from every node; note that since the set $T_l$ contains at least one honest node, the polynomial $A$ can be considered random.
\ourskip

\noindent{\bf Protocol Sketch.}
We are now ready to provide an informal sketch of the $\msf{Setup}$ protocol and $\msf{Toss}$ protocol, see the $\msf{ABFT}-\msf{Beacon}$ box below. 
The content of the box is mostly a succinct summary of Section~\ref{ssec:informal-setup}.
What might be unclear at this point is the condition of the if statement in the $\msf{Toss}$ function.
It states that the tossing key $tk_i$ is supposed to be correct in order to produce a share: indeed it can happen that one of the key boxes that is part of the $\msf{MultiCoin}_l$ does not provide a correct tossing key for the $i$th node, in which case the combined tossing key $tk_i$ cannot be correct either.
This however is not a problem, as the protocol still guarantees that at least $2f+1$ nodes hold correct combined tossing keys, and thus there will be always enough shares at our disposal to run $\msf{ExtractBits}$.
\renewcommand{\algorithmcfname}{\msf{ABFT}-\msf{Beacon}}
\begin{algorithm-hbox}
    \SetKwFunction{FMain}{$\msf{Setup}$}
  \SetKwProg{Fn}{ }{:}{}
  \Fn{\FMain{}}
    {
    \tcc{Instructions for node $\cP_i$.}
    {\bf Initialize} growing the \chdag{} $\cD$\\
    In the data field of your units include (as specified in Section~\ref{ssec:setup-details}):\\
    \Indp
        {\bf At round $0$}: a key box $KB_i$,\\
        {\bf At round $3$}: votes regarding correctness of key boxes present in $\cD$,\\
        {\bf At rounds $\geq 6$}: shares necessary to extract randomness from \msf{SecretBits}.\\
    \Indm
    {\bf Run} $\msf{ChooseHead}$ to determine the head unit at round $6$ and let $l$ be its creator.\\
    {\bf Combine} the key sets $\{KS_j : j\in T_l\}$ and let $(tk_i, VK)$ be the corresponding tossing key and verification keys.
    }
    \SetKwFunction{FMain}{$\msf{Toss}$}
  \SetKwProg{Fn}{ }{:}{}
  \Fn{\FMain{$m$}}
    {
    \tcc{Code for node $\cP_i$.}
    \If{$tk_i$ is correct}
    {
        $s_i\leftarrow \msf{CreateShare}(m, tk_i)$\\
        {\bf multicast} $(s_i, i)$\\
    }
    {\bf wait} until receiving a set of $f+1$ valid shares $S$\\
    \tcc{validity is checked using $\msf{VerifyShare}()$}
{\bf output} $\sigma_m :=\msf{ExtractBits}(m,S,VK)$
    }
\caption{ }
\end{algorithm-hbox}

\subsection{Details of the Setup}\label{ssec:setup-details}
We provide some more details regarding several components of the $\msf{Setup}$ phase of $\msf{ABFT}-\msf{Beacon}$ that were treated informally in the previous subsection.
\ourskip

\noindent{\bf Voting on Key Sets.}
Just before creating the unit at round $3$ the $i$th node is supposed to inspect all the key boxes that are present in its current copy of the \chdag. Suppose $KB_k$ for some $k\in [N]$ is one of such key sets.
The only piece of information about $KB_k$ that is known to $\cP_i$ but hidden from the remaining nodes is its tossing key.
Node $\cP_i$ recovers this key by decrypting it using its secret key (dedicated for $k$) $sk_{k \to i}$ 
$$tk_{k,i} \leftarrow \dec_{k\to i}(e_{k,i}).$$
Now, if node $\cP_k$ is dishonest, it might have included an incorrect tossing key, to check correctness, the $i$th node verifies whether
\begin{equation*}
g^{tk_{k,i}} \stackrel{?}{=} vk_{k,i},
\end{equation*}
where $g$ is the fixed generator of the group $G$ we are working over. 
The $i$th node includes the following piece of information in its round-$3$ unit
\begin{equation*}
    \msf{VerKey}\inparen{KB_k, i} = 
    \begin{cases}
        1 &~~~~ \mbox{if  $tk_{k,i}$ correct} \\
        \dec_{k \to i}(e_{k,i})~~~ &~~ \mbox{oth.}
    \end{cases}
\end{equation*}
\noindent 
Note that if a node $\cP_i$ votes that $KB_k$ is incorrect (by including the bottom option of $\msf{VerKey}$ in its unit) it cannot lie, since other nodes can verify whether the plaintext it provided agrees with the ciphertext in $KB_k$ (as the encryption scheme is assumed to be deterministic) and if that is not the case, they treat a unit with such a vote as invalid (and thus not include it in their \chdag{}). 
Thus, consequently, the only way dishonest nodes could cheat here is by providing positive votes for incorrect key boxes.
This can not harm honest nodes, because by positively verifying some key set a node declares that from now on it will be providing shares for this particular key set whenever requested. 
If in reality the corresponding tossing key is not available to this node, it will not be able to create such shares and hence all its units will be henceforth considered invalid.

One important property this scheme has is that it is safe for an honest recipient $\cP_i$ of $e_{k,i}$ to reveal the plaintext decryption of $e_{k,i}$ in case it is not (as expected) the tossing key $tk_{k,i}$ -- indeed if $\cP_k$ is dishonest then either he actually encrypted some data $d$ in $e_{k,i}$ in which case he learns nothing new (because $d$ is revealed) or otherwise he obtained $e_{k,i}$ by some other means, in which case $\dec_{k \to i} (e_{k,i})$ is a random string, because no honest node ever creates a ciphertext encrypted with $p_{k \to i}$ and we assume that the encryption scheme is semantically secure.
Note also that if instead of having $N$ key pairs per node we used a similar scheme with every node  having just one key pair, then the adversary could reveal some tossing keys of honest nodes through the following attack: the adversary copies an honest node's (say $j$th) ciphertext $e_{j,i}$ and includes it as $e_{k,i}$ in which case $\cP_i$ is forced to reveal $\dec_i(e_{j,i}) = tk_{j,i}$ which should have remained secret!
This is the reason why we need dedicated key pairs for every pair of nodes.
\ourskip

\noindent{\bf Forming MultiCoins.}
The unit $V:=U[i; 6]$ created by the $i$th node in round $6$ defines a set $T_i \subseteq [N]$ as follows: $k\in N$ is considered an element of $T_i$ if and only if all the three conditions below are met
\begin{enumerate}
    \item $U[k; 0] \leq V$,
    \item For every $j\in [N]$ such that $U[j; 3]\leq V$ it holds that
    $$\msf{VerKey}\inparen{KB_k, j} = 1.$$
\end{enumerate}
\noindent 
At this point it is important to note that every node that has $U[i;6]$ in its copy of the \chdag{} can compute $T_i$ as all the conditions above can be checked deterministically given only the \chdag{}.
\ourskip

\noindent{\bf SecretBits via MultiCoins.}
Recall that to implement $\msf{CommonVote}$ and $\msf{GeneratePermutation}$ it suffices to implement a more general primitive $\msf{SecretBits}(i, r)$ whose purpose is to inject a secret at round $r$ that can be recovered by every node in round $r+1$ but cannot be recovered by the adversary till at least one honest node has created a unit of round $r$.

The technical subtlety that becomes crucial here is that the $\msf{SecretBits}(i,r)$ is only called for $r\geq 9$ and only by nodes that have $U[i;6]$ in their local \chdag{} (see Lemma~\ref{lemma:multicoin-can-be-tossed}).
More specifically, this allows us to implement $\msf{SecretBits}(i, \cdot)$ through $\msf{MultiCoin}_i$.
The rationale behind doing so is that every node that sees $U[i;6]$ can also see all the Key Sets that comprise the $i$th $\msf{MultiCoin}$ and thus consequently it ``knows'' what $\msf{MultiCoin}_i$ is\footnote{We epmhasize that using a fixed MultiCoin, for instance $\msf{MultiCoin}_1$ would not be correct here, as there is no guarantee the $1$st node has delivered its unit at round $6$. More generally, it is crucial that for different units $U_0$ we allow to use different MultiCoins, otherwise we would have solved Byzantine Consensus without randomness, which is impossible by the FLP Theorem~\cite{FLP85}. }.

Suppose now that we would like to inject a secret at round $r$ for index $i\in [N]$.
Define a nonce $m:=\mbox{``i||r''}$ and request all the nodes $\cP_k$ such that $U[i;6] \leq U[k, r]$ to include in $U[k, r]$ a share for the nonce $m$ for every Key Set $KS^j$ such that $j\in T_i$.
In addition, if $\cP_k$ voted that $KS_j$ is incorrect in round $3$, or $\cP_k$ did not vote for $KB_j$ at all (since $KB_j$ was not yet available to him at round $3$) then $\cP_k$ is not obligated to include a share (note that its unit $U[k,3]$ that is below $U[k, r]$ contains evidence that its key was incorrect).

As we then show in Section~\ref{ssec:setup-details}, given any unit $U\in\cD$ of round $r+1$ one can then extract the value of $\msf{MultiCoin}_i$ from the shares present in round-$r$ units in $\cD$.
Thus, consequently, the value of $\msf{SecretBits}(i, r)$ in a \chdag{} $\cD$ is available whenever any unit in $\cD$ has round $\geq r+1$, hence we arrive at

\begin{lemma}[Secret Bits from Multicoins]\label{lemma:secretbits2}
The above defined scheme based on MultiCoins correctly implements $\msf{SecretBits}$. 
\end{lemma}
\noindent
The proof of the above lemma is provided in Section~\ref{ssec:proofs-coin}.
\ourskip

\section{Acknowledgements}

First of all, authors would like to thank Matthew Niemerg for introducing us to the topic, constant support, and countless hours spent on valuable discussions.
Additionally, we would like to show our gratitude to Micha\l{} Handzlik, Tomasz Kisielewki, Maciej Gawron, and \L{}ukasz Lachowski, for reading the paper, proposing changes that improved its consistency and readability, and discussions that helped us to make the paper easier to understand.

\noindent 
This research was funded by the Aleph Zero Foundation.

\bibliographystyle{ACM-Reference-Format}
\bibliography{references}


\begin{thebibliography}{47}


\ifx \showCODEN    \undefined \def \showCODEN     #1{\unskip}     \fi
\ifx \showDOI      \undefined \def \showDOI       #1{#1}\fi
\ifx \showISBNx    \undefined \def \showISBNx     #1{\unskip}     \fi
\ifx \showISBNxiii \undefined \def \showISBNxiii  #1{\unskip}     \fi
\ifx \showISSN     \undefined \def \showISSN      #1{\unskip}     \fi
\ifx \showLCCN     \undefined \def \showLCCN      #1{\unskip}     \fi
\ifx \shownote     \undefined \def \shownote      #1{#1}          \fi
\ifx \showarticletitle \undefined \def \showarticletitle #1{#1}   \fi
\ifx \showURL      \undefined \def \showURL       {\relax}        \fi
\providecommand\bibfield[2]{#2}
\providecommand\bibinfo[2]{#2}
\providecommand\natexlab[1]{#1}
\providecommand\showeprint[2][]{arXiv:#2}

\bibitem[\protect\citeauthoryear{Abd{-}El{-}Malek, Ganger, Goodson, Reiter, and
  Wylie}{Abd{-}El{-}Malek et~al\mbox{.}}{2005}]%
        {MGGR05}
\bibfield{author}{\bibinfo{person}{Michael Abd{-}El{-}Malek},
  \bibinfo{person}{Gregory~R. Ganger}, \bibinfo{person}{Garth~R. Goodson},
  \bibinfo{person}{Michael~K. Reiter}, {and} \bibinfo{person}{Jay~J. Wylie}.}
  \bibinfo{year}{2005}\natexlab{}.
\newblock \showarticletitle{Fault-scalable Byzantine fault-tolerant services}.
  In \bibinfo{booktitle}{\emph{Proceedings of the 20th {ACM} Symposium on
  Operating Systems Principles 2005, {SOSP} 2005, Brighton, UK, October 23-26,
  2005}}. \bibinfo{pages}{59--74}.
\newblock
\urldef\tempurl%
\url{https://doi.org/10.1145/1095810.1095817}
\showDOI{\tempurl}


\bibitem[\protect\citeauthoryear{Abraham, Gueta, Malkhi, Alvisi, Kotla, and
  Martin}{Abraham et~al\mbox{.}}{2017}]%
        {AGMAKT17}
\bibfield{author}{\bibinfo{person}{Ittai Abraham}, \bibinfo{person}{Guy Gueta},
  \bibinfo{person}{Dahlia Malkhi}, \bibinfo{person}{Lorenzo Alvisi},
  \bibinfo{person}{Ramakrishna Kotla}, {and} \bibinfo{person}{Jean{-}Philippe
  Martin}.} \bibinfo{year}{2017}\natexlab{}.
\newblock \showarticletitle{Revisiting Fast Practical Byzantine Fault
  Tolerance}.
\newblock \bibinfo{journal}{\emph{CoRR}}  \bibinfo{volume}{abs/1712.01367}
  (\bibinfo{year}{2017}).
\newblock
\showeprint[arxiv]{1712.01367}
\urldef\tempurl%
\url{http://arxiv.org/abs/1712.01367}
\showURL{%
\tempurl}


\bibitem[\protect\citeauthoryear{Abraham, Malkhi, and Spiegelman}{Abraham
  et~al\mbox{.}}{2019}]%
        {AMS19}
\bibfield{author}{\bibinfo{person}{Ittai Abraham}, \bibinfo{person}{Dahlia
  Malkhi}, {and} \bibinfo{person}{Alexander Spiegelman}.}
  \bibinfo{year}{2019}\natexlab{}.
\newblock \showarticletitle{Asymptotically Optimal Validated Asynchronous
  Byzantine Agreement}. In \bibinfo{booktitle}{\emph{Proceedings of the 2019
  {ACM} Symposium on Principles of Distributed Computing, {PODC} 2019, Toronto,
  ON, Canada, July 29 - August 2, 2019.}} \bibinfo{pages}{337--346}.
\newblock
\urldef\tempurl%
\url{https://doi.org/10.1145/3293611.3331612}
\showDOI{\tempurl}


\bibitem[\protect\citeauthoryear{Baird}{Baird}{2016}]%
        {Baird16}
\bibfield{author}{\bibinfo{person}{Leemon Baird}.}
  \bibinfo{year}{2016}\natexlab{}.
\newblock \showarticletitle{The swirlds hashgraph consensus algorithm: Fair,
  fast, byzantine fault tolerance}.
\newblock \bibinfo{journal}{\emph{Swirlds Tech Reports SWIRLDS-TR-2016-01,
  Tech. Rep.}} (\bibinfo{year}{2016}).
\newblock


\bibitem[\protect\citeauthoryear{Ben{-}Or and El{-}Yaniv}{Ben{-}Or and
  El{-}Yaniv}{2003}]%
        {BE03}
\bibfield{author}{\bibinfo{person}{Michael Ben{-}Or} {and} \bibinfo{person}{Ran
  El{-}Yaniv}.} \bibinfo{year}{2003}\natexlab{}.
\newblock \showarticletitle{Resilient-optimal interactive consistency in
  constant time}.
\newblock \bibinfo{journal}{\emph{Distributed Computing}} \bibinfo{volume}{16},
  \bibinfo{number}{4} (\bibinfo{year}{2003}), \bibinfo{pages}{249--262}.
\newblock
\urldef\tempurl%
\url{https://doi.org/10.1007/s00446-002-0083-3}
\showDOI{\tempurl}


\bibitem[\protect\citeauthoryear{Boldyreva}{Boldyreva}{2003}]%
        {Boldyreva03}
\bibfield{author}{\bibinfo{person}{Alexandra Boldyreva}.}
  \bibinfo{year}{2003}\natexlab{}.
\newblock \showarticletitle{Threshold Signatures, Multisignatures and Blind
  Signatures Based on the Gap-Diffie-Hellman-Group Signature Scheme}. In
  \bibinfo{booktitle}{\emph{Public Key Cryptography - {PKC} 2003, 6th
  International Workshop on Theory and Practice in Public Key Cryptography,
  Miami, FL, USA, January 6-8, 2003, Proceedings}}. \bibinfo{pages}{31--46}.
\newblock
\urldef\tempurl%
\url{https://doi.org/10.1007/3-540-36288-6\_3}
\showDOI{\tempurl}


\bibitem[\protect\citeauthoryear{Boneh, Bonneau, B{\"{u}}nz, and Fisch}{Boneh
  et~al\mbox{.}}{2018}]%
        {BBBF18}
\bibfield{author}{\bibinfo{person}{Dan Boneh}, \bibinfo{person}{Joseph
  Bonneau}, \bibinfo{person}{Benedikt B{\"{u}}nz}, {and} \bibinfo{person}{Ben
  Fisch}.} \bibinfo{year}{2018}\natexlab{}.
\newblock \showarticletitle{Verifiable Delay Functions}. In
  \bibinfo{booktitle}{\emph{Advances in Cryptology - {CRYPTO} 2018 - 38th
  Annual International Cryptology Conference, Santa Barbara, CA, USA, August
  19-23, 2018, Proceedings, Part {I}}}. \bibinfo{pages}{757--788}.
\newblock
\urldef\tempurl%
\url{https://doi.org/10.1007/978-3-319-96884-1\_25}
\showDOI{\tempurl}


\bibitem[\protect\citeauthoryear{Boneh, Lynn, and Shacham}{Boneh
  et~al\mbox{.}}{2004}]%
        {BLS04}
\bibfield{author}{\bibinfo{person}{Dan Boneh}, \bibinfo{person}{Ben Lynn},
  {and} \bibinfo{person}{Hovav Shacham}.} \bibinfo{year}{2004}\natexlab{}.
\newblock \showarticletitle{Short Signatures from the Weil Pairing}.
\newblock \bibinfo{journal}{\emph{J. Cryptology}} \bibinfo{volume}{17},
  \bibinfo{number}{4} (\bibinfo{year}{2004}), \bibinfo{pages}{297--319}.
\newblock
\urldef\tempurl%
\url{https://doi.org/10.1007/s00145-004-0314-9}
\showDOI{\tempurl}


\bibitem[\protect\citeauthoryear{Bracha}{Bracha}{1987}]%
        {Bracha87}
\bibfield{author}{\bibinfo{person}{Gabriel Bracha}.}
  \bibinfo{year}{1987}\natexlab{}.
\newblock \showarticletitle{Asynchronous Byzantine Agreement Protocols}.
\newblock \bibinfo{journal}{\emph{Inf. Comput.}} \bibinfo{volume}{75},
  \bibinfo{number}{2} (\bibinfo{year}{1987}), \bibinfo{pages}{130--143}.
\newblock
\urldef\tempurl%
\url{https://doi.org/10.1016/0890-5401(87)90054-X}
\showDOI{\tempurl}


\bibitem[\protect\citeauthoryear{Bracha and Toueg}{Bracha and Toueg}{1983}]%
        {BT83}
\bibfield{author}{\bibinfo{person}{Gabriel Bracha} {and} \bibinfo{person}{Sam
  Toueg}.} \bibinfo{year}{1983}\natexlab{}.
\newblock \showarticletitle{Resilient Consensus Protocols}. In
  \bibinfo{booktitle}{\emph{Proceedings of the Second Annual {ACM}
  {SIGACT-SIGOPS} Symposium on Principles of Distributed Computing, Montreal,
  Quebec, Canada, August 17-19, 1983}}. \bibinfo{pages}{12--26}.
\newblock
\urldef\tempurl%
\url{https://doi.org/10.1145/800221.806706}
\showDOI{\tempurl}


\bibitem[\protect\citeauthoryear{Buchman, Kwon, and Milosevic}{Buchman
  et~al\mbox{.}}{2018}]%
        {BKM18}
\bibfield{author}{\bibinfo{person}{Ethan Buchman}, \bibinfo{person}{Jae Kwon},
  {and} \bibinfo{person}{Zarko Milosevic}.} \bibinfo{year}{2018}\natexlab{}.
\newblock \showarticletitle{The latest gossip on {BFT} consensus}.
\newblock \bibinfo{journal}{\emph{CoRR}}  \bibinfo{volume}{abs/1807.04938}
  (\bibinfo{year}{2018}).
\newblock
\showeprint[arxiv]{1807.04938}
\urldef\tempurl%
\url{http://arxiv.org/abs/1807.04938}
\showURL{%
\tempurl}


\bibitem[\protect\citeauthoryear{Buterin and Griffith}{Buterin and
  Griffith}{2017}]%
        {BG17}
\bibfield{author}{\bibinfo{person}{Vitalik Buterin} {and}
  \bibinfo{person}{Virgil Griffith}.} \bibinfo{year}{2017}\natexlab{}.
\newblock \showarticletitle{Casper the Friendly Finality Gadget}.
\newblock \bibinfo{journal}{\emph{CoRR}}  \bibinfo{volume}{abs/1710.09437}
  (\bibinfo{year}{2017}).
\newblock
\showeprint[arxiv]{1710.09437}
\urldef\tempurl%
\url{http://arxiv.org/abs/1710.09437}
\showURL{%
\tempurl}


\bibitem[\protect\citeauthoryear{Cachin, Guerraoui, and Rodrigues}{Cachin
  et~al\mbox{.}}{2011}]%
        {CGR11}
\bibfield{author}{\bibinfo{person}{Christian Cachin}, \bibinfo{person}{Rachid
  Guerraoui}, {and} \bibinfo{person}{Lu{\'{\i}}s E.~T. Rodrigues}.}
  \bibinfo{year}{2011}\natexlab{}.
\newblock \bibinfo{booktitle}{\emph{Introduction to Reliable and Secure
  Distributed Programming {(2.} ed.)}}.
\newblock \bibinfo{publisher}{Springer}.
\newblock
\showISBNx{978-3-642-15259-7}
\urldef\tempurl%
\url{https://doi.org/10.1007/978-3-642-15260-3}
\showDOI{\tempurl}


\bibitem[\protect\citeauthoryear{Cachin, Kursawe, Lysyanskaya, and
  Strobl}{Cachin et~al\mbox{.}}{2002}]%
        {CKLS02}
\bibfield{author}{\bibinfo{person}{Christian Cachin}, \bibinfo{person}{Klaus
  Kursawe}, \bibinfo{person}{Anna Lysyanskaya}, {and} \bibinfo{person}{Reto
  Strobl}.} \bibinfo{year}{2002}\natexlab{}.
\newblock \showarticletitle{Asynchronous verifiable secret sharing and
  proactive cryptosystems}. In \bibinfo{booktitle}{\emph{Proceedings of the 9th
  {ACM} Conference on Computer and Communications Security, {CCS} 2002,
  Washington, DC, USA, November 18-22, 2002}}. \bibinfo{pages}{88--97}.
\newblock
\urldef\tempurl%
\url{https://doi.org/10.1145/586110.586124}
\showDOI{\tempurl}


\bibitem[\protect\citeauthoryear{Cachin, Kursawe, Petzold, and Shoup}{Cachin
  et~al\mbox{.}}{2001}]%
        {CKPS01}
\bibfield{author}{\bibinfo{person}{Christian Cachin}, \bibinfo{person}{Klaus
  Kursawe}, \bibinfo{person}{Frank Petzold}, {and} \bibinfo{person}{Victor
  Shoup}.} \bibinfo{year}{2001}\natexlab{}.
\newblock \showarticletitle{Secure and Efficient Asynchronous Broadcast
  Protocols}. In \bibinfo{booktitle}{\emph{Advances in Cryptology - {CRYPTO}
  2001, 21st Annual International Cryptology Conference, Santa Barbara,
  California, USA, August 19-23, 2001, Proceedings}}.
  \bibinfo{pages}{524--541}.
\newblock
\urldef\tempurl%
\url{https://doi.org/10.1007/3-540-44647-8\_31}
\showDOI{\tempurl}


\bibitem[\protect\citeauthoryear{Cachin, Kursawe, and Shoup}{Cachin
  et~al\mbox{.}}{2005}]%
        {CKS05}
\bibfield{author}{\bibinfo{person}{Christian Cachin}, \bibinfo{person}{Klaus
  Kursawe}, {and} \bibinfo{person}{Victor Shoup}.}
  \bibinfo{year}{2005}\natexlab{}.
\newblock \showarticletitle{Random Oracles in Constantinople: Practical
  Asynchronous Byzantine Agreement Using Cryptography}.
\newblock \bibinfo{journal}{\emph{J. Cryptology}} \bibinfo{volume}{18},
  \bibinfo{number}{3} (\bibinfo{year}{2005}), \bibinfo{pages}{219--246}.
\newblock
\urldef\tempurl%
\url{https://doi.org/10.1007/s00145-005-0318-0}
\showDOI{\tempurl}


\bibitem[\protect\citeauthoryear{Cachin and Tessaro}{Cachin and
  Tessaro}{2005}]%
        {CT05}
\bibfield{author}{\bibinfo{person}{Christian Cachin} {and}
  \bibinfo{person}{Stefano Tessaro}.} \bibinfo{year}{2005}\natexlab{}.
\newblock \showarticletitle{Asynchronous verifiable information dispersal}. In
  \bibinfo{booktitle}{\emph{24th IEEE Symposium on Reliable Distributed Systems
  (SRDS'05)}}. IEEE, \bibinfo{pages}{191--201}.
\newblock


\bibitem[\protect\citeauthoryear{Cachin and Vukolic}{Cachin and
  Vukolic}{2017}]%
        {CV17}
\bibfield{author}{\bibinfo{person}{Christian Cachin} {and}
  \bibinfo{person}{Marko Vukolic}.} \bibinfo{year}{2017}\natexlab{}.
\newblock \showarticletitle{Blockchain Consensus Protocols in the Wild}.
\newblock \bibinfo{journal}{\emph{CoRR}}  \bibinfo{volume}{abs/1707.01873}
  (\bibinfo{year}{2017}).
\newblock
\showeprint[arxiv]{1707.01873}
\urldef\tempurl%
\url{http://arxiv.org/abs/1707.01873}
\showURL{%
\tempurl}


\bibitem[\protect\citeauthoryear{Canetti and Rabin}{Canetti and Rabin}{1993}]%
        {CR93}
\bibfield{author}{\bibinfo{person}{Ran Canetti} {and} \bibinfo{person}{Tal
  Rabin}.} \bibinfo{year}{1993}\natexlab{}.
\newblock \showarticletitle{Fast asynchronous Byzantine agreement with optimal
  resilience}. In \bibinfo{booktitle}{\emph{Proceedings of the Twenty-Fifth
  Annual {ACM} Symposium on Theory of Computing, May 16-18, 1993, San Diego,
  CA, {USA}}}. \bibinfo{pages}{42--51}.
\newblock
\urldef\tempurl%
\url{https://doi.org/10.1145/167088.167105}
\showDOI{\tempurl}


\bibitem[\protect\citeauthoryear{Cascudo and David}{Cascudo and David}{2017}]%
        {CD17}
\bibfield{author}{\bibinfo{person}{Ignacio Cascudo} {and}
  \bibinfo{person}{Bernardo David}.} \bibinfo{year}{2017}\natexlab{}.
\newblock \showarticletitle{{SCRAPE:} Scalable Randomness Attested by Public
  Entities}. In \bibinfo{booktitle}{\emph{Applied Cryptography and Network
  Security - 15th International Conference, {ACNS} 2017, Kanazawa, Japan, July
  10-12, 2017, Proceedings}}. \bibinfo{pages}{537--556}.
\newblock
\urldef\tempurl%
\url{https://doi.org/10.1007/978-3-319-61204-1\_27}
\showDOI{\tempurl}


\bibitem[\protect\citeauthoryear{Castro and Liskov}{Castro and Liskov}{1999}]%
        {CL99}
\bibfield{author}{\bibinfo{person}{Miguel Castro} {and}
  \bibinfo{person}{Barbara Liskov}.} \bibinfo{year}{1999}\natexlab{}.
\newblock \showarticletitle{Practical Byzantine Fault Tolerance}. In
  \bibinfo{booktitle}{\emph{Proceedings of the Third {USENIX} Symposium on
  Operating Systems Design and Implementation (OSDI), New Orleans, Louisiana,
  USA, February 22-25, 1999}}. \bibinfo{pages}{173--186}.
\newblock
\urldef\tempurl%
\url{https://dl.acm.org/citation.cfm?id=296824}
\showURL{%
\tempurl}


\bibitem[\protect\citeauthoryear{Correia, Neves, and Ver{\'{\i}}ssimo}{Correia
  et~al\mbox{.}}{2006}]%
        {CNV06}
\bibfield{author}{\bibinfo{person}{Miguel Correia},
  \bibinfo{person}{Nuno~Ferreira Neves}, {and} \bibinfo{person}{Paulo
  Ver{\'{\i}}ssimo}.} \bibinfo{year}{2006}\natexlab{}.
\newblock \showarticletitle{From Consensus to Atomic Broadcast: Time-Free
  Byzantine-Resistant Protocols without Signatures}.
\newblock \bibinfo{journal}{\emph{Comput. J.}} \bibinfo{volume}{49},
  \bibinfo{number}{1} (\bibinfo{year}{2006}), \bibinfo{pages}{82--96}.
\newblock
\urldef\tempurl%
\url{https://doi.org/10.1093/comjnl/bxh145}
\showDOI{\tempurl}


\bibitem[\protect\citeauthoryear{Danezis and Hrycyszyn}{Danezis and
  Hrycyszyn}{2018}]%
        {DH18}
\bibfield{author}{\bibinfo{person}{George Danezis} {and} \bibinfo{person}{David
  Hrycyszyn}.} \bibinfo{year}{2018}\natexlab{}.
\newblock \showarticletitle{Blockmania: from Block DAGs to Consensus}.
\newblock \bibinfo{journal}{\emph{arXiv preprint arXiv:1809.01620}}
  (\bibinfo{year}{2018}).
\newblock


\bibitem[\protect\citeauthoryear{Fischer, Lynch, and Paterson}{Fischer
  et~al\mbox{.}}{1985}]%
        {FLP85}
\bibfield{author}{\bibinfo{person}{Michael~J. Fischer},
  \bibinfo{person}{Nancy~A. Lynch}, {and} \bibinfo{person}{Mike Paterson}.}
  \bibinfo{year}{1985}\natexlab{}.
\newblock \showarticletitle{Impossibility of Distributed Consensus with One
  Faulty Process}.
\newblock \bibinfo{journal}{\emph{J. {ACM}}} \bibinfo{volume}{32},
  \bibinfo{number}{2} (\bibinfo{year}{1985}), \bibinfo{pages}{374--382}.
\newblock
\urldef\tempurl%
\url{https://doi.org/10.1145/3149.214121}
\showDOI{\tempurl}


\bibitem[\protect\citeauthoryear{Gennaro, Jarecki, Krawczyk, and Rabin}{Gennaro
  et~al\mbox{.}}{2003}]%
        {GJKR03}
\bibfield{author}{\bibinfo{person}{Rosario Gennaro}, \bibinfo{person}{Stanislaw
  Jarecki}, \bibinfo{person}{Hugo Krawczyk}, {and} \bibinfo{person}{Tal
  Rabin}.} \bibinfo{year}{2003}\natexlab{}.
\newblock \showarticletitle{Secure Applications of Pedersen's Distributed Key
  Generation Protocol}. In \bibinfo{booktitle}{\emph{Topics in Cryptology -
  {CT-RSA} 2003, The Cryptographers' Track at the {RSA} Conference 2003, San
  Francisco, CA, USA, April 13-17, 2003, Proceedings}}.
  \bibinfo{pages}{373--390}.
\newblock
\urldef\tempurl%
\url{https://doi.org/10.1007/3-540-36563-X\_26}
\showDOI{\tempurl}


\bibitem[\protect\citeauthoryear{Gennaro, Jarecki, Krawczyk, and Rabin}{Gennaro
  et~al\mbox{.}}{2007}]%
        {GJKR07}
\bibfield{author}{\bibinfo{person}{Rosario Gennaro}, \bibinfo{person}{Stanislaw
  Jarecki}, \bibinfo{person}{Hugo Krawczyk}, {and} \bibinfo{person}{Tal
  Rabin}.} \bibinfo{year}{2007}\natexlab{}.
\newblock \showarticletitle{Secure Distributed Key Generation for Discrete-Log
  Based Cryptosystems}.
\newblock \bibinfo{journal}{\emph{J. Cryptology}} \bibinfo{volume}{20},
  \bibinfo{number}{1} (\bibinfo{year}{2007}), \bibinfo{pages}{51--83}.
\newblock
\urldef\tempurl%
\url{https://doi.org/10.1007/s00145-006-0347-3}
\showDOI{\tempurl}


\bibitem[\protect\citeauthoryear{Hadzilacos and Toueg}{Hadzilacos and
  Toueg}{1994}]%
        {HT94}
\bibfield{author}{\bibinfo{person}{Vassos Hadzilacos} {and}
  \bibinfo{person}{Sam Toueg}.} \bibinfo{year}{1994}\natexlab{}.
\newblock \bibinfo{booktitle}{\emph{A modular approach to fault-tolerant
  broadcasts and related problems}}.
\newblock \bibinfo{type}{{T}echnical {R}eport}. \bibinfo{institution}{Cornell
  University}.
\newblock


\bibitem[\protect\citeauthoryear{Hanke, Movahedi, and Williams}{Hanke
  et~al\mbox{.}}{2018}]%
        {HMW18}
\bibfield{author}{\bibinfo{person}{Timo Hanke}, \bibinfo{person}{Mahnush
  Movahedi}, {and} \bibinfo{person}{Dominic Williams}.}
  \bibinfo{year}{2018}\natexlab{}.
\newblock \showarticletitle{Dfinity technology overview series, consensus
  system}.
\newblock \bibinfo{journal}{\emph{arXiv preprint arXiv:1805.04548}}
  (\bibinfo{year}{2018}).
\newblock


\bibitem[\protect\citeauthoryear{Kate, Huang, and Goldberg}{Kate
  et~al\mbox{.}}{2012}]%
        {KHG12}
\bibfield{author}{\bibinfo{person}{Aniket Kate}, \bibinfo{person}{Yizhou
  Huang}, {and} \bibinfo{person}{Ian Goldberg}.}
  \bibinfo{year}{2012}\natexlab{}.
\newblock \showarticletitle{Distributed Key Generation in the Wild}.
\newblock \bibinfo{journal}{\emph{{IACR} Cryptology ePrint Archive}}
  \bibinfo{volume}{2012} (\bibinfo{year}{2012}), \bibinfo{pages}{377}.
\newblock
\urldef\tempurl%
\url{http://eprint.iacr.org/2012/377}
\showURL{%
\tempurl}


\bibitem[\protect\citeauthoryear{Kiayias, Russell, David, and
  Oliynykov}{Kiayias et~al\mbox{.}}{2017}]%
        {KRDO17}
\bibfield{author}{\bibinfo{person}{Aggelos Kiayias}, \bibinfo{person}{Alexander
  Russell}, \bibinfo{person}{Bernardo David}, {and} \bibinfo{person}{Roman
  Oliynykov}.} \bibinfo{year}{2017}\natexlab{}.
\newblock \showarticletitle{Ouroboros: {A} Provably Secure Proof-of-Stake
  Blockchain Protocol}. In \bibinfo{booktitle}{\emph{Advances in Cryptology -
  {CRYPTO} 2017 - 37th Annual International Cryptology Conference, Santa
  Barbara, CA, USA, August 20-24, 2017, Proceedings, Part {I}}}.
  \bibinfo{pages}{357--388}.
\newblock
\urldef\tempurl%
\url{https://doi.org/10.1007/978-3-319-63688-7\_12}
\showDOI{\tempurl}


\bibitem[\protect\citeauthoryear{Kotla, Alvisi, Dahlin, Clement, and
  Wong}{Kotla et~al\mbox{.}}{2009}]%
        {KADCW09}
\bibfield{author}{\bibinfo{person}{Ramakrishna Kotla}, \bibinfo{person}{Lorenzo
  Alvisi}, \bibinfo{person}{Michael Dahlin}, \bibinfo{person}{Allen Clement},
  {and} \bibinfo{person}{Edmund~L. Wong}.} \bibinfo{year}{2009}\natexlab{}.
\newblock \showarticletitle{Zyzzyva: Speculative Byzantine fault tolerance}.
\newblock \bibinfo{journal}{\emph{{ACM} Trans. Comput. Syst.}}
  \bibinfo{volume}{27}, \bibinfo{number}{4} (\bibinfo{year}{2009}),
  \bibinfo{pages}{7:1--7:39}.
\newblock
\urldef\tempurl%
\url{https://doi.org/10.1145/1658357.1658358}
\showDOI{\tempurl}


\bibitem[\protect\citeauthoryear{Kwon and Buchman}{Kwon and Buchman}{[n.d.]}]%
        {KB18}
\bibfield{author}{\bibinfo{person}{Jae Kwon} {and} \bibinfo{person}{Ethan
  Buchman}.} \bibinfo{year}{[n.d.]}\natexlab{}.
\newblock \showarticletitle{A Network of Distributed Ledgers}.
\newblock  (\bibinfo{year}{[n.\,d.]}).
\newblock
\urldef\tempurl%
\url{https://cosmos.network/cosmos-whitepaper.pdf}
\showURL{%
\tempurl}


\bibitem[\protect\citeauthoryear{Lamport}{Lamport}{1978}]%
        {Lamport78}
\bibfield{author}{\bibinfo{person}{Leslie Lamport}.}
  \bibinfo{year}{1978}\natexlab{}.
\newblock \showarticletitle{Time, Clocks, and the Ordering of Events in a
  Distributed System}.
\newblock \bibinfo{journal}{\emph{Commun. {ACM}}} \bibinfo{volume}{21},
  \bibinfo{number}{7} (\bibinfo{year}{1978}), \bibinfo{pages}{558--565}.
\newblock
\urldef\tempurl%
\url{https://doi.org/10.1145/359545.359563}
\showDOI{\tempurl}


\bibitem[\protect\citeauthoryear{Lenstra and Wesolowski}{Lenstra and
  Wesolowski}{2017}]%
        {LW17}
\bibfield{author}{\bibinfo{person}{Arjen~K. Lenstra} {and}
  \bibinfo{person}{Benjamin Wesolowski}.} \bibinfo{year}{2017}\natexlab{}.
\newblock \showarticletitle{Trustworthy public randomness with sloth, unicorn,
  and trx}.
\newblock \bibinfo{journal}{\emph{{IJACT}}} \bibinfo{volume}{3},
  \bibinfo{number}{4} (\bibinfo{year}{2017}), \bibinfo{pages}{330--343}.
\newblock
\urldef\tempurl%
\url{https://doi.org/10.1504/IJACT.2017.10010315}
\showDOI{\tempurl}


\bibitem[\protect\citeauthoryear{Miller, Xia, Croman, Shi, and Song}{Miller
  et~al\mbox{.}}{2016}]%
        {MXCSS16}
\bibfield{author}{\bibinfo{person}{Andrew Miller}, \bibinfo{person}{Yu Xia},
  \bibinfo{person}{Kyle Croman}, \bibinfo{person}{Elaine Shi}, {and}
  \bibinfo{person}{Dawn Song}.} \bibinfo{year}{2016}\natexlab{}.
\newblock \showarticletitle{The Honey Badger of {BFT} Protocols}. In
  \bibinfo{booktitle}{\emph{Proceedings of the 2016 {ACM} {SIGSAC} Conference
  on Computer and Communications Security, Vienna, Austria, October 24-28,
  2016}}. \bibinfo{pages}{31--42}.
\newblock
\urldef\tempurl%
\url{https://doi.org/10.1145/2976749.2978399}
\showDOI{\tempurl}


\bibitem[\protect\citeauthoryear{Milosevic, Hutle, and Schiper}{Milosevic
  et~al\mbox{.}}{2011}]%
        {MHS11}
\bibfield{author}{\bibinfo{person}{Zarko Milosevic}, \bibinfo{person}{Martin
  Hutle}, {and} \bibinfo{person}{Andr{\'{e}} Schiper}.}
  \bibinfo{year}{2011}\natexlab{}.
\newblock \showarticletitle{On the Reduction of Atomic Broadcast to Consensus
  with Byzantine Faults}. In \bibinfo{booktitle}{\emph{30th {IEEE} Symposium on
  Reliable Distributed Systems {(SRDS} 2011), Madrid, Spain, October 4-7,
  2011}}. \bibinfo{pages}{235--244}.
\newblock
\urldef\tempurl%
\url{https://doi.org/10.1109/SRDS.2011.36}
\showDOI{\tempurl}


\bibitem[\protect\citeauthoryear{Moser and Melliar{-}Smith}{Moser and
  Melliar{-}Smith}{1999}]%
        {MM99}
\bibfield{author}{\bibinfo{person}{Louise~E. Moser} {and}
  \bibinfo{person}{P.~M. Melliar{-}Smith}.} \bibinfo{year}{1999}\natexlab{}.
\newblock \showarticletitle{Byzantine-Resistant Total Ordering Algorithms}.
\newblock \bibinfo{journal}{\emph{Inf. Comput.}} \bibinfo{volume}{150},
  \bibinfo{number}{1} (\bibinfo{year}{1999}), \bibinfo{pages}{75--111}.
\newblock
\urldef\tempurl%
\url{https://doi.org/10.1006/inco.1998.2770}
\showDOI{\tempurl}


\bibitem[\protect\citeauthoryear{Most{\'{e}}faoui, Moumen, and
  Raynal}{Most{\'{e}}faoui et~al\mbox{.}}{2015}]%
        {MMR15}
\bibfield{author}{\bibinfo{person}{Achour Most{\'{e}}faoui},
  \bibinfo{person}{Hamouma Moumen}, {and} \bibinfo{person}{Michel Raynal}.}
  \bibinfo{year}{2015}\natexlab{}.
\newblock \showarticletitle{Signature-Free Asynchronous Binary Byzantine
  Consensus with t {\textless} n/3, O(n2) Messages, and {O(1)} Expected Time}.
\newblock \bibinfo{journal}{\emph{J. {ACM}}} \bibinfo{volume}{62},
  \bibinfo{number}{4} (\bibinfo{year}{2015}), \bibinfo{pages}{31:1--31:21}.
\newblock
\urldef\tempurl%
\url{https://doi.org/10.1145/2785953}
\showDOI{\tempurl}


\bibitem[\protect\citeauthoryear{Nakamoto}{Nakamoto}{2008}]%
        {Nakamoto08}
\bibfield{author}{\bibinfo{person}{Satoshi Nakamoto}.}
  \bibinfo{year}{2008}\natexlab{}.
\newblock \showarticletitle{Bitcoin: A peer-to-peer electronic cash system}.
\newblock  (\bibinfo{year}{2008}).
\newblock


\bibitem[\protect\citeauthoryear{Pass and Shi}{Pass and Shi}{2017}]%
        {PassS17}
\bibfield{author}{\bibinfo{person}{Rafael Pass} {and} \bibinfo{person}{Elaine
  Shi}.} \bibinfo{year}{2017}\natexlab{}.
\newblock \showarticletitle{Hybrid Consensus: Efficient Consensus in the
  Permissionless Model}. In \bibinfo{booktitle}{\emph{31st International
  Symposium on Distributed Computing, {DISC} 2017, October 16-20, 2017, Vienna,
  Austria}}. \bibinfo{pages}{39:1--39:16}.
\newblock
\urldef\tempurl%
\url{https://doi.org/10.4230/LIPIcs.DISC.2017.39}
\showDOI{\tempurl}


\bibitem[\protect\citeauthoryear{Pedersen}{Pedersen}{1991}]%
        {Pedersen91}
\bibfield{author}{\bibinfo{person}{Torben~P. Pedersen}.}
  \bibinfo{year}{1991}\natexlab{}.
\newblock \showarticletitle{A Threshold Cryptosystem without a Trusted Party
  (Extended Abstract)}. In \bibinfo{booktitle}{\emph{Advances in Cryptology -
  {EUROCRYPT} '91, Workshop on the Theory and Application of of Cryptographic
  Techniques, Brighton, UK, April 8-11, 1991, Proceedings}}.
  \bibinfo{pages}{522--526}.
\newblock
\urldef\tempurl%
\url{https://doi.org/10.1007/3-540-46416-6\_47}
\showDOI{\tempurl}


\bibitem[\protect\citeauthoryear{Pietrzak}{Pietrzak}{2019}]%
        {Pietrzak19}
\bibfield{author}{\bibinfo{person}{Krzysztof Pietrzak}.}
  \bibinfo{year}{2019}\natexlab{}.
\newblock \showarticletitle{Simple Verifiable Delay Functions}. In
  \bibinfo{booktitle}{\emph{10th Innovations in Theoretical Computer Science
  Conference, {ITCS} 2019, January 10-12, 2019, San Diego, California, {USA}}}.
  \bibinfo{pages}{60:1--60:15}.
\newblock
\urldef\tempurl%
\url{https://doi.org/10.4230/LIPIcs.ITCS.2019.60}
\showDOI{\tempurl}


\bibitem[\protect\citeauthoryear{Pointcheval and Stern}{Pointcheval and
  Stern}{1996}]%
        {PS96}
\bibfield{author}{\bibinfo{person}{David Pointcheval} {and}
  \bibinfo{person}{Jacques Stern}.} \bibinfo{year}{1996}\natexlab{}.
\newblock \showarticletitle{Security Proofs for Signature Schemes}. In
  \bibinfo{booktitle}{\emph{Advances in Cryptology - {EUROCRYPT} '96,
  International Conference on the Theory and Application of Cryptographic
  Techniques, Saragossa, Spain, May 12-16, 1996, Proceeding}}.
  \bibinfo{pages}{387--398}.
\newblock
\urldef\tempurl%
\url{https://doi.org/10.1007/3-540-68339-9\_33}
\showDOI{\tempurl}


\bibitem[\protect\citeauthoryear{Popov}{Popov}{2016}]%
        {Popov16}
\bibfield{author}{\bibinfo{person}{Serguei Popov}.}
  \bibinfo{year}{2016}\natexlab{}.
\newblock \showarticletitle{The tangle}.
\newblock \bibinfo{journal}{\emph{cit. on}} (\bibinfo{year}{2016}),
  \bibinfo{pages}{131}.
\newblock


\bibitem[\protect\citeauthoryear{Shamir}{Shamir}{1979}]%
        {Shamir79}
\bibfield{author}{\bibinfo{person}{Adi Shamir}.}
  \bibinfo{year}{1979}\natexlab{}.
\newblock \showarticletitle{How to Share a Secret}.
\newblock \bibinfo{journal}{\emph{Commun. {ACM}}} \bibinfo{volume}{22},
  \bibinfo{number}{11} (\bibinfo{year}{1979}), \bibinfo{pages}{612--613}.
\newblock
\urldef\tempurl%
\url{https://doi.org/10.1145/359168.359176}
\showDOI{\tempurl}


\bibitem[\protect\citeauthoryear{Syta, Jovanovic, Kokoris{-}Kogias, Gailly,
  Gasser, Khoffi, Fischer, and Ford}{Syta et~al\mbox{.}}{2017}]%
        {SJKGGKFF17}
\bibfield{author}{\bibinfo{person}{Ewa Syta}, \bibinfo{person}{Philipp
  Jovanovic}, \bibinfo{person}{Eleftherios Kokoris{-}Kogias},
  \bibinfo{person}{Nicolas Gailly}, \bibinfo{person}{Linus Gasser},
  \bibinfo{person}{Ismail Khoffi}, \bibinfo{person}{Michael~J. Fischer}, {and}
  \bibinfo{person}{Bryan Ford}.} \bibinfo{year}{2017}\natexlab{}.
\newblock \showarticletitle{Scalable Bias-Resistant Distributed Randomness}. In
  \bibinfo{booktitle}{\emph{2017 {IEEE} Symposium on Security and Privacy, {SP}
  2017, San Jose, CA, USA, May 22-26, 2017}}. \bibinfo{pages}{444--460}.
\newblock
\urldef\tempurl%
\url{https://doi.org/10.1109/SP.2017.45}
\showDOI{\tempurl}


\bibitem[\protect\citeauthoryear{Wesolowski}{Wesolowski}{2019}]%
        {Wesolowski19}
\bibfield{author}{\bibinfo{person}{Benjamin Wesolowski}.}
  \bibinfo{year}{2019}\natexlab{}.
\newblock \showarticletitle{Efficient Verifiable Delay Functions}. In
  \bibinfo{booktitle}{\emph{Advances in Cryptology - {EUROCRYPT} 2019 - 38th
  Annual International Conference on the Theory and Applications of
  Cryptographic Techniques, Darmstadt, Germany, May 19-23, 2019, Proceedings,
  Part {III}}}. \bibinfo{pages}{379--407}.
\newblock
\urldef\tempurl%
\url{https://doi.org/10.1007/978-3-030-17659-4\_13}
\showDOI{\tempurl}


\end{thebibliography}

\appendix

\section{Practical Considerations}\label{sec:practical}

The protocol that we described in the main body of the paper has excellent theoretical properties and achieves optimal asymptotic guarantees, however in the original form might not be viable for practical implementation.
The high level reason for that is that it was designed to operate in the harshest possible adversarial setting (i.e. the adversary controlling $f$ out of $3f+1$ nodes and being able to arbitrarily delay messages) and it was not optimized for the "optimistic case".
This means intuitively that as of now, the protocol can withstand arbitrarily oppressive "attacks" of the adversary, but does not get any faster when no attacks take place.

In this section we present several adjustments to the protocol that allow us to make it significantly faster in the optimistic case, without compromising its strong security guarantees.
As a consequence, we obtain a version of the protocol that is well suited for practical implementation -- we show in particular (see Section~\ref{subsec:synchronous}) that under partial synchrony it matches the optimistic $3$-round validation delay of PBFT and Tendermint~\cite{CL99,BKM18}.

Below, we list the main sources of practical inefficiency of the protocol:
\begin{enumerate}
\item The reliance on RBC: each \dagr{} effectively takes $3$ \asynchr{}s, so in practice the protocol requires more rounds than simple synchronous or partially synchronous protocols such as PBFT. Furthermore, performing $\Omega(N)$ instances of RBC simultaneusly by a node, forces it to send $\Omega(N^2)$ distinct messages per round, which might not be feasible.
\item The worst-case assumption that the adversary can pretty much arbitrarily manipulate the structure of the $\chdag{}$ forces the randomness to be revealed with large delay and causes inefficiency.
\item The total size of metadata (total size of units, ignoring transactions) produced in one round by all the nodes is quite large: $\Omega(N^2 \lambda)$ bits, because each round-$r$ unit contains hashes of $\Omega(N)$ units from round $r-1$.
\end{enumerate}

We show how to solve these three issues in Sections~\ref{sec:rbc_gossip}, \ref{sec:new_voting}
 and \ref{subsec:size-of-units} respectively.

\subsection{From RBC to Multicast and Random Gossip}\label{sec:rbc_gossip}

As a practical solution, we propose to use a combination of multicast and random gossip in place of RBC to share the \chdag{} between nodes.
More specifically, the following \msf{Quick}-\msf{DAG}-\msf{Grow} algorithm is meant to replace the \msf{DAG}-\msf{Grow} from the main body (based on RBC).

\renewcommand{\algorithmcfname}{\msf{Quick}-\msf{DAG}-\msf{Grow}$(\cD)$}
\begin{algorithm-hbox}
  \SetKwFunction{FMain}{$\msf{CreateUnit}(\msf{data})$}
  \SetKwProg{Fn}{ }{:}{}
  \Fn{\FMain}
    {
    \For{$r=0,1,2, \ldots$}
        {
        \If{r>0}
        {
            {\bf wait until} $|\{U\in \cD:\H(U)=r-1\}|\geq 2f+1$
        }
        $P \leftarrow \{$maximal $\cP_i$'s unit of round $<r$ in $\cD:\cP_i\in\cP\}$\\
        {\bf create} a new unit $U$ with $P$ as parents\\
        {\bf include} \msf{data} in $U$\\
        {\bf add} $U$ to $\cD$\\
        {\bf multicast} $U$
        }
    }
\SetKwFor{Loop}{loop forever}{}{}
  \SetKwFunction{FMain}{$\msf{ReceiveUnits}$}
  \SetKwProg{Fn}{ }{:}{}
  \Fn{\FMain}
    {
    \SetKwProg{Upon}{upon}{ do}{end}
    \Loop{}
    {
    \Upon{receiving a unit $U$}
        {
        \lIf{$U$ is correct}
        {{\bf add} $U$ to buffer $\mathcal{B}$}
        
        \While{exists $V\in \mathcal{B}$ whose all parents are in $\cD$}
            {

                {\bf move} $V$ from $\mathcal{B}$ to $\cD$
            }
        }
    }
    }
\SetKwFor{Loop}{loop forever}{}{}
  \SetKwFunction{FMain}{$\msf{Gossip}$}
  \SetKwProg{Fn}{ }{:}{}
  \Fn{\FMain}
    {
    \tcc{Run by node $i$}
    \SetKwProg{Upon}{upon}{ do}{end}
    \Loop{}
        {
        $j\leftarrow$ randomly sampled node\\
        {\bf send} $j$ concise info about $\cD_i$\\
        {\bf receive} all units in $\cD_j\setminus\cD_i$\\
        {\bf send} all units in $\cD_i\setminus\cD_j$\\
        }
    }
  \caption{ }

\end{algorithm-hbox}

At this point we emphasize that while we believe that our proposal yields a truly efficient algorithm, it is rather cumbersome to formally reason about the communication complexity of such a solution (specifically about the gossip part) and we do not attempt to provide a formal treatment.
Instead, in a practical implementation one needs to include a set of rules that make it impossible for malicious nodes to request the same data over and over again (in order to slow down honest nodes).
To analyze further effects of switching from RBC to multicast + gossip, recall that RBC in the original protocol allows us to solve the following two problems
\begin{enumerate}
    \item the {\bf data availability} problem, as a given piece of data is locally output by an honest node in RBC only when it is guaranteed to be eventually output by every other honest node,
    \item the problem of {\bf forks}, since only one version of each unit may be locally output by an honest node.
\end{enumerate} 

\noindent 
The data availability problem, we discuss in Section~\ref{subsec:data-availability}.
Regarding forks: their existence in the \chdag{} does not pose a threat to the theoretical properties of our consensus protocol.
Indeed, in Section~\ref{subsec:analysis_quickalef} we show that the protocol is guaranteed to reach consensus even in the presence of forks in \chdag{}.
The only (yet serious) problem with forks is that if malicious nodes produce a lot of them and all these forks need to be processed by honest nodes, this can cause crashes due to lack of resources, i.e., RAM or disc space.
To counter this issue we add an auxiliary mechanism to the protocol whose purpose is bounding the number of possible forks (see Section~\ref{subsec:bounding-number-of-forks}).
Additionally, we show that without a very specific set of precautions, such an attack is rather simple to conduct and, to best of our knowledge, affects most of the currently proposed DAG-based protocols. 
We formalize it as a \emph{Fork Bomb} attack in section \ref{sec:fork-bomb}.

\subsubsection{Ensuring data availability via gossip}\label{subsec:data-availability}
Suppose for a moment that we have multicast as the only mechanism for sending units to other nodes (i.e., the creator multicasts his newly created unit to all the nodes).
While this is clearly the fastest way to disseminate information through a network, it is certainly prone to adversarial attacks on data availability.
This is a consequence of the fact that only the creator of this unit would then share it with other nodes.
Hence, if the creator is malicious, it may introduce intentional inconsistencies in local copies of \chdag{}s stored by honest nodes by sending newly created units only to some part of the network or even sending different variants to different nodes.
To see that such an attack can not be trivially countered, let us explore possible scenarios, depending on the particular rules of unit creation:
\begin{itemize}
    \item If node $i$ {\bf is allowed} to choose a unit $U$ as a parent for its newly created unit $V$ despite the fact that it does not know the whole ``context`` of $U$ (i.e., it does not have all of $U$'s parents in its local copy of \chdag{}), then we risk that the functions $\msf{Vote}(V, \cdot)$ and $\msf{UnitDecide}(V, \cdot)$ will never succeed to terminate with a non-$\bot$ verdict. Hence, node $i$ may never be able to make a decision about $V$.
    \item If on the other hand $i$ {\bf is not allowed} to choose $U$ as a parent in such a case, then the growth of \chdag{} may stall. Indeed, another honest node might have created $U$ with a parent $W$ that was available to him but not $i$ (because $W$ was created by a malicious node).
\end{itemize}

Hence it seems that to counter such intentional or unintentional (resulting for example from temporary network failures) inconsistencies, there has to be a mechanism in place that allows nodes to exchange information about units they did not produce. 
To this end in the \msf{Quick}-\msf{DAG}-\msf{Grow} protocol each node regularly reconciliates its local version of the \chdag{} with randomly chosen peers, in a gossip-like fashion.
Since each two honest nodes sync with each other infinitely often (with probability $1$), this guarantees data availability.

\subsubsection{Bounding the number of forks}\label{subsec:bounding-number-of-forks}

We introduce a mechanism that bounds the number of possible forks to $N$ variants (or rather "branches") per node.

Note that if a (dishonest) node starts broadcasting a forked unit, this is quickly noticed by all honest nodes, and the forker can be banned after such an incident.
Note however that we cannot reliably ``rollback" the \chdag{} to a version before the fork has happened, as two different honest nodes might have already build their units on two different fork branches.
Still, since units are signed by their creators, the forker can be proven to be malicious and punished accordingly, for example by slashing its stake.
This means that the situation when a malicious node just creates a small number of variants of a unit is not really dangerous -- it will be detected quite soon and this node will be banned forever.

What might be potentially dangerous though is when a group of malicious nodes collaborate to build a multi-level "Fork Bomb" (Section~\ref{sec:fork-bomb}) composed of a huge number of units and the honest nodes are forced to download them all, which could lead to crashes.
This is the main motivation behind introducing the \emph{alert protocol}.

We note that sending a unit via RBC can be thought of as ``committing'' to a single variant of a unit. 
In the absence of RBC the simplest attack of an adversary would be to send a different variant of a unit to every node.
Since just after receiving this unit, honest nodes cannot be aware of all the different variants, each of them might legally use a different variant as its parent.
This means that there is no way to bound the number of different variants of one forked unit in the \chdag{} by less than $N$ and we would like to propose a mechanism that does not allow to create more (while without any additional rules there might be an exponential number of variants -- see Section~\ref{sec:fork-bomb}).

In the solution we propose every node must ``commit'' to at most one variant of every received unit, by analogy to what happens in RBC. The general rule is that honest nodes can accept only those units that someone committed to. By default (when no forks are yet detected), the commitment will be simply realized by creating a next round unit that has one particular variant as its parent. On the other hand, if a fork is observed, then every node will have to send its commitment using RBC, to prevent attacks akin to the fork bomb (Section~\ref{sec:fork-bomb}).

A node can work in one of two states: normal and alert. In the normal state, the node multicasts its own units, gossips with other nodes, and builds its local copy of the $\chdag{}$. 
Now assume that a node $k$ detects a fork, i.e., obtains two distinct variants $U_1$ and $U_2$ of the round-$r$ unit created by $i$.
Then node $k$ enters the alert mode.
It stops building the $\chdag{}$ and broadcasts a special message $\msf{Alert}(k, i)$ using RBC.
The message consists of three parts:
\begin{enumerate}
    \item Proof of node $i$ being malicious, i.e., the pair $(U_1, U_2)$.
    \item The hash and the round number of the unit of highest round created by node $i$ that is currently present in $\cD_k$ (possibly null).
    \item The alert $id$ (every node assigns numeric ids to its alerts, ie., $id=0,1,2, \ldots$ and never starts a new alert before finishing the previous one).
\end{enumerate}

\noindent
By broadcasting the hash, node $k$ commits to a single variant of a whole chain of units created by node $i$, up to the one of highest round, currently known by node $k$. Since up to this point node $k$ did not know that node $i$ was malicious, then indeed units created by node $i$ in $\cD_k$ form a chain. %

Implementing this part requires some additional care, since now units cannot be validated and added one by one. Also, the sender should start with proving that the highest unit actually exists (by showing the unit and the set of its parents), so a malicious node cannot send arbitrarily large chunks of units, that cannot be parsed individually.

At this point we also note that since every node performs its alerts sequentially (by assigning numbers to them), in the worst case there can be at most $O(N)$ alerts run in parallel.

We now provide details on how a node should behave after a fork is discovered. Assume that node $k$ is honest, and $i$ is malicious (forking).
\begin{enumerate}
    \item Node $k$ stops communicating with node $i$ immediately after obtaining a proof that $i$ is malicious.
    \item Immediately after obtaining a proof that $i$ is malicious, node $k$ enters the alert mode, and broadcasts $\msf{Alert}(k, i)$.
    \item Node $k$ can issue at most one $\msf{Alert}(k, \cdot)$ at once, and alerts must be numbered consecutively.
    \item Node $k$ cannot participate in $\msf{Alert}(j, \cdot)$ if it has not locally terminated all previous alerts issued by node $j$.
    \item Node $k$ exits alert mode immediately after $\msf{Alert}(k, i)$ ends locally for $k$, and stores locally the proof of finishing the alert, along with the proof that node $i$ is malicious.
    \item After exiting the alert mode, node $k$ can still accept units created by $i$, but only if they are below at least one unit that some other node committed to. %
\end{enumerate}

\noindent 
The \msf{Quick}-\msf{DAG}-\msf{Grow} protocol enriched by the alert mechanism implemented through the above rules can serve as a reliable way to grow a \chdag{} that provably contains at most $N$ forks per node.
More specifically, in Section~\ref{sec:quick_aleph_dag} we prove the following

\begin{theorem}\label{thm:quick-dag-grow-properties}
Protocol \msf{Quick}-\msf{DAG}-\msf{Grow} with alert system is \msf{reliable} and \msf{growing} (as in definition \ref{def:protocol-properties}).
Additionally, every unit can be forked at most $N$ times.
\end{theorem}

\subsection{Adjustments to the Consensus Mechanism}\label{sec:new_voting}
In \alef{} protocol the mechanism of choosing round heads is designed in a way such that no matter whether the adversary is involved or not, the head is chosen in an expected constant number of rounds.
While this is optimal from the theoretical perspective, in practice we care a lot about constants, and for this reason we describe a slight change to the consensus mechanism that allows us to achieve latency of $3$ rounds under favorable network conditions (which we expect to be the default) and still have worst-case constant latency.

The key changes are to make the permutation (along which the head is chosen) ``a little" deterministic and to change the pattern of initial deterministic common votes.

\noindent{\bf Permutation.} Recall that in order to choose the head for round $r$ in \alef{} first a random permutation over units at this round $\pi_r$ is commonly generated and then in the order determined by $\pi_r$ the first unit decided $1$ is picked as the head.
The issue is that in order for the randomness to actually ``trick" the adversary, and not allow him to make the latency high, the permutation can be revealed at earliest at round $r+4$ in case of \alef{}, and $r+5$ in case of \quickalef{}\footnote{The difference stems from the potential existence of forks in \chdag{} constructed in \quickalef{} and will become more clear after reading lemma~\ref{lemma:quick-choose-head}}.
In particular the decision cannot be made earlier than a round after the permutation is revealed.
To improve upon that in the optimistic case, we make the permutation $\pi_r$ partially deterministic and allow to recover the first entry of the permutation already in round $r$.
More specifically, the mechanism of determining $\pi_r$ is as follows:
\begin{itemize}
    \item Generate pseudo-randomly an index $i_0 \in [N]$ based on just the round number $r$ or on the head of round $r-1$.
    \item Let $\tau_r$ be a random permutation constructed as previously using $\msf{SecretBits}(\cdot, r+5)$.
    \item In $\pi_r$ as first come all units created by node $i_0$ in round $r$ (there might be multiple of them because of forking), sorted by hashes, and then come all the remaining units of round $r$ in the order as determined by $\tau_r$.
\end{itemize}
\noindent Such a design allows the nodes to learn $i_0$ and hence the first candidate for the head right away in round $r$ and thus (in case of positive decision) choose it already in round $r+3$.
The remainder of the permutation is still random and cannot be manipulated by the adversary, thus the theoretical guarantees on latency are preserved.
For completeness we provide pseudocode for the new version of $\msf{GeneratePermutation}$.
In the pseudocode $\msf{DefaultIndex}$ can be any deterministic function that is computed from $\cD$ after the decision on round $r-1$ has been made.
Perhaps the simplest option is $\msf{DefaultIndex}(r, \cD):=1+(r~~\mathrm{mod}~~N)$.
In practice we might want to use some more complex strategy, which promotes nodes creating units that spread fast; for instance, it could depend on the head chosen for the $(r-1)$-round.

\renewcommand{\algorithmcfname}{$\msf{GeneratePermutation}(r,\cD)$}
\begin{algorithm-hbox}
    \lIf{$\H(\cD)<r+3$}{{\bf output $\bot$}}
    $i_0 \leftarrow \msf{DefaultIndex}(r,\cD)$\\
    $(V_1, V_2,\ldots, V_l) \leftarrow $ {list of units created by $\cP_{i_0}$ in round $r$ sorted by hashes}\\
    \tcc{Typically $l=1$, $l>1$ can only happen if $\cP_{i_0}$ forked in round $r$.}
    \For{each unit $U$ of round $r$ in $\cD$}
    {
        $i \leftarrow \mbox{the creator of $U$}$\\
        $x\leftarrow \msf{SecretBits}(i, r+5, \cD)$\\
        \lIf{$x=\bot$}{{\bf output} $(V_1, V_2,\ldots, V_l)$}
        {\bf assign} $\mathrm{priority}(U) \leftarrow \hash(x||U) \in \{0,1\}^\lambda$\\
    }
    {\bf let} $(U_1, U_2, \ldots, U_k)$ be the units in $\cD$ of round $r$ not created by $\cP_{i_0}$ sorted by $\mathrm{priority}(\cdot)$\\
    {\bf output} $(V_1, V_2,\ldots, V_l,U_1, U_2, \ldots, U_k)$\\
  \caption{ }
\end{algorithm-hbox}

\noindent{\bf Common Votes.} In \quickalef{} we use a slightly modified sequence of initial deterministic common votes as compared to \alef{}. 
This change is rather technical but, roughly speaking, is geared towards making the head decision as quickly as possible in the protocol.
Here is a short summary of the new $\msf{CommonVote}$ scheme.
Let $U_0$ be a unit created by $i$ and $r:=\H(U_0)$, then for $d=2,3,\ldots$ we define
$$\msf{CommonVote}(U_0, r+d) := 
\begin{cases}
1 & \mbox{for }d = 2\\
0 & \mbox{for }d = 3\\
\msf{SecretBits}(i, r+d+1) & \mbox{for }d\geq 4
\end{cases}
$$
Where in the last case we extract a single random bit from $\msf{SecretBits}$, as previously.

\subsection{Reducing size of units}\label{subsec:size-of-units}

We note that encoding parents of a unit by a set of their hashes is rather inefficient as it takes roughly $N \cdot \lambda$ bits to store just this information ($\lambda$ bits per parent).
Since the reasonable choices for lambda in this setting are $256$ or $512$, this is a rather significant overhead.
In this section we propose a solution that reduces this to just a small constant number (i.e. around $2$) of bits per parent.

Recall that in the absence of forks every unit is uniquely characterized by its creator id and by round number.
Since forking is a protocol violation that is severely penalized (for instance, by slashing the node's stake), one should expect it to happen rarely if at all.
For this reason it makes sense to just use this simple encoding of units: $(i, r)$ (meaning the unit created by $i$ at round $r$) and add a simple fallback mechanism that detects forked parents.
More specifically, the parents of a unit $U$ are encoded as a combination of the following two pieces of data
\begin{enumerate}
    \item A list $L_U=[r_1, r_2, \ldots, r_N]$ of length $N$ that contains the round numbers of parents of $U$ corresponding to creators $1,2, \ldots, N$. In the typical situation the list has only a small number of distinct elements, hence can be compressed to just $2$ or $3$ bits per entry.
    \item The ``control hash" $$h_U := h(h_1||h_2||\ldots||h_N),$$ where $h_1, h_2, \ldots, h_N$ are the hashes of parents of $U$.
\end{enumerate}
\noindent 
Consequently, the above encoding requires $O(N + \lambda)$ bits instead of $\Omega(N\lambda)$ bits as the original one.
While such a construction does not allow to identify forks straight away, it does allow to identify inconsistencies. Indeed, if an honest node $k$ receives a unit $U$ from node $j$, then it will read the parent rounds and check if all parent units are already present in $\cD_k$. If not, then these units will be requested from the sender of $U$, and failing to obtain them will result in dropping the unit $U$ by node $k$.
In case $k$ knows all parent units, it may verify the control hash and if inconsistency is detected, node $k$ should request\footnote{The request should include signed hashes of all units from $\cD_k$ that are indeed parents of $U$ to prevent $k$ from sending false requests.} all parent units present in $\cD_j$ from node $j$. If node $j$ is also honest, then node $k$ will reveal a forking unit among parents of $U$ and issue an alert before adding $U$ to $\cD_k$. On the other hand, in case $\cP_j$ does not provide parents of $U$, $\cP_k$ can simply drop the unit $U$.

\section{Analysis of protocols constructing \chdags{}}

Within the paper, two possible protocols constructing \chdag{} are considered, namely, \msf{DAG}-\msf{Grow} and \msf{Quick}-\msf{DAG}-\msf{Grow}. 
In this section, we analyze these protocols in the context of desirable properties as in definition \ref{def:protocol-properties}.

\subsection{\msf{DAG}-\msf{Grow}}\label{sec:dag_grow_proof}
We provide a proof of Theorem~\ref{thm:broadcast-to-dag}, i.e., we show that the \msf{DAG}-\msf{Grow} protocol is reliable, ever-expanding, fork-free and advances the DAG rounds at least as quickly as the asynchronous are progressing. 

\begin{proof}[Proof of Theorem~\ref{thm:broadcast-to-dag}]$ $ 

\noindent{\bf Reliable.} Node $\cP_i$ adds a unit $U$ to its \chdag{} only if it is locally output by \chrbc{}.
By the properties of \chrbc{} we conclude that if a unit is locally output for one honest node, it is eventually locally output by every other honest node.
\ourskip

\noindent{\bf Ever-expanding.}
Assume the contrary, i.e., there exists $r$ such that no honest node $\cP_i$ can gather $2f+1$ units of round $r$ in $\cD_i$ and hence no node is able to produce a unit of round $r+1$.
Let $r_0$ be the minimum such $r$.
We know that $r_0 > 0$, since every honest node participates in broadcasting of units of round $0$ without waiting for parents.
The units of round $r_0-1$ created by honest nodes are eventually added to local copies of all honest nodes, hence at some point, every honest node can create and broadcast a unit of round $r$, which then is eventually included in all local copies of $\chdag{}$.
As there are $2f+1$ honest nodes, we arrive at a contradiction.
\ourskip

\noindent{\bf Fork-free.}
This is a direct consequence of the ``reliability" of \chrbc{}, see Lemma \ref{lemma:RBC}({\it Fast Agreement})
\ourskip

\noindent{\bf DAG rounds vs asynchronous rounds.}
By Lemma~\ref{lemma:unit-production-speed}
each honest node produces at least one new unit each $5$ \asynchr{}s.
Since each new unit needs to have higher \dagr{}, it concludes the proof. 
\end{proof}

\subsection{\msf{Quick}-\msf{DAG}-\msf{Grow}}\label{sec:quick_aleph_dag}

Since \msf{Quick}-\msf{DAG}-\msf{Grow} does not rely on reliable broadcast, it does not enjoy all the good properties of the \msf{DAG}-\msf{Grow} protocol.
Most notably, it allows nodes to process forks and add forked units to the \chdag{}. 
We prove that it is still reliable and ever-growing, and additionally that the number of forks created by a single node is bounded by $N$ for each round, i.e., we prove Theorem~\ref{thm:quick-dag-grow-properties}.

\begin{proof}[Proof of Theorem~\ref{thm:quick-dag-grow-properties}]$ $

\noindent{\bf Reliable.}
First, we observe that while the alert system may slow the protocol down for some time, it may not cause it to stall indefinitely.
Indeed, each honest node needs to engage in at most $N$ alert \msf{RBC} instances per each discovered forker, i.e., at most $\frac{1}{3}N^2$ instances in total. 
Since an RBC instance started by an honest node is guaranteed to terminate locally for each node, there is some point in time of the protocol execution, name it $T_0$, after which no more alerts instantiated by honest nodes are active anymore.
Consequently, honest nodes can gossip their local copies of \chdag{} after $T_0$ without any obstacles.
This means in particular that every pair of honest nodes exchange their \chdag{} information infinitely often (with probability one), which implies reliability.

\ourskip

\noindent{\bf Ever-expanding.}
This follows from reliability as the set of honest is large enough ($2f+1$) to advance the rounds of the \chdag{} by themselves.
Indeed, by induction, for every round number $r>0$ every honest node eventually receives at least $2f+1$ units created by non-forkers (for instance honest nodes) in round $r-1$ and hence can create a correct unit of round $r$ by choosing these units as parents.

\noindent{\bf Bounding number of forks.}
An honest node $i$ adds a fork variant to its local copy of \chdag{} only in one of the two scenarios:
\begin{itemize}
    \item It was the first version of that fork that $i$ has received,
    \item Some other node has publicly committed to this version via the \msf{alert} system. 
\end{itemize}
Since there are only $N-1$ other nodes that could cause alerts, the limit of maximum of $N$ versions of a fork follows. 
\end{proof}

\section{Analysis of consensus protocols}\label{sec:consensus_analysis}

The organization of this section is as follows:
In Subsection~\ref{subsec:analysis_alef} we analyze the consensus mechanism in the \alef{} protocol.
In particular we show that $\msf{ChooseHead}$ is consistent between nodes and incurs only a constant number of rounds of delay in expectation. 
This also implies that every unit (created by an honest node) waits only a constant number of asynchronous rounds until it is added to the total order, which is necessary in the proof of Theorem~\ref{thm:AB}.
The Subsection~\ref{subsec:analysis_quickalef} is analogous but discusses \quickalef{} (introduced in Section~\ref{sec:practical}) instead.
Finally in Subsection \ref{subsec:synchronous} we show an optimistic bound of $3$ rounds validation for \quickalef{}.

Since the proofs in Subsection~\ref{subsec:analysis_quickalef} are typically close adaptations of these in Subsection~\ref{subsec:analysis_alef} we recommend the reader to start with the latter.
In the analysis of \quickalef{} we distinguish between two cases: when the adversary attacks using forks, and when no forks have happened recently.
In the former case the latency might increase from expected $O(1)$ to $O(\log N)$ rounds, yet each time this happens at least one of the nodes gets banned\footnote{Banning forkers is implemented by adding one simple rule to the growing protocol: never pick forking nodes as parents.} and thus it can happen only $f$ times during the whole execution of the protocol.
We do not attempt to optimize constants in this section, but focus only on obtaining optimal asymptotic guarantees.
A result of practical importance is obtained in Subsection~\ref{subsec:synchronous} where we show that $3$ rounds are enough in the "optimistic" case.

Throughout this section we assume that the $\msf{SecretBits}$ primitive satisfies the properties stated in Definition~\ref{def:secret-bits}.
An implementation of $\msf{SecretBits}$ is provided in Section~\ref{ssec:setup-details}.

\subsection{\alef{}}\label{subsec:analysis_alef}

\noindent 

We would to prove that the expected number of rounds needed for a unit to be ordered by the primitive OrderUnits is constant, and that each honest node is bound to produce the same ordering. 
To achieve this, first we need a series of technical lemmas.

\begin{lemma}[Vote latency]\label{lemma:vote-time} 
Let $X$ be a random variable which, for a unit $U_0$, indicates the number of rounds after which all units vote unanimously on $U_0$.
Formally, let $U_0$ be a unit of round $r$ and define $X=X(U_0)$ to be the smallest $l$ such that there exists $\sigma\in \{0, 1\}$ such that for every unit $U$ of round $r+l$ we have $Vote(U_0, U) = \sigma$. 
Then for $K\in\mathbb{N}$ we have
    \[
        P(X\geq K)\leq 2^{4-K}.
    \]
\end{lemma}
\begin{proof}
Fix a unit $U_0$.
Let $r'>r+4$ be a round number and let $\cP_k$ be the first honest node to create a unit $U$ of the round $r'$.
Let $\sigma$ denote a vote on $U_0$ that was cast by at least $f+1$ units in $\down{U}$. 
Then, every unit of round $r'$ will have at least one parent that votes $\sigma$. 
It is easy to check that if also $\sigma = \msf{CommonVote}(U_0, r', \cD_k)$, then every unit of round $r'$ is bound to vote $\sigma$.

By Definition \ref{def:secret-bits} (1), the votes of units in $\down{U}$ are independent of $\msf{CommonVote}(U_0, r', \cD_k)$ since it uses $\msf{SecretBits}(i, r', \cD_k)$.
Therefore, if at round $r'-1$ voting was not unanimous, then with probability at least $1/2$ it will be unanimous starting from round $r'$ onward. 
Since $P(X\geq 5)\leq \frac{1}{2}$, we may calculate $P(X\geq K)\leq 2^{-K+4}$, for $K>4$ by induction, and observe that it  is trivially true for $K=1,2,3,4$.
\end{proof}

\begin{lemma}[Decision latency]\label{lemma:UnitDecide}
Let $Y$ be a random variable that, for a unit $U_0$, indicates the number of rounds after which all honest nodes decide on $U_0$.
Formally, let $U_0$ be a unit of a round $r$ and define $Y=Y(U_0)$ to be the smallest $l$ such that there exists $\sigma$ such that for every honest node $\cP_i$ if $\H(\cD_i)\geq r+l$, then $\msf{Decide}(U_0; \cD_i)=\sigma$. 
Then for $K\in\mathbb{N}, K>0$ we have
\[
    P(Y\geq K)\leq \frac{K}{2^{K-5}} = O\inparen{K\cdot 2^{-K}}.
\]
\end{lemma}

\begin{proof}
First, we need to check that $Y$ is well-defined, i.e. for every $U_0$ there is $l\in\mathbb N$ and there is $\sigma$ such that for every honest node $\cP_i$ if $\H(\cD_i)\geq r+l$, then $\msf{Decide}(U_0; \cD_i)=\sigma$.

We observe that if there is an honest node $\cP_k$ and a unit $U\in\cD_k$ such that $\msf{UnitDecide}(U_0, U, \cD_k)=\sigma\neq\bot$, then eventually for every honest node $\cP_i$, we will have 
$\msf{Decide}(U_0; \cD_i)=\sigma$.
Indeed, fix an honest node $\cP_k$ and a unit $U\in\cD_k$ of round $r'>r+4$ such that $\msf{UnitDecide}(U_0, U, \cD_k)=\sigma\neq\bot$.
Then, at least $2f+1$ units of round $r'-1$ vote $\sigma$ and $\msf{CommonVote}(U_0, r')=\sigma$. 
Now by the definitions of $\msf{UnitVote}$, we see that every other unit $U'$ of round $r'$ must have $\msf{UnitVote}(U_0, U', \cD_k)=\sigma$, and thus by a simple induction we get that, every unit of round greater than $r'$ will also vote $\sigma$ on $U_0$.
Finally, by the definition of $\msf{Decide}$ we see that once $\msf{UnitDecide}$ outputs a $\sigma\neq\bot$, then the result of $\msf{UnitVote}$ is $\sigma$ and does not change with the growth of the local copy of $\chdag{}$.

 By the above, if we prove that $\msf{UnitDecide}$ will always eventually output non-$\bot$, then the definition of $Y$ will be correct.

Let $Y'=Y'(U_0)$ be defined as the smallest $l$ such that for every unit $U$ of round $r+l$ the function $\msf{UnitDecide(U_0, U)}$ outputs $\sigma$. 
If $X(U_0)=l$, then either $Y'(U_0)=l$, or $Y'(U_0)$ is the index of the first round after $l$ at which $\msf{CommonVote}$ equals $\sigma$.
The probablity that at round $r+l$ nodes vote unanimously $\sigma$ for the first time, but then no $\msf{CommonVote}$ was equal to $\sigma$ before round $r+L$ equals $P(X=l)/2^{L-1-l}$, and thus we have:

\begin{align*}
P(Y'\geq L) &\leq \sum_{l=0}^{L-1} 2^{-L+1+l}P(X=l) + P(X\geq L) \\
&= \sum_{l=0}^{L-1}2^{-L+1+l}P(X=l) + \sum_{l=L}^{+\infty} P(X=l) \\
&= 2^{-L}P(X\geq 0)  + \sum_{l=0}^{L-1}  2^{-L+l}P(X\geq l) \\
& \leq \frac{L+1}{2^{L-4}}
\end{align*}

Since $X$ is well-defined, then so is $Y'$ and we see that  $\msf{UnitDecide}$ eventually outputs a non-$\bot$.
Moreover, observe that $Y\leq Y'+1$.
Indeed, if every node of round $r+l$ has decided $\sigma$, then the $\chdag{}$ of height $\geq r+l+1$ includes at least one of the decisions, and we can read the last required secret. 
Finally, $Y\leq Y'+1$ implies 
\[ 
    P(Y\geq K)\leq P(Y'\geq K-1)\leq \frac{K}{2^{K-5}}.
\]
\end{proof}

The following lemma guarantees that at every level there are a lot of units that are decided on $1$.
\begin{lemma}[Fast positive decisions]\label{lemma:popular-units}
Assume that an honest node $\cP_i$ has just created a unit $U$ of round $r+3$. 
At this point in the protocol execution, there exists a set $\mathcal{S}_r$ of at least $2f+1$ units of round $r$ such that for every $U_0\in \mathcal{S}_r$, every honest node will eventually have $\msf{Decide}(U_0, \cD_i)=1$.
\end{lemma}
\begin{proof}

Let $K$ be a set of $2f+1$ nodes that created units in $\down{U}$. 
Additionally, let $\mathcal{T}$ be the set of $2f+1$ units of round $r+1$ created by nodes in $K$.
Every unit in $\down{U}$ is above $2f+1$ units of round $r+1$, hence it is above at least $f+1$ units in $\mathcal{T}$ (the remaining $f$ units may be created by nodes outside of $K$). 
By the pigeonhole hole principle, there exists a unit $U_0\in\mathcal{T}$ such that at least $f+1$ units in $\down{U}$ are above it. 
Set $\mathcal{S}_r := \down{U_0}$.

Let $V$ be a unit of round $r+3$ and $V'\in\down{V}$ be its parent that is above $U_0$ (which has to exist since every subset of $2f+1$ units of round $r+2$ must include at least one unit above $U_0$).
Since $\msf{CommonVote}(U_r,W)$ equals $1$ for all $W$ of rounds $r+1,r+2,r+3$, for $U_r\in S_r$ we have $\msf{Vote}(U_r,U_0)=1$, hence $\msf{Vote}(U_r,V')=1$, and finally, $\msf{Vote}(U_r,V)=1$.

Thus, during all subsequent rounds $\msf{Vote}(U_r,\cdot)=1$ and $U_r$ will be decided $1$ as soon as the next $\msf{CommonVote}$ equals $1$.
\end{proof}

Intuitively, this lemma states that the set of potential heads that will be eventually positively decided is large.
Additionally, it is defined relatively quickly, i.e., before the adversary sees the content of any unit of round $r+3$.
Importantly, it is defined before the permutation in which potential heads will be considered is revealed.
Note that it has some resemblance to the spread protocol defined in~\cite{BE03}.

In general, the above result cannot be improved, as the adversary can always slow down $f$ nodes, thus their units, unseen by others, may not be considered ``ready'' to put in the linear order.
Also, note that this lemma does not require any common randomness, as the votes in round $r+3$ depend only on the structure of the $\chdag{}$.

\begin{lemma}[Fast negative decisions]\label{lemma:negative-decisions}
Let $U$ be a unit of round $r$ such that for some honest node $i$, $U\notin \cD_i$ even though $\H(\cD_i) \geq r+4$.
Then, for any local view $\cD$, $\msf{Decide}(U,\cD) \neq 1$. 
\end{lemma}

\begin{proof}
Assume for contradiction that $\cD$ is local view such that $\msf{Decide}(U,\cD) = 1$.
By definition of \msf{Decide}, there has to be unit $V_1\in\cD$ such that $\msf{UnitDecide}(U,V_1,\cD) = 1$.
Let now $\cD'$ be a local view such that $\cD,\cD_i\subseteq\cD'$.
Such $\cD'$ is possible to construct since all local views have to be compatible by point $(2)$ of Theorem~\ref{thm:broadcast-to-dag}.

Let $V_0$ be a unit of \dagr{} $\H(\cD_i)+4$ in $\cD_i$ (and hence in $\cD'$).
Since $V_0$ can't be above $U$, for every $V\in \down{V_0}$ we have $\msf{Vote}(U,V,\cD_i) = 0$.
Since $|\! \down{V_0}|\geq 2f+1$, each unit of \dagr{} $r+4$ in $\cD'$ is above at least one unit in $\down{V_0}$ and hence votes either $0$ or by \msf{CommonVote}.
But \msf{CommonVote} for $U$ always equals $0$ at \dagr{} $\H(U)+4$ by definition, and consequently, each unit of \dagr{} $r+4$ has to \msf{Vote} $0$ for $U$.
If all units votes unanimously at any given \dagr{}, such vote is passed over to the next \dagr{}s contradicting the existence of $V_1$ deciding $1$, what concludes the proof.
\end{proof}

\begin{fact}\label{fact:var-max}
Let $X_1, \ldots, X_M$ be random variables, $K\in\mathbb{R}$. 
Then
\[
    P(\max(X_1, \ldots, X_M)\geq K) \leq \sum_{m=1}^M P(X_m\geq K)
\]
\end{fact}
\begin{proof}
Simply observe that if $\max(X_1, \ldots, X_M)\geq K$, then for at least one $m\in [M]$ it must be $X_m\geq K$
\[
    \{\max(X_1, \ldots, X_M)\geq K\} \subseteq  \bigcup_{m=1}^M \{X_m\geq K\}.
\]
\end{proof}

The next lemma shows the bound on the number of additional rounds that are needed to choose a head.

\begin{lemma}[ChooseHead latency]\label{lemma:choose-head}
The function \msf{ChooseHead} satisfies the following properties:
\begin{itemize}
    \item {\bf Agreement}. For every round $r$ there is a uniquely chosen head $U$, i.e., for every \chdag{} $\cD$ maintained by an honest node, $\msf{ChooseHead}(r,\cD)\in \{\bot, U\}$.
    \item {\bf Low latency}. Let $Z_r$ be a random variable defined as the smallest $l$ such that for every local copy $\cD$ of height $r+l$ we have $\msf{ChooseHead}(r, \cD) \neq \bot$.
Then for $K\in\mathbb{N}, K>0$ we have $P(Z_r\geq K)= O\inparen{K\cdot 2^{-K}}$
\end{itemize}

\end{lemma}

\begin{proof}$ $

\noindent {\bf Agreement}. 
Suppose for the sake of contradiction that there exist two \chdags{} $\cD_i$ and $\cD_j$ maintained by honest nodes $i,j$, such that  $\msf{ChooseHead}(r,\cD_i) = U$ and $\msf{ChooseHead}(r,\cD_j) = U'$ for two distinct units $U, U'$.
Note that for this we necessarily have
$$ \H(\cD_i)\geq r+4\ \ \ \mbox{ and}\ \ \  \H(\cD_j)\geq r+4,$$
as otherwise $\msf{GeneratePermutation}(r,\cdot)$ would return $\bot$.

Further, as necessarily $\msf{Decide}(U,\cD_i)=1$ and $\msf{Decide}(U',\cD_j)=1$, by Lemma~\ref{lemma:negative-decisions} we obtain that both $U$ and $U'$ need to be present in $\cD_i$ and $\cD_j$.

To conclude the argument it remains to observe that the priorities computed by  $\msf{GeneratePermutation}(r,\cdot)$ are consistent between different \chdags{} as they are computed deterministically based on the common outcome of $\msf{SecretBits}$.
We thus arrive at the desired contradiction: if, say, $U$ has higher priority than $U'$ then the $j$th node should have choosen $U$ instead of $U'$.

\noindent {\bf Constant latency}. 
Consider the set $\mathcal{S}_r$ from Lemma~\ref{lemma:popular-units}.
Let $\cP_i$ be a node, $x=\msf{SecretBits}(i, r+4, \cD_i)$, and let $$\pi_r=(U_1,
\ldots,U_k)=\msf{GeneratePermutation}(r, \cD_i).$$
The permutation given by priorities  $\hash(x||U_i)$ for $i=1,2, \ldots,N$ is uniformly random by Definition \ref{def:secret-bits} (1) and the assumption that $\hash$ gives random bits in $\{0,1\}^\lambda$.
Moreover, it is independent of the set $\mathcal{S}_r$, as this set is defined when the first unit of round $r+3$ is created and $\pi_r$ was unknown before it.

Let $S_r$ denote a random variable defined as the smallest index $s$ such that $s$-th element of the permutation $\pi_r$ is in $\mathcal{S}_r$.
Since the permutation is uniformly random and independent of $\mathcal{S}_r$, then for $s=1, \ldots, f+1$ we have
\[
    P(S_r=s) = \dfrac{2f+1}{3f+2-s}\prod_{j=1}^{s-1} \dfrac{f+1-j}{3f+2-j},
\]
and $P(S_r>f+1)=0$, hence for $s=1, \ldots, k$, we have
\[
P(S_r=s)\leq 3^{-s+1}.
\]

By Lemma~\ref{lemma:popular-units}, all units from $\mathcal{S}_r$ must be decided $1$, then to calculate $\msf{ChooseHead}(r, \cD_i)$ we need to wait for decisions on at most $S_r$ units of round $r$.
Using Fact~\ref{fact:var-max} for random variables from Lemma~\ref{lemma:UnitDecide} we see that for $K>4$
\[
    P(Z_r\geq K\mid S_r=s) \leq \sum_{s=1}^k P(Y(U_s)\geq K)  \leq s\cdot \frac{K}{2^{K-5}},
\]
therefore
\[
    P(Z_r\geq K) \leq \sum_{s=1}^N P(S_r=s)P(Z_r\geq K\mid S_r=s) =O\inparen {K\cdot 2^{-K}}.
\]
\end{proof}

We end this section with a theorem that shows how long does it take to append a head to the linear order.

\begin{theorem}[OrderUnits latency]\label{thm:choose-head-latency}
Let $W_r$ be a random variable that indicates number of rounds required to append a head of round $r$ to the linear order. 
Formally, $W_r$ is defined as the smallest $l$ such that for every local copy $\cD_i$ of height $r+l$ we have
$$\msf{ChooseHead}(r, \cD_i) \in \msf{OrderUnits}(\cD_i).$$ 
Then
$$\mathbb{E}(W_r) =O(1)$$.
\end{theorem}

\begin{proof}
The first part of the proof will show that for all $K\in\mathbb{N}, K>0$ we have
\[
    P(W_r\geq K)= O\inparen{K\cdot 2^{-K}}.
\]
Let $\cP_i$ be a node.
$\msf{ChooseHead}(r, \cD_i) \in \msf{OrderUnits}(\cD_i)$ implies that for all $j=0, 1, \ldots, r$ we have $$\msf{ChooseHead}(j, \cD_i)\neq\bot.$$ 
Therefore, if $W_r\geq K$, then for at least one round $j$ we have $Z_j\geq K+r-j$, where $Z_j$ is a random variable from Lemma~\ref{lemma:choose-head}.
Since
\[
    (W_r\geq K)\subseteq  \bigcup_{j=0}^r (Z_j\geq K+r-j),
\]
then
\begin{align*}
P(W_r\geq K) &\leq \sum_{j=0}^r P(Z_j\geq K+r-j)  \leq O(1)\sum_{j=0}^r \frac{K+j}{2^{K+j}} \\
&= O(1) \sum_{j=K}^{K+r} \frac{j}{2^{j}} \leq O(1) \sum_{j=K}^{+\infty} \frac{j}{2^{j}} = O\inparen{ K\cdot 2^{-K}}.
\end{align*}

Since $W_r$ has values in $\mathbb{N}$, then
$$\mathbb{E}(W_r) = \sum_{K=1}^{+\infty} P(W_r\geq K),$$
and finally we have
\[
\mathbb{E}(Z_r) = \sum_{K=1}^{+\infty} P(Z_r\geq K) = O(1) \sum_{K=1}^{+\infty} \frac{K}{2^{K}} = O(1).
\]
\end{proof}

\subsection{ \quickalef{}}\label{subsec:analysis_quickalef}
\noindent
In this section we analyze the \quickalef{} consensus mechanism. 
The main difference when compared to \alef{} is that now we allow the adversary to create forks.
Below, we introduce two definitions that allow us to capture forking situations in the \chdag{} so that we can prove how the latency of the protocol behaves in their presence.
At the end of this subsection  we also show that these situations are rare and can happen only $O(N)$ times.

\begin{definition2}\label{def:fork-witness}
Let $r$ be a round number. We say that a unit $U$ of round $r+4$ is a \emph{fork witness} if it is above at least two variants of the same unit of round $r$ or $r+1$.
\end{definition2}

\begin{definition2}\label{def:locally-forked}
Let $\mathcal{U}$ be a set of all $2f+1$ units of round $r+4$ created by honest nodes. We say that the $r$th round of the \chdag{} is \emph{publicly forked}, if there are at least $f+1$ fork witnesses among units $\mathcal{U}$. %
\end{definition2}

\noindent
We proceed with the analysis of the latency of \quickalef{} starting again by determining the delay in rounds until all nodes start to vote unanimously.

\begin{lemma}[Vote latency \lbrack for \quickalef{}\rbrack]\label{lemma:quick-vote-time} 
Let $X$ be a random variable that for a unit $U_0$ represents the number of rounds after which all units unanimously vote on $U_0$. 
Formally, let $U_0$ be a unit of round $r$ and define $X=X(U_0)$ to be the smallest $l$ such that for some $\sigma\in \{0, 1\}$, for every unit $U$ of round $r+l$ we have $Vote(U_0, U) = \sigma$. 
Then for $K\in\mathbb{N}$ we have
    \[
        P(X\geq K)\leq \inparen{\frac34}^{(K-6)/2}.
    \]
\end{lemma}

\begin{proof}
Fix a unit $U_0$.
Let $r'>r+5$ be an odd round number and let $\cP_k$ be the first honest node to create a unit $U$ of round $r'$.

Let $\sigma$ denote a vote such that there is no unit $V\in\down{U}$ for which every unit in $\down{V}$ votes $\sigma$ on $U_0$.
The value of $\sigma$ is well-defined, i.e., at least one value from $\{0,1\}$ satisfies this property (pick $\sigma=0$ if both $0$ and $1$ satisfy the above property).
Indeed, every two units $V, V'\in\down{U}$ have at least one parent unit in common, thus it is not possible that every unit in $\down{V}$ votes $\sigma$, and every unit in $\down{V'}$ votes $1-\sigma$.

Recall that the adversary cannot predict $\msf{CommonVote}(U_0, r'-1)$ or $\msf{CommonVote}(U_0, r')$ before $U$ is created.
Now, if 
\[\msf{CommonVote}(U_0, r'-1, \cD_k) = \msf{CommonVote}(U_0, r', \cD_k) = 1-\sigma,\]
then every unit in $\down{U}$ votes $1-\sigma$ on $U_0$, and so does $U$.
If some other unit $U'$ of round $r'$ were to vote $\sigma$ on $U_0$, then it would mean that at least $2f+1$ units in $\down{U'}$ vote $\sigma$ on $U_0$, which is necessary to overcome the $r'$-round common vote.
But this is impossible, as this set of $2f+1$ units must include at least one unit from $\down{U}$.
Therefore, if at round $r'-2$ voting was not unanimous, then with probability at least $1/4$ it will be unanimous starting from round $r'$ onward. 
Hence, we may check $P(X\geq K)\leq \inparen{\frac34}^{(K-6)/2}$, for $K>5$ by induction, and observe that it is trivially true for $K=1,2,3,4,5$.
\end{proof}
\noindent 
Here we obtain a slightly worse bound than in the corresponding Lemma~\ref{lemma:vote-time} for the \alef{} protocol.
The main reason for that is that the property that units $U$ and $U'$ have at least $f+1$ common parent units does not longer hold, as now up to $f$ of them may be forked.
One can try to improve this result by introducing different voting rules, and analyzing forked and non-forked cases separately, but this would only make the argument cumbersome, and would not help to decrease the asymptotic latency in the asynchronous case, nor the numerical latency in the partially synchronous case.

\begin{lemma}[Decision latency \lbrack for \quickalef{}\rbrack]\label{lemma:quick-UnitDecide}
Let $Y$ be a random variable that for a unit $U_0$ indicates the number of rounds after which all honest nodes decide on $U_0$.
Formally, let $U_0$ be a unit of a round $r$ and define $Y=Y(U_0)$ to be the smallest $l$ such that there exists $\sigma$ such that for every honest node $\cP_i$ if $\H(\cD_i)\geq r+l$, then $\msf{Decide}(U_0; \cD_i)=\sigma$. 
Then for $K\in\mathbb{N}$ we have
\[
    P(Y\geq K)=  O\inparen {K\cdot \inparen{\frac34}^{K/2}}.
\]
\end{lemma}

\begin{proof}
The proof is analogous to that of Lemma~\ref{lemma:UnitDecide}.
First, we show that $Y$ is well-defined.
The argument is exactly the same, not affected by introducing forks.
If at least one unit decides a non-$\bot$ value $\sigma$, then every unit of equal or higher round votes $\sigma$, and eventually there will be a round with common vote equal to $\sigma$, which will trigger the decision among all nodes.

We define $Y'=Y'(U_0)$ as in the proof of Lemma~\ref{lemma:UnitDecide} and obtain

\begin{align*}
P(Y'\geq L) &\leq 2^{-L}P(X\geq 0)  + \sum_{l=0}^{L-1}  2^{-L+l}P(X\geq l) \\
& \leq (L+1)\inparen{\frac34}^{(L-6)/2}
\end{align*}

The final difference is that now $Y\leq Y'+2$, because we need one additional round to read the value of the last required secret:
\[ 
    P(Y\geq K)\leq P(Y'\geq K-2)\leq (K-1)\inparen{\frac34}^{(K-8)/2}.
\]
\end{proof}

\noindent
Similarly as for \alef{}, for $O(1)$ latency it is crucial to show that a large set of units of round $r$ is guaranteed to be decided positively.
\begin{lemma}[Fast positive decisions \lbrack for \quickalef{}\rbrack ]\label{lemma:quick-popular-units}
Let $r\geq 0$ be a round number. Assume that an honest node $\cP_i$ has created a unit $U$ of round $r+4$.
There exists at least one unit $U_0$ of round $r$ such that $\msf{Decide}(U_0, \cD_i)=1$. If additionally $U$ is not a fork witness, then there exists a set $\mathcal{S}_{i,r}$ of at least $1.5f+1$ units of round $r$ such that for every $U_0\in \mathcal{S}_{i,r}$ we have $\msf{Decide}(U_0, \cD_i)=1$.
\end{lemma}

\begin{proof}
For convenience, denote by $\mathcal{U}_{r'}$ the set of all units of round $r'$ below $U$.

Let us first observe that whenever for a unit $U_0$ every unit $V\in \mathcal{U}_{r+2}$ votes $1$ on $U_0$, then it is decided $1$ by $U$, i.e.,
$$\inparen{\forall~~V\in \mathcal{U}_{r+2} \ \ \ \msf{Vote}(U_0, V)=1} \Rightarrow \msf{UnitDecide}(U_0, U) = 1,$$
This can be seen as follows: 
\begin{itemize}
    \item every unit $W\in \mathcal{U}_{r+3}$ votes $1$, i.e., $\msf{Vote}(U_0, W)=1$ because $\down{W}$ unanimously vote $1$,
    \item since the common vote at round $r+4$ is $1$ and all round-$(r+3)$ parents of $U$ unanimously vote $1$, $\msf{UnitDecide}(U_0, U) = 1$.
\end{itemize}
\noindent 
The remaining part of the proof is to show that in the general case there is at least one unit $U_0 \in \mathcal{U}_r$ that is voted $1$ by every unit in $\mathcal{U}_{r+2}$ and that there is at least $1.5f+1$ such units in the case when $U$ is not a forking witness.

\medskip
For the first part, note that every unit $W$ of round $r+1$ has at least $f+1$ honest parents at round $r$, thus by reverse counting there exists a unit $U_0$ or round $r$, created by an honest node that is a parent of at least $f+1$ honest units at round $r+1$.
Note that $U_0$ is as desired, i.e., every unit $V$ of round $r+2$ votes $1$ for $U_0$ because at least one honest unit $W \in \down{V}$ is above $U_0$.

\medskip
For the second part, suppose that $U$ is not a fork witness, in which case the sets $\mathcal{U}_r$ and $\mathcal{U}_{r+1}$ contain no forks and hence if we denote
\begin{align*}
    a &:= N-|\mathcal{U}_r|,\\
     b&:=N-|\mathcal{U}_{r+1}|,
\end{align*}
then $0\leq a, b \leq f$.
We define the set $\mathcal{S}_{i,r}$ to be all units $U_0\in\mathcal{U}_r$ that are parents of at least $f-b+1$ units from $\mathcal{U}_{r+1}$.

Suppose now that $U_0$ is such a unit, we show that all units in $\mathcal{U}_{r+2}$ vote $1$ on $U_0$.
Indeed, take any $V\in \mathcal{U}_{r+2}$.
Since $V$ has at least $2f+1$ parents in $\mathcal{U}_{r+1}$, at least one of them is above $U_0$ (by definition of $\mathcal{S}_{i,r}$) and thus $\down{V}$ does not unanimously vote on $0$.
The default vote at round $r+2$ is $1$ therefore $V$ has to vote $1$ on $U_0$, as desired.

Finally, it remains to show that $\mathcal{S}_{i,r}$ is large.
For brevity denote the number of elements in $\mathcal{S}_{i,r}$ by $x$.
We start by observing that the number of edges connecting units from $\mathcal{U}_r$ and $\mathcal{U}_{r+1}$ is greater or equal than $$(2f+1)\cdot|\mathcal{U}_{r+1}|=(2f+1)\cdot(N-b).$$
On the other hand, we can obtain an upper bound on the number of such edges by assuming that each unit from $\mathcal{S}_{i,r}$ has the maximum number of $N-b$ incoming edges and each unit from $\mathcal{U}_r\setminus\mathcal{S}_{i,r}$ has the maximum number of $f-b$ incoming edges, thus:
\[
x\cdot (N-b) + (N-a-x)(f-b) \geq (2f+1)(N-b),
\]
which after simple rearrangements leads to
\[
x \geq 2f+1 -\frac{(f-a)(f-b)}{2f+1}.
\]
Since  $0\leq a,b\leq f$ we obtain
\[
x\geq 1.5f + 1.
\]

\end{proof}

It is instructive to compare Lemma~\ref{lemma:quick-popular-units} regarding \quickalef{} to Lemma~\ref{lemma:popular-units} regarding \alef{}.
Note that in the non-forked case every unit from the set $\mathcal{U}_{r+2}$ has at least $2f+1$ parents in the set $\mathcal{U}_{r+1}$, and by pigeonhole principle must be above all units from $\mathcal{S}_{i,r}$.
Therefore, we decrease the number of rounds after which units from $\mathcal{S}_{i,r}$ become ``popular'' from $3$ to $2$.
This is achieved at the cost of reducing the size of this set by $0.5f$ units, but is necessary to speed up the protocol in the optimistic case.

\begin{lemma}[Fast negative decisions \lbrack for \quickalef{}\rbrack] \label{lemma:quick-negative-decisions}
Let $r\geq 0$ be a round number. Let $U$ be a unit of round $r+3$, and $U_0\not\leq U$ be a unit of round $r$. Then, for any local view $\cD\ni U_0$ with $\H(\cD) \geq r+5$ we have $\msf{Decide}(U_0, \cD) = 0$.
\end{lemma}

\begin{proof}
Let $U'$ be any unit of round $r+3$.
Then $U'$ votes $0$ on $U_0$.
Indeed, if $U'$ were to vote $1$, then there would be at least $2f+1$ units in $\down{U'}$ above $U_0$, and therefore at least one unit in $\down{U}$ above $U_0$, but this would imply $U\geq U_0$.
Hence, every unit of round greater or equal to $r+3$ votes $0$ on $U_0$, and finally every unit of round $r+5$ decides $0$ on $U_0$ due to the common vote being equal to $0$.
\end{proof}

\noindent In \quickalef{} the latency of choosing head can differ depending on whether we are in the forking situation or not.
The following lemma states this formally.

\begin{lemma}[ChooseHead latency \lbrack for \quickalef{}\rbrack]\label{lemma:quick-choose-head}
The function \msf{ChooseHead} satisfies the following properties:
\begin{itemize}
    \item {\bf Agreement}. For every round $r$ there is a uniquely chosen head $U$, i.e., for every \chdag{} $\cD$ maintained by an honest node, $\msf{ChooseHead}(r,\cD)\in \{\bot, U\}$.
    \item {\bf Low latency}. Let $Z_r$ be a random variable defined as the smallest $l$ such that for every local copy $\cD$ of height $r+l$ we have $\msf{ChooseHead}(r, \cD) \neq \bot$.
Then for $K\in\mathbb{N}$ we have
\[
   P(Z_r\geq K)= O\inparen{K\cdot\inparen{\frac34}^{K/2}}
\]
if the protocol is locally non-forked at round $r$, and
\[
   P(Z_r\geq K)= O\inparen{N^2\cdot K\cdot\inparen{\frac34}^{K/2}}
\]
otherwise.
\end{itemize}
\end{lemma}

\begin{proof}

\noindent {\bf Agreement}. 
Similarly as in the corresponding part of the proof of Lemma~\ref{lemma:choose-head} we observe that the permutation (computed by different nodes) along which the head at a given round is chosen is consistent in the sense that: if two nodes both hold a pair of units $U_0, U_1$ of round $r$ then these units will end up obtaining the same priorities at both nodes and thus will be placed in the same relative order in the permutation.

Consequently all we need to show is that at the time when an honest node $i$ decides a unit $U$ as the head of round $r$ then every unit $U'$ that is not part of its local copy $\cD_i$ will be decided $0$.
This indeed follows from Lemma~\ref{lemma:quick-negative-decisions} since for $\msf{ChooseHead}(r, \cD_i)$ to not be $\bot$ the \chdag{} $\cD_i$ needs to have $\H(\cD_i)\geq r+3$.

\noindent {\bf Low latency}. 
The proof of the non-forked case differs from the proof of Lemma~\ref{lemma:choose-head} in two ways.
Firstly, the set of units that will eventually be decided $1$ is smaller, but still has size $\Omega(N)$.
Secondly, the first unit in the permutation is known in advance.
Here we prove a slightly stronger result, namely we assume that the first $c\geq 0$ units are known, where $c$ should be thought of as a small constant (independent of $N$).
For concreteness one can simply think of $c=1$.

Let $r\geq 0$ be a round number. Let $U$ be the first $(r+5)$-round unit created by an honest node.
If the protocol is locally non-forked at round $r$, then there are at least $f+1$ $(r+4)$-round units that are not fork witnesses, hence there is at least one unit $V\in\down{U}$, created by node $i$, that is not a fork witness.

Let $\mathcal{U}_r$ denote all round-$r$ units below $V$, and consider the set $\mathcal{S}_{i,r}\subseteq \mathcal{U}_r$ from Lemma~\ref{lemma:quick-popular-units}.
The set $\mathcal{U}_r$ has at most $N$ elements and contains no forks.
From Lemma~\ref{lemma:quick-negative-decisions} we know that all round-$r$ units not in $\mathcal{U}_r$ (possibly including some forks of units from $\mathcal{U}_r$) are decided $0$ by every node at round $r+5$.
Therefore, we can only consider units from $\mathcal{U}_r$ when calculating $\msf{GeneratePermutation}(r)$.
Note that the unit priorities are independent from sets $\mathcal{S}_{i,r}, \mathcal{U}_r$, because the unit $V$ exists before the adversary can read $\msf{SecretBits}(\cdot, r+5)$ using shares from the unit $U$.

Let $s_{i,r}$ denote a random variable defined as the smallest index $s$ such that $s$-th element of the permutation on units $\mathcal{U}_r$ is in $\mathcal{S}_{i,r}$, excluding the first $c$ default indices (in other words, $P(s_{i,r}\leq c)=0$).
The permutation is uniformly random and independent of $\mathcal{S}_{i,r}$.
Since we want to calculate an upper bound on the probability distribution of $s_{i,r}$, then we might assume w.l.o.g. that the set $\mathcal{U}_r$ has exactly $3f+1-c$ elements, and that the set $\mathcal{S}_{i,r}$ has $1.5f+1-c$, in both cases excluding units created by the $c$ default nodes.

Then, for $s=1, \ldots, 1.5f+1$ we have
\[
    P(s_{i,r}=s) = \dfrac{1.5f+1-c}{3f+2-c-s}\prod_{j=1}^{s-1} \dfrac{1.5f+1-j}{3f+2-c-j},
\]
and $P(s_{i,r}>1.5f+1)=0$, hence for $s=1, 2, \ldots$ we have
\[
P(s_{i,r}=s)\leq 2^{-s+c}.
\]

By Lemma~\ref{lemma:quick-popular-units}, all units from $\mathcal{S}_{i,r}$ must be decided $1$, then to calculate $\msf{ChooseHead}(r, \cD_i)$ we need to wait for decisions on at most $s_{i,r}+c$ units of round $r$ (here we include the default indices with highest priority).
Using Fact~\ref{fact:var-max} for random variables from Lemma~\ref{lemma:quick-UnitDecide} we see that for $K>5$
\begin{align*}
    P(Z_r\geq K\mid S_r=s) &\leq \sum_{k=1}^{s+c} P(Y(U_k)\geq K) \\
    &\leq (s+c)(K-1) \inparen{\frac34}^{(K-8)/2},
\end{align*}
therefore
\begin{align*}
    P(Z_r\geq K) &\leq \sum_{s=1}^N P(S_r=s)P(Z_r\geq K\mid S_r=s) \\
    &= O\inparen{K\cdot\inparen{\frac34}^{K/2}}.
\end{align*}
\noindent 
On the other hand, if the round is publicly forked, then we observe that we need to wait for decision about at most $N^2$ units of round $r$ (since every unit can be forked at most $N$ times), and use Fact~\ref{fact:var-max}.
\end{proof}

\begin{theorem}[OrderUnits latency \lbrack for \quickalef{}\rbrack]\label{thm:quick-order-units-latency}
Let $W_r$ be a random variable representing the number of rounds required to append a head of round $r$ to the linear order, assuming that the default indices for rounds $0, \ldots, r$ are known. 
Formally, $W_r$ is defined as the smallest $l$ such that for every local copy $\cD_i$ of height $r+l$ we have
\[\msf{ChooseHead}(r, \cD_i) \in \msf{OrderUnits}(\cD_i).\]
Then
\[\mathbb{E}(W_r) =O(1)\]
if the $\dfrac{4}{\log\frac43}\log N$ rounds prior to $r$ are not publicly forked, and
\[\mathbb{E}(W_r) =O(\log N)\]
otherwise.
\end{theorem}

\begin{proof}
Let $\cP_i$ be a node.
$\msf{ChooseHead}(r, \cD_i) \in \msf{OrderUnits}(\cD_i)$ is equivalent to the statement that for $j=0, 1, \ldots, r$ we have $$\msf{ChooseHead}(j, \cD_i)\neq\bot.$$ 
Therefore, if $W_r\geq K$, then for at least one round $j$ we have $Z_j\geq K+r-j$, where $Z_j$ is a random variable from Lemma~\ref{lemma:quick-choose-head}.

We start with the forked variant. Since
\[
    (W_r\geq K)\subseteq\bigcup_{j=0}^r (Z_j\geq K+r-j),
\]
then

\begin{align*}
P(W_r\geq K) &\leq \sum_{j=0}^r P(Z_j\geq K+r-j) \\
&\leq O(1)\cdot N^2\sum_{j=0}^r(K+j)\inparen{\frac34}^{(K+j)/2} \\
&= O(1)\cdot N^2\sum_{j=K}^{K+r} j\inparen{\frac34}^{j/2} \\
&\leq O(1)\cdot N^2 \sum_{j=K}^{+\infty} j\inparen{\frac34}^{j/2} \\
&= O\inparen{N^2\cdot K\cdot\inparen{\frac34}^{K/2}}.
\end{align*}
\noindent
Note that above we used the fact that the constant hidden under the big O notation is independent of the round number.

Observe now that for $K>\dfrac{4}{\log\frac43}\log N$
$$P(W_r\geq K - 4\log\frac34\log N) = O\inparen{K\cdot\inparen{\frac34}^{K/2}},$$
and since $W_r$ has values in $\mathbb{N}$, then
$$\mathbb{E}(W_r) = \sum_{K=1}^{+\infty} P(W_r\geq K) = O\inparen{\log N}.$$

To analyze the non-forked case we make a very similar argument, but we use the tighter bound for the last $\dfrac{4}{\log\frac43}\log N$ rounds.
Denote $A:=\dfrac{4}{\log\frac43}\log N$ for brevity.

\begin{align*}
&P(W_r\geq K) \leq \sum_{j=0}^r P(Z_j\geq K+r-j) \\
&\leq O(1)\cdot \inparen{N^2\sum_{j=0}^{r-A}(K+j)\inparen{\frac34}^{(K+j)/2} + \sum_{j=r-A+1}^{r}(K+j)\inparen{\frac34}^{(K+j)/2}} \\
&\leq O(1)\cdot \inparen{\sum_{j=A}^{r}(K+j)\inparen{\frac34}^{(K+j)/2} + \sum_{j=r-A+1}^{r}(K+j)\inparen{\frac34}^{(K+j)/2}} \\
&\leq O(1)\cdot \sum_{j=0}^{r}(K+j)\inparen{\frac34}^{(K+j)/2} = O(1)\cdot \sum_{j=K}^{K+r} j\inparen{\frac34}^{j/2} \\
&\leq O(1)\cdot \sum_{j=K}^{+\infty} j\inparen{\frac34}^{j/2} = O\inparen{K\cdot\inparen{\frac34}^{K/2}},
\end{align*}

and therefore
$$\mathbb{E}(W_r) = \sum_{K=1}^{+\infty} P(W_r\geq K) = O\inparen{1}.$$
\end{proof}

\medskip 
\noindent {\bf Bounding the number of forked rounds.} For completeness we prove that there can be only a small, finite number of publicly forked rounds for which the latency could be $O(\log N)$, and for the remaining rounds (the default case) the latency is $O(1)$ as in Lemma~\ref{lemma:quick-choose-head}.

Note that according to the Definition~\ref{def:locally-forked}, when a dishonest node creates a single fork and shows it to one (or a small number) of honest nodes, this does not yet make the corresponding round publicly forked.
Indeed, to increase the protocol latency the adversary must show the fork to the majority of honest nodes, and it must happen during the next $4$ rounds.
This definition captures a certain trade-off inherent in all possible forking strategies, namely if the attack is to increase the protocol latency, it cannot stay hidden from the honest parties.

\begin{definition2}\label{def:discovered-forker}
We say that node $i$ is a \emph{discovered forker}, if there are two variants of the same unit $U, U'$ created by $i$, and there is a round $r$, such that every round-$r$ unit is above $U$ and $U'$.
\end{definition2}

\noindent
Being a discovered forker has two straightforward consequences.
First, no unit of round $r$ or higher can connect directly to units created by the forker, hence the node $i$ is effectively banned from creating units of round higher than $r-2$.
Moreover, the round-$r$ head unit is also above the fork, and thus we can easily introduce some consensus-based mechanisms of punishing the forker.

\begin{lemma}\label{lemma:single-unit-spread}
Let $r\geq 0$ be a round number. Among all $2f+1$ round-$r$ units created by honest nodes there is at least one unit $U_0$ such that at least $f+1$ units of round $r+1$ created by honest nodes are above $U_0$, and every unit of round $r+2$ (ever created) is above $U_0$.
\end{lemma}

\begin{proof}
Recall that units created by honest nodes cannot be forked.
Every honest unit of round $r+1$ has at least $2f+1$ parents created by distinct nodes, hence it has at least $f+1$ parents created by honest nodes.
By the pigeonhole principle, there exists at least one round-$r$ unit $U_0$, created by an honest node, such that it is a parent of at least $f+1$ units  of round $r+1$ created by honest nodes (denote these units by $\mathcal{U}$).
Now, every unit $U$ of round $r+2$ has at least $2f+1$ parents created by distinct nodes, hence must have at least one parent in $\mathcal{U}$, which implies $U\geq U_0$.
\end{proof}

\begin{lemma}\label{lemma:how-many-forked-rounds}
Whenever a a round is publicly forked then during the next $7$ rounds at least one forker is discovered. 
Consequently, if $f'$ denotes the total number of discovered forkers then at most $7f' = O(N)$ rounds of the protocol can be publicly forked.
\end{lemma}

\begin{proof}
Let $U_0$ be the round-$(r+5)$ unit whose existence is proved in Lemma~\ref{lemma:single-unit-spread}.
Then $U_0$ is above at least one of $f+1$ fork witnesses, and therefore every unit of round $r+7$ is above a fork created by some node $i$ (of course, this could happen earlier).
Since no unit of round $r+7$ or higher can connect directly to units created by $i$, then no unit of round higher than $r+5$ created by node $i$ can be present in the \chdag{}.
Indeed, every such unit must have children to be added to the local copy held by nodes other than $i$, and those children units would be invalid.
Therefore the node $i$ can create forks at six rounds, from $r$ to $r+5$, and these rounds can possibly be publicly forked, hence the $6f'$ upper bound.
Note that node $i$ could have forked at some publicly forked rounds with number less than $r$, but then some other malicious node must have been discovered.
\end{proof}

\subsection{\quickalef -- The Optimistic Case}\label{subsec:synchronous}

Here, we conduct an analysis of the \quickalef{} protocol behavior in a partially synchronous scenario.
It shows that the protocol apart from being capable of operating with constant latency in an asynchronous environment (as proved in the previous subsection), but also matches the validation time achieved by state-of-the-art synchronous protocols.

Since the protocol itself does not operate on the network layer, but on \chdag{},
for sake of this subsection we need to translate partial synchrony to the language of \chdag{}s.

\begin{definition2}
A \chdag{} is synchronous at round $r$ if every unit of round $r$ produced by an honest node has units of round $r-1$ produced by all honest nodes as its parents. 
\end{definition2}
\noindent 
Note that in case of synchronous network the above condition can be enforced by restricting each honest node to wait $\Delta$ before producing the next unit.
Since the protocol is secure even with no synchrony assumptions, the delay $\Delta$ can be adaptively adjusted to cope with partial synchrony in a similar manner as, say, in Tendermint.

\begin{theorem}
If $\cD$ is a \chdag{} that is  synchronous in rounds $r+1$ and $r+2$ and $\msf{DefaultIndex}(r,\cD)$ is honest, then any unit of round $r+3$ is enough to compute the head at round $r$.
\end{theorem}

\begin{proof}
Let $i$ be the $\msf{DefaultIndex}$ at round $r$ and $U$ be the unit created by $\cP_i$ at round $r$.
Since the \chdag{} is synchronous at round $r+1$, all honest nodes create a unit that is above $U$ and thus vote $1$ on $U$.
Further, again by synchrony, each unit $V$ of round $r+2$ created by an honest node is above at least $2f+1$ units that are above $U$, hence $V$ "sees" $2f+1$ votes on $1$ for $U$ and since the common vote at this round is $1$ we can conclude that
$$\msf{UnitDecide}(U, V, \cD)=1.$$
\noindent 
Finally, at round $r+3$ the unit $U$ is returned as the first proposal for the head by the $\msf{GeneratePermutation}$ function\footnote{It may seem that the decision could have been made at round $r+2$ already. Such an ``optimization" though breaks safety of the protocol in case of a malicious proposer. Indeed, if the proposer forks the proposed unit, some node can decide it as head in round $r+2$, while another node could miss this variant and instead decide a different variant on $1$ later during the protocol execution. This cannot happen at round $3$ by Lemma~\ref{lemma:negative-decisions}} and thus 
$$\msf{ChooseHead}(r, \cD) = U.$$
\end{proof}

\section{Transaction arrival model and a proof of Theorem~\ref{thm:AB}}\label{sec:proofs-ab}

We are interested in the strongest possible form of protocol latency, i.e., 
not only measuring time between the moment a transaction is placed in a unit and the moment when it is commonly output as a part of total ordering but instead to start the ``timer" already when the transaction is input to a node.
For this, we introduce the following model. 
\smallskip

\noindent {\bf Transaction arrival model.}
We assume that nodes continuously receive transactions as input and store them in their unbounded buffers.
We make the following assumption regarding the buffer sizes:
\begin{definition2}[Uniform buffer distribution]
The ratio between buffer sizes of two honest nodes at the time of creation of their respective units of a given \dagr{} is bounded by a constant $C_B\geq 1$.
\end{definition2}
\noindent 
While it may look as a limiting assumption, it is not hard to ensure it in a real-life scenario by requiring each transaction to be sent to a randomly chosen set of nodes. 
An adversary could violate this assumption by inputting many transactions to the system in a skewed manner (i.e., sending them only to several chosen nodes), but such attack would be expensive to conduct (since usually a small fee is paid for each transaction) and hence each node with such artificially inflated buffer could be considered Byzantine. 
\smallskip

\noindent{\bf Including transactions in units.} Assuming such a model, each honest node randomly samples $\frac{1}{N}$ transactions from its buffer to be included in each new unit. 
Additionally, it removes transactions that are already included in units of its local copy of \chdag{} from its buffer. 
\smallskip

\noindent {\bf DAG-rounds and async-rounds.}
There are two different notions of a ``round'' used throughout the paper -- \dagr{} as defined in Section~\ref{sec:asynch-as-dag} and \asynchr{}, as defined by~\cite{CR93} and in Section~\ref{sec:async_round}.
While the \dagr{} is more commonly used in the paper, it should be thought as a structural property of \chdag{}s and hence not necessarily a measure of passing time. 
Hence, to be in line with the literature, all proofs regarding latency are conducted in terms of \asynchr{}s.
The following lemma provides a useful one-sided estimate of unit production speed in terms of \asynchr{}s, roughly it says that \dagr{}s (asymptotically) progress at least as quickly as \asynchr{}s.

\begin{lemma}\label{lemma:unit-production-speed}
Each honest node initializes \chrbc{} for a newly produced unit at least once during each $5$ consecutive \asynchr{}s.
\end{lemma}
\begin{proof}

Let $\cP_i$ be an honest node.
We prove that if $\cP_i$ have instantiated \chrbc{} for unit of \dagr{} $r_d$ at \asynchr{} $r_a$ then it necessarily instantiates it again (for a different unit) before \asynchr{} $r_a+5$ ends.
We show it by induction on $r_d$.

If $r_d=1$, then necessarily $r_a=1$ since every node is required to produce a unit in its first atomic step. 
Then, by Lemma~\ref{lemma:RBC}({\bf Latency}), $\cP_i$ will have at least $2f+1$ units of \dagr{} $1$ in its local copy of \chdag{} by the end of \asynchr{} $4$ and hence will be required to instantiate \chrbc{} at that time, what concludes the proof in this case.

Now assume $r_d>1$.
To produce a unit of \dagr{} $r_d$, there must have been at least $2f+1$ units of \dagr{} $r_d-1$ in $\cD_i$.
Then at least $2f+1$ instances of \chrbc{} for units of that \dagr{} have had to locally terminate for $\cP_i$ before $r_a$.
Hence, by Lemma~\ref{lemma:RBC}({\bf Fast Agreement}), we get that every honest node needs to have at least $2f+1$ units of \dagr{} $r_d-1$ in its local copy of \chdag{} before \asynchr{} $r_a+2$ ends.
Consequently, each honest node is bound to instantiate \chrbc{} for its unit of \dagr{} $r_d$ or higher before \asynchr{} $r_a+2$ ends.
We consider the following two cases.
\smallskip

\noindent Case 1. Every honest node has instantiated \chrbc{} for a unit of \dagr{} $r_d$ before \asynchr{} $r_a+2$ ended. 
Then, every honest node (including $\cP_i$) will locally output the result of this \chrbc{} instances before \dagr{} $r_d+5$ ends, by Lemma~\ref{lemma:RBC}({\bf Latency}).
Hence, $\cP_i$ is bound to instantiate \chrbc{} for a unit of \dagr{} at least $r_d+1$ before \asynchr{} $r_a+5$ ends.

\noindent Case 2. If Case 1 does not hold, then by the above, at least one honest node has instantiated \chrbc{} for a unit of \dagr{} $r'_d>r_d$ before \asynchr{} $r_a+2$ ended. 
Then, again using Lemma~\ref{lemma:RBC}({\bf Fast Agreement}), we conclude that every honest node (including $\cP_i$) needs to hold at least $2f+1$ units of \dagr{} $r'_d-1\geq r_d$ before \asynchr{} $r_a+4$ ends. 
Then $\cP_i$ is bound to instantiate \chrbc{} for a unit of \dagr{} at least $r_d+1$ before \asynchr{} $r_a+4$ ends.
\smallskip

\end{proof}

\noindent 
Next, we formulate a lemma connecting the concept of \asynchr{} and the transaction arrival model described on the beginning of this section.

\begin{lemma}\label{lemma:transaction-input}
If a transaction $tx$ is input to at least $k$ honest nodes during \asynchr{} $r_a$ then it is placed in an honest unit in \asynchr{} $r_a+ O(\frac{N}{k})$.
\end{lemma}
\begin{proof}
By Lemma~\ref{lemma:unit-production-speed} each of the $k$ honest nodes with $tx$ in buffer is bound to produce at least one unit per $5$ \asynchr{}s.
Since the honest nodes are sampling transactions to be input to units independently, the probability of $tx$ being input to at least one unit during any quintuple of \asynchr{}s after $r_a$ is at least 
\[
1-\Big(1-\frac{1}{N}\Big)^{k} > 1-\Big(1-\frac{k}{N}+\frac{k^2}{2N^2}\Big) = \frac{k}{N} - \frac{k^2}{2N^2} > \frac{k}{2N},
\]
\noindent by extended Bernoulli's inequality, since $\frac{k}{N}\leq 1$.

Hence the expected number of \asynchr{}s before $tx$ is included in a unit is $O(\frac{N}{k})$.

\end{proof}
\smallskip
\noindent 
We now proceed to proving two technical lemmas that are required to show that, on average, every transaction input to the system is placed in a unit only a small (constant) number of times.
We begin with a simple probability bound that is then used in the second lemma.

\begin{lemma}\label{lemma:sets}
Let $n, m \in \N$ and suppose that $Q_1, Q_2, \ldots, Q_{n}$ are finite sets with $|Q_i|\geq m$ for every $i=1,2, \ldots, n$. For every $i=1,2, \ldots, n$ define a random variable $S_i \subseteq Q_i$ that is obtained by including independently at random each element of $Q_i$ with probability $\frac{1}{n}$. If $m\geq 48n$ then
$$P\inparen{\min_{T\subseteq [n], |T|=n/3} \abs{\bigcup_{i\in T} S_i} \leq \frac{m}{6}}\leq e^{-n}.$$
\end{lemma}
\begin{proof}
For convenience denote $Q=\bigcup_{i\in [n]}Q_i$.
Let us fix any subset $T\subseteq [n]$ of size $\nfrac{n}{3}$ and denote $S^T =\bigcup_{i\in T} S_i$. For every element $q\in Q$ define a random variable $Y_q \in \{0,1\}$ that is equal to $1$ if $q\in S^T$ and $0$ otherwise. Under this notation, we have
$$\abs{S^T} = \sum_{q\in Q} Y_q.$$
We will apply Chernoff Bound to upper bound the probability of $$\sum_{q\in Q}Y_q\leq \frac{m}{6}.$$
For this, note first that $\{Y_q\}_{q\in Q}$ are independent.
Moreover, we have $\expect{\sum_{q\in Q}Y_q}\geq \frac{m}{3}$ from linearity of expectation. 
Hence from the Chernoff Bound we derive
$$P\inparen{\abs{S^T}\leq \frac{m}{6}}=P\inparen{\sum_{q\in Q} Y_q \leq \frac{m}{6}}\leq e^{-\frac{m}{24}}\leq e^{-2n}.$$
Finally, by taking a union bound over all (at most $2^n$ many) sets $T$ of size $\nfrac{n}{3}$ we obtain
$$P\inparen{\min_{T\subseteq [n], |T|=n/3} \abs{S^T}\leq \frac{m}{6}}\leq 2^n \cdot e^{-2n} = e^{-n}.$$

\end{proof}

\begin{lemma}\label{lemma:transaction-units-relation}
Let $T_R$ be the number of, not necessarily different, transactions (i.e., counted with repetitions) input to units during the first $R=\poly(N)$ rounds of the protocol, and $T'_R$ be the number of pairwise different transactions in these units, then $T_R = O(T'_R + RN)$ except with probability $e^{-\Omega(N)}$.
\end{lemma}

\begin{proof}
For every round $r$ denote by $B_r$ the common (up to a constant) size of this round's buffer of any honest node.
Since every unit contains at most $O\inparen{\nfrac{B_r}{N}}$ transactions in round $r$, we have that the total number of transactions input in units this round is $O(B_r)$ and thus $T_R = O\inparen{\sum_{r=0}^R B_r}.$

From Lemma~\ref{lemma:spread} at every round $r$, there exists a set $\mathcal{S}_r$ of units at round $r$ such that every unit of round at least $r+3$ is above every unit in $\mathcal{S}_r$.
Let also $\mathcal{H}_r\subseteq \mathcal{S}_r$ be any subset of at least $f+1$ honest units and $\mathcal{T}_r$ be the set of all transactions in $\mathcal{H}_r$.
Since for any $r_1, r_2$ such that $r_2\geq r_1+3$ we have $\mathcal{H}_{r_1}\geq \mathcal{H}_{r_2}$ and thus $\mathcal{T}_{r_1} \cap \mathcal{T}_{r_2} = \emptyset$.
In particular
$$T_R'\geq \abs{\bigcup_{r=0}^R \mathcal{T}_r}\geq \frac{1}{3} \sum_{r=0}^R \abs{\mathcal{T}_r}.$$
From Lemma~\ref{lemma:sets} we obtain that unless $B_r < 48N$ it holds that $\abs{\mathcal{T}_r}\geq \frac{B_r}{6}$ with probability $1-e^{-N}$. Thus finally, we obtain that with high probability (no matter how large $B_r$ is)
$\abs{\mathcal{T}_r} \geq \frac{B_r}{6} - 8N,$
and consequently
$$T'_R \geq \frac{1}{3} \sum_{r=0}^R \abs{\mathcal{T}_r} \geq \frac{1}{18}\sum_{r=0}^R B_r - \frac{8}{3}RN,$$
and therefore
$$T_R =O\inparen{\sum_{r=0}^R B_r} =  O(T'_R + RN).$$
\end{proof}

\bigskip

\begin{proof}[Proof of Theorem~\ref{thm:AB}]$ $

\noindent {\bf Total Order} and {\bf Agreement.}
It is enough to show that all the honest nodes produce a common total order of units in \chdag{}.
Indeed, the transactions within units can be ordered in an arbitrary yet deterministic way, and then a consistent total order of units trivially yields a consistent total order of transactions.
The total order is deterministically computed based on the function $\msf{ChooseHead}$ applied to every round, which is consistent by Lemma~\ref{lemma:choose-head}({\bf Agreement}).
\ourskip

\noindent {\bf Censorship Resilience} and {\bf Latency.} By Lemma~\ref{lemma:transaction-input} we obtain that each transaction input to $k$ honest nodes is placed in a unit $U$ created by some honest node $\cP_i$ after $O(\nfrac{N}{k})$ \asynchr{}s.

By Lemma~\ref{lemma:RBC}({\bf Latency}) each honest node receives $U$ via \chrbc{} within $3$ \asynchr{}s from its creation and hence will append it as a parent of its next unit, which will be created within $5$ \asynchr{}s by Lemma~\ref{lemma:unit-production-speed} and added to \chdag{}s of all honest nodes within next $3$ \asynchr{}s, again by Lemma~\ref{lemma:RBC}({\bf Latency}).
Thus, after $O(\frac{N}{k})$, the unit $U$ containing $tx$ is below all maximal units created by honest nodes.
Denote by $r_d$ the maximum $\dagr{}$ over all honest nodes at that moment and by $r_a$ the currenct \asynchr{}.

Since each unit has at least $2f+1$ parents, it follows that every unit of round $r_d+1$ (even created by an dishonest node) is above $U$.
For this reason, the head of round $r_d+1$ is also above $U$ and in particular the unit $U$ is ordered at latest by the time when the head of this round is chosen.
Note that by Lemma~\ref{lemma:choose-head}({\bf latency}), the head at round $r_d+1$ is determined after only $O(1)$ \dagr{}s.
Combining this fact with Lemma~\ref{lemma:unit-production-speed} it follows that after $5\cdot O(1)$ \asynchr{} each honest node can determine the head from its version of the poset.
We obtain that the total latency is bounded by $$O(\nfrac{N}{k}) + O(1)\cdot O(1)= O(\nfrac{N}{k}).$$
\ourskip

\noindent {\bf Communication Complexity.} Let $t_r$ be the number of transactions that have been included in honest units at level $r$.
By Lemma~\ref{lemma:RBC}({\bf Communication Complexity}) we obtain that the total communication complexity utilized by instances of \chrbc{} over any sequence of $R$ \dagr{}s is
\[
    O\inparen{N\sum_{r=1}^Rt_r + RN^2\log N} = O(T_R+ RN^2\log N),
\]
where $T_R$ is the total number of transactions included in units at all rounds $r=0,1, \ldots, R$. From Lemma~\ref{lemma:transaction-units-relation} we conclude that $T_R = O(T'_R + RN)$ where $T'_R$ is the number of distinct transactions in these units. The final bound on communication complexity is
$$O(T'_R+ RN^2\log N).$$

\smallskip

\end{proof}

\section{Randomness Beacon}\label{ssec:proofs-coin}
The main goals of this section are to show that \msf{ABFT}-\msf{Beacon} implements a Randomness Beacon (thus providing a proof of Theorem~\ref{thm:strongcoin}), and to prove correctness of two implementations of $\msf{SecretBits}$ that are provided in Section~\ref{sec:global_coin}, i.e.,  prove Lemma~\ref{lemma:secretbits} and Lemma~\ref{lemma:secretbits2}.

Let us now briefly summarize what is already clear and what is left to prove regarding the randomness beacon.
\begin{enumerate}
    \item Every node $\cP_i$ reliably broadcasts its key box $KB_i$ (in $U[i;0]$) that contains a commitment to a polynomial $A_i$ (of degree at most $f$) and encrypted tossing keys for all the nodes. Some of the tossing keys might be incorrect if $\cP_i$ is dishonest.
    \item For every $i\in [N]$ the unit $U[i;6]$ defines a set $T_i \subseteq [N]$ and the $i$th $\msf{MultiCoin}$ is defined as a a threshold signature scheme as if it were generated using the polynomial
    $$B_i(x) = \sum_{j\in T_i} A_j(x).$$
    \item A threshold signature with respect to $\msf{MutliCoin}_i$, i.e., $m^{B_i(0)}$ for some $m\in G$ can be generated by simply multiplying together the threshold signatures $m^{A_j(0)}$ for $j\in T_i$. 
    \item Thus, on a high level, what remains to prove is that:
    \begin{itemize}
        \item For every $i\in [N]$ such that $U[i;6]$ is available, it is possible for the nodes to commonly generate threshold signatures with respect to key set $KS_j$ for every $j\in T_i$.
        \item For every $i\in [N]$ the key $B_i(0)$ remains secret, or more precisely, the adversary cannot forge signatures $m^{B_i(0)}$ for a randomly chosen $m\in G$.
    \end{itemize}
\end{enumerate}
\noindent Below, we proceed to formalizing and proving these properties.

\noindent{\bf Properties of Key Sets.} 
From now on we will use the abstraction of Key Sets to talk about pairs $(TK, VK)$ where 
\begin{align*}
TK& = \inparen{A(1), A(2), \ldots, A(N)}\\
    VK& = \inparen{g^{A(1)}, g^{A(2)}, \ldots, g^{A(N)}}
\end{align*} for some degree $\leq f$ polynomial $A\in \Z_q[x]$.
In the protocol, the vector $VK$ is always publicly known (computed from a commitment or several commitments) while $TK$ is distributed among the nodes. It can be though distributed dishonestly, in which case it might be that several (possibly all) honest nodes are not familiar with their corresponding tossing key.
In particular, every Key Box defines a key set and every MultiCoin defines a key set.

In the definition below we formalize what does it mean for a key set to be "usable" for generating signatures (and thus really generating randomness)
\begin{definition2}[Usable Key Set]
We say that a Key Set $KS=(TK, VK)$ is \emph{usable} with a set of key holders $T\subseteq  [N]$ if:
\begin{enumerate}
    \item every honest node is familiar with $VK$ and $T$,
    \item $|T|\geq 2f+1$,
    \item for every honest node $\cP_i$, if $i\in T$, then $i$ holds $tk_i$.
\end{enumerate}
\end{definition2}
\noindent 
The next lemma demonstrates that indeed a usable key set can be used to generate signatures.

\begin{lemma}\label{lemma:can-toss-approved-key-set}
Suppose that $(TK, VK)$ is a usable key set, then if the nodes being key holders include the shares for a nonce $m$ in the \chdag{} at round $r$, then every honest node can recover the corresponding signature $\sigma_m$ at round $r+1$.

\end{lemma}
\begin{proof}
Since the key set is usable, we know that there is exactly one underlying polynomial $A\in \Z_q[x]$ of degree $f$ that was used to generate $(TK, VK)$, hence every unit at round $r$ created by a share dealer $i\in T$ can be verified using $\msf{VerifyShare}$ and hence it follows that any set of $f+1$ shares decodes the same value $\sigma_m = m^{A(0)}$ using $\msf{ExtractShares}$ (see also~\cite{Boldyreva03}).
It remains to observe that among a set of (at least $2f+1$) round-$r$ parents of a round-$(r+1)$ unit, at least $f+1$ are created by nodes in $T$ and thus contain appropriate shares.
\end{proof}
\noindent

\ourskip

\noindent{\bf Generating Signatures using MultiCoins.} 
We proceed to proving that the nodes can generate threshold signatures with respect to MultiCoins.
The following structural lemma says, intuitively, that at every round $r$ there is a relatively large set of units that are below {\bf every} unit of round $r+3$.
\begin{lemma}[Spread]\label{lemma:spread}
Assume that an honest node $\cP_i$ has just created a unit $U$ of round $r+3$. 
At this moment, a set $\mathcal{S}_r$ of $2f+1$ units of round $r$ is determined such that for every unit $V$ of round $r+3$ and for any unit $W\in\mathcal{S}_r$, we have $W\leq V$.
\end{lemma}

\begin{proof}
Let $K$ be a set of $2f+1$ nodes that created units in $\down{U}$. 
Let $\mathcal{T}$ be a set of $2f+1$ units of round $r+1$ created by nodes in $K$. 
Every unit in $\down{U}$ is above $2f+1$ units of round $r+1$, hence it is above at least $f+1$ units in $\mathcal{T}$ (the remaining $f$ units may be created by nodes outside of $K$). 
By the pigeonhole hole principle, there exists a unit $U_0\in\mathcal{T}$ such that at least $f+1$ units in $\down{U}$ are above it. 
Thus, every subset of $2f+1$ units of round $r+2$ must include at least one unit above $U_0$, so every valid unit $V$ of round $r+3$ will have at least one parent above $U'$, thus $V\geq U'$. Finally, choose $\mathcal{S}_r := \down{V}$.
\end{proof}

\noindent 
We now proceed to proving a structural lemma that confirms several claims that were made in Subsection~\ref{ssec:informal-setup} without proofs. We follow the notation of Subsection~\ref{ssec:informal-setup}.

\begin{lemma}\label{lemma:multicoin-can-be-tossed}
Let $T_i$ for $i\in [N]$ be the sets of indices of Key Sets forming the $i$th MultiCoin, then
\begin{enumerate}
    \item For every $i$ such that $U[i;6]$ exists, $T_i$ has at least $f+1$ elements.
    \item For every $i\in [N]$ and every $j\in T_i$, the Key Set $KS_j$ is usable.
    \item For every unit $V$ of round $r+1$ with $r\geq 6$, there exists a set of $f+1$ units in $\down{V}$ such that they contain all shares required to reconstruct all Key Sets that are below $V$.
\end{enumerate}
\end{lemma}

\begin{proof}
Let $\cD$ here be the local view of an honest node that has created a unit of round at least $6$.
From Lemma~\ref{lemma:spread} we obtain sets $\mathcal{S}_0$ and $\mathcal{S}_3$, which are in $\cD$ because it includes at least one unit of round $6$.

For any $i\in [N]$ such that $U[i;6]$ exists in $\cD$, the  set $T_i$ is defined as a function of the unit $U[i;6]$ and all units below it, therefore it is consistent between honest nodes.
Additionally, $U[i;6]$ is above at least $f+1$ units from $\mathcal{S}_0$ created by honest nodes ($2f+1$ units minus at most $f$ created by malicious nodes), and creators of those units must be included in $T_i$.
Indeed, if a unit of round zero was created by an honest node, then it is valid, and every round-$3$ unit must vote $1$ on it, as false-negative votes cannot be forged.

From now on, let $T$ denote the union of all $T_i$s that are well defined (i.e., exist) in $\cD$.
Every unit $W$ in the set $\mathcal{S}_3$ must vote $1$ on every Key Box $KB_l$ for $l\in T$, because $W$ is seen by every unit of round $6$. Since $|\mathcal{S}_3|\geq 2f+1$ and due to the definition of votes, we see that every Key Set $KS_l$ for $l\in T$ is usable by (possibly a superset of) creators of $\mathcal{S}_3$.

If we take the intersection of the set of creators of $\mathcal{S}_3$, and the set of $\down{V}$ creators, then the resulting set contains at least $f+1$ nodes. Units in $\down{V}$ created by these $f+1$ nodes must contain valid shares from round $r$ to all Key Sets created by nodes from $T$. Indeed, this is required by the protocol, since these nodes declared at round $3$ that they will produce valid shares (recall that $f+1$ valid shares are enough to read the secret).

Finally we remark that the node who holds $\cD$ might not be familiar with $\mathcal{S}_3$, but, nevertheless, it can find at least $f+1$ units in $\down{V}$ containing all required shares, because it knows $T$.
\end{proof}

\noindent In part (3) of the above lemma we show that we can always collect shares from a single set of $f+1$ units, which yields an optimization from the viewpoint of computational complexity.
The reason is that each round of the setup requires generating $O(N)$ threshold signatures, but if we can guarantee that each signature has shares generated by the same set of $f+1$ nodes, then we may compute the Lagrange coefficients for interpolation only once (as they depend only on node indices that are in the set) thus achieve computational complexity $O(N^3)$ instead of $O(N^4)$ as in the basic solution.

\noindent{\bf Security of MultiCoins.} 
The remaining piece is unpredictability of multicoins, which is formalized below as

\begin{lemma}\label{lemma:multicoin-cannot-be-predicted2}
If a polynomially bounded adversary is able to forge a signature generated using $\msf{MultiCoin}_{i_0}$ (for any $i_0 \in [N]$) for a fresh nonce $m$ with probability $\eps>0$, then there exists a polynomial time algorithm for solving the computational Diffie-Hellman problem over the group $G$ that succeeds with probability $\eps$.
\end{lemma}

\begin{proof}
We start by formalizing our setting with regard to the encryption scheme that is used to encrypt the tossing keys in key boxes.
Since it is semantically secure we can model encryption as a black box.
More specifically we assume that for a key pair $(sk, pk)$, an adversary that does not know $sk$ has to query a black box encryption oracle $\enc_{pk}(d)$ to create the encryption of a string $d$ for this public key.
Further, we assume that whenever a string $s$ is fed into the decryption oracle $\dec_{sk}(\cdot )$ such that $s$ was not obtained by calling the encryption oracle by the adversary, then the output $\dec_{sk}(\cdot )$ is a uniformly random string of length $|s|$.
Since the key pairs we use in the protocol are dedicated, the adversary has no way of obtaining ciphertexts that were not created by him, for the key sets he is forced to use.

Given this formalization we are ready to proceed with the proof.
We fix any $i_0 \in [N]$ and consider the $\msf{MultiCoin}_{i_0}$.
We use a simulation argument (see e.g. \cite{PS96}) to show a reduction from forging signatures to solving the Diffie-Hellman problem.
To this end, let $(g^\alpha, g^\beta)$ be a random input to the Diffie-Hellman problem, i.e., the goal is to compute $g^{\alpha\beta}$. 

Suppose that an adversary can forge a signature of a random element $h\in G$ with probability $\eps>0$ when the honest nodes behave honestly.
We now define a {\bf simulated run} of the protocol (to distinguish from an {\bf honest run} in which all honest node behave honestly) which is indistinguishable from an honest run from the perspective of the adversary, yet we can transform a successful forgery into a solution of the given Diffie-Hellman triplet.

\noindent {\bf Description of the simulated run.} For convenience assume that the nodes controlled by the adversary have indices $1,2, \ldots, f$. 
At the very beginning, every honest node $i\in \{f+1, f+2, \ldots, N\}$ generates uniformly random values $r_{i,0}, r_{i,1}, \ldots, r_{i,f} \in \Z_q$ and defines a polynomial $A_i$ of degree at most $f$ by specifying its values at $f+1$ points as follows:
\begin{align*}
    A_i(0) &= \alpha + r_{i,0},\\
    A_i(1) &= r_{i,1},\\
    &\vdots\\
    A_i(f) &= r_{i,f},
\end{align*}
where $\alpha \in Z_q$ is the unknown value from the Diffie-Hellman instance.
Note that since $\alpha$ is unknown the $i$th node has no way to recover the polynomial $A_i(x) = \sum_{j=0}^f a_{i,j}x^j$ but still, it can write each of the coefficients in the following form
$$a_{i,j} = a_{i,j,0} + \alpha a_{i,j,1} \ \ \ \ \mbox{for }j=0,1, \ldots, f,$$
where $a_{i,j,0}, a_{i,j,1}\in \Z_q$ are known values.
Since $g^\alpha$ is known as well, this allows the $i$th node to correctly compute the commitment
$$C_i = \inparen{g^{a_{i,0}}, g^{a_{i,1}}, \ldots, g^{a_{i,f}}},$$
pretending that he actually knows $A_i$.
Similarly he can forge a key box $KB_i$ so that from the viewpoint of the adversary it looks honest.
This is the case because $\cP_i$ knows the tossing keys for the adversary $tk_1, tk_2, \ldots, tk_f$ since these are simply $r_{i,1}, r_{i,2}, \ldots, r_{i,f}$.
Thus $\cP_i$ can include in $E_i$ the correct values $e_{i,1}, e_{i,2}, \ldots, e_{i,f}$ and simply fill the remaining slots with random strings of appropriate length (since these are meant for the remaining honest nodes which will not complain about wrong keys in the simulation).

In the simulation all the honest nodes are requested to vote positively on key boxes created by honest nodes, even though the encrypted tossing keys are clearly not correct.
For the key boxes dealt by the adversary the honest nodes act honestly, i.e., they respect the protocol and validate as required.
Suppose now that the setup is over and that the honest nodes are requested to provide shares for signing a particular nonce $m$.
The value of the hash $h=\hashg(m) \in G$ is sampled in the simulation as follows: first sample a random element $\delta_m \in \Z_q$ and subsequently set
$$\hashg(m) = g^{\delta_m}.$$
Since in the simulation the value $\delta_m$ is revealed to the honest nodes, the honest node $\cP_j$ can compute its share corresponding to the $i$th key set (for $i\geq f+1$) as
$$h^{A_i(j)} = \inparen{g^{\delta_m}}^{A_i(j)} = \inparen{g^{A_i(j)}}^{\delta_m},$$
where the last expression can be computes since $g^{A_i(j)}$ is $vk_{i,j}$ (thus computable from the commitment) and $\delta_m$ is known to $j$.

\noindent {\bf Solving Diffie-Hellman from a forged signature.} After polynomially many rounds of signing in which the hash function is simulated to work as above, we let the hash function output $g^{\beta}$ (the element from the Diffie-Hellman input) for the final nonce $m'$ for which the adversary attempts forgery.
The goal of the adversary is now to forge the signature, i.e., compute $\inparen{g^{\beta}}^{A(0)},$ where
$$A(0) = \sum_{i\in T_{i_0}} A_i(0),$$
is the cumulative secret key for the multicoin $\msf{MultiCoin}_{i_0}$.
For that, the adversary does not see any shares provided by honest nodes.
Note that if we denote by $H$ and $D$ the set of indices of honest and dishonest nodes (respectively) within $T_{i_0}$ then the value of $A(0)$ is of the following form
\begin{align*}
    A(0) &= \sum_{i\in T_{i_0}} A_i(0)\\
    &= \sum_{i\in H} A_i(0) + \sum_{i\in D} A_i(0)\\
    &= |H| \cdot \alpha + \sum_{i\in H } r_{i,0} + \sum_{i\in D i} A_i(0).
\end{align*}
We observe that each value $A_i(0)$ for $i\in D$ can be recovered by honest nodes, since at least $f+1$ honest nodes must have validated their keys when voting on $KB_i$ and thus honest nodes collectively know at least $f+1$ evaluations of $A_i$ and thus can interpolate $A_i(0)$ as well.
Similarly, the values $r_{i,0}$ for $i\in H$ are known by honest nodes.
Finally we observe that since by Lemma~\ref{lemma:multicoin-can-be-tossed} (a), $|T_i|\geq f+1$, it follows that $|H|>0$.
Thus consequently, we obtain that $A(0)$ is of the form
$$A(0) = \gamma \alpha + \mu,$$
where $\gamma, \mu \in \Z_q$ are known to the honest nodes and $\gamma \neq 0$.
Hence, the signature forged by the adversary has the form $$\inparen{g^{\beta}}^{A(0)} = g^{\gamma \alpha \beta + \beta \mu} = \inparen{g^{\alpha \beta}}^{\gamma} \cdot \inparen{g^{\beta}}^{\mu},$$
and can be used to recover $g^{\alpha \beta}$ since $\gamma, \mu$ and $g^{\beta}$ are known.

The above reasoning shows that we can solve the Diffie-Hellman problem given an adversary that can forge signatures; however to formally conclude that, we still need to show that the view of the adversary in such a simulated run is indistinguishable from an honest run.

\noindent {\bf Indistinguishability of honest and simulated runs.}
We closely investigate all the data that is emitted by honest nodes throughout the protocol execution to see that these distributions are indistinguishable between honest and simulated runs. 
Below we list all possible sources of data that adversary observes by listening to honest nodes.
\begin{enumerate}
    \item The key boxes in an honest run and a simulated run look the same to the adversary, as they are always generated from uniformly random polynomials, and the evaluations at points $1,2,\ldots, f$ are always honestly encrypted for the nodes controlled by the adversary. The remaining encrypted keys are indistinguishable from random strings and thus the adversary cannot tell the difference between honest runs and simulated runs, since the values encrypted by honest nodes are never decrypted.
    \item The votes (i.e., the $\msf{VerKey}(\cdot, \cdot)$ values) of honest nodes on honest nodes are always positive, hence the adversary does not learn anything here.
    \item The shares generated by honest nodes throughout the execution of the protocol are always correct, no matter whether the run is honest or simulated (in which case the honest nodes do not even know their keys).
    \item The votes of honest nodes on dishonest key boxes: this is by far the trickiest to show and the whole reason why we need dedicated key pairs for encryption in key boxes.
    We argue that the adversary can with whole certainty predict the votes of honest nodes, hence this piece of data is insignificant: intuitivily the adversary cannot discover that he "lives in simulation" by inspecting this data.
    Here is the reason why the adversary learns nothing from these votes.
    There are two cases: 
    \begin{enumerate}
        \item a dishonest node $i\in [f]$ includes as $e_{i,j}$ some ciphertext that was generated using the encryption black box $\enc_{i\to j}$. That means that the adversary holds a string $d$ such that $e_{i,j}=\enc_{i \to j}(d)$. In this case the adversary can itself check whether $d$ is the correct tossing key. Thus the adversary can predict the vote of $\cP_j$: either it accepts if $d$ is the correct key, or it rejects in which case it reveals $\dec_{i \to j}(e_{i,j})=d$.
        \item a dishonest node $i\in [f]$ includes as $e_{i,j}$ some string that was never generated by the encryption black box $\enc_{i\to j}$.
        In this case, the decryption of $e_{i,j}$ is simply a random string and thus agrees with the key only with negligible probability.
        Thus the adversary can tell with whole certainty that the vote on $e_{i,j}$ will be negative and he only learns a random string $\dec_{i \to j}(e_{i,j})$.
        \end{enumerate}
\end{enumerate}
All the remaining data generated by honest nodes throughout the execution of the protocol is useless to the adversary as he could simulate it locally.
This concludes the proof of indistinguishability between honest and simulated runs from the viewpoint of the adversary.
\end{proof}

\noindent{\bf Proofs regarding SecretBits and ABFT-Beacon.} 
We are now ready to prove Lemma~\ref{lemma:secretbits2}, i.e. show that the implementation of  $\msf{SecretBits}$ with MultiCoin during the $\msf{ABFT}-\msf{Beacon}$ setup satisfies the Definition~\ref{def:secret-bits} of $\msf{SecretBits}$.
\begin{proof}[Proof of Lemma~\ref{lemma:secretbits2}]
Let $\cP_k$ be any honest node. 
We start by showing that during the run of the $\msf{ABFT}-\msf{Beacon}$ setup if $\cP_k$ invokes $\msf{SecretBits}(i,r)$, then $r\geq 10$ and the unit of round $6$ created by $\cP_i$ is already present in the local copy $\cD_k$.

All $\msf{SecretBits}$ invocations happen inside $\msf{ChooseHead}(6)$ or, more precisely, in $\msf{GeneratePermutation}(6,\cD)$ and $\msf{CommonVote}(U,\cdot)$, for units $U$ of round $6$.
Let $\mathcal{U}$ be a set of units of round $6$ present in $\cD_k$ at the time of any particular call to $\msf{ChooseHead}(6)$.
Additionally, let $\mathcal{C}$ be the set of all nodes that created units in $\mathcal{U}$. 
The procedure $\msf{GeneratePermutation}(6,\cD)$ calls $\msf{SecretBits}(i, 6+4)$ for all $i\in\mathcal{C}$, and function $\msf{CommonVote}(U,V)$ is deterministic if $\H(V)\leq\H(U)+4$ and hence calls only $\msf{SecretBits}(i, k)$, where $i$ is $U$'s creator in $\mathcal{C}$ and $k\geq 10$.

Thus, from the above, the function $\msf{SecretBits}$ will be invoked sufficiently late, hence all units of round $6$ defining required MultiCoins will be already visible at that time.
Lemma~\ref{lemma:multicoin-can-be-tossed} guarantees that then the MultiCoin can be used to generate signatures since all the keys sets comprising it are usable. 
The secrecy follows from the fact that the signatures generated by MultiCoins are unpredictable for the adversary (Lemma~\ref{lemma:multicoin-cannot-be-predicted2}).
\end{proof}
\noindent 
The above proof justifies why is it easier to implement $\msf{SecretBits}$ as a function of two arguments, instead of only using the round number.

The next lemma shows that by combining key sets as in our construction, we obtain a single key set that is secure, i.e., that there are enough nodes holding tossing keys and that the generated bits are secret.

\noindent Given the above lemma, we are ready to conclude the proof of Theorem~\ref{thm:strongcoin}.
\begin{proof}[Proof of Theorem~\ref{thm:strongcoin}]

From Lemma~\ref{lemma:secretbits2} we obtain that the Setup phase of $\msf{ABFT}-\msf{Beacon}$ executes correctly and a common Key Set $(TK, VK)$ is chosen as its result.
Moreover, by Lemma~\ref{lemma:multicoin-cannot-be-predicted2} the output Key Set is secure and thus can be used to generate secrets. 
Consequently, $\msf{Setup}$ + $\msf{Toss}$ indeed implement a Randomness Beacon according to Definition~\ref{def:randomnessbeacon}.

\noindent{\bf Communication Complexity.}
What remains, is to calculate the communication complexity of the protocol. Observe that during the setup, every unit might contain at most:
\begin{itemize}
    \item $O(N)$ hashes of parent nodes,
    \item $O(N)$ verification keys,
    \item $O(N)$ encrypted tossing keys,
    \item $O(N)$ votes,
    \item $O(N)$ secret shares.
\end{itemize}
Since $\msf{ChooseHead}$ builds only an expected constant number of \dagr{}s before terminating (see Lemma~\ref{lemma:choose-head}), every node reliably broadcasts $O(1)$ units during the setup in expectation each of which of size $O(N)$, thus the total communication complexity per node is $O(N^2 \log N)$.
\ourskip

\noindent{\bf Latency.}
Assume that honest nodes decide to instantiate $\msf{Toss}$ for a given nonce $m$ during \asynchr{} $r_a$.
Then, each honest node multicasts its respective share before \asynchr{} $r_a+1$ ends, thus all such shares are bound to be delivered before the end of \asynchr{} $r_a+2$.
Thus, each honest node may reveal the value of a \msf{Toss} after at most $2$ \asynchr{}, what concludes the proof.
\end{proof}
\noindent 
The last remaining piece is to prove Lemma~\ref{lemma:secretbits} that claims correctness of the main implementation of $\msf{SecretBits}$.

\begin{proof}[Proof of Lemma~\ref{lemma:secretbits}]
The case with a trusted dealer follows for instance from~\cite{Boldyreva03}.
Consider now the harder case: when the $\msf{Setup}$ of $\msf{ABFT}-\msf{Beacon}$ is used to deal the keys.
Then, the nodes end up holding a Key Set that is secure by~Lemma~\ref{lemma:multicoin-cannot-be-predicted2} and can be used for signing messages by Lemma~\ref{lemma:can-toss-approved-key-set}.
\end{proof}

\section{Reliable Broadcast}\label{sec:rbc}

In this section we describe \chrbc{} -- a modified version of Reliable Broadcast (RBC) protocol introduced by~\cite{Bracha87} and improved by~\cite{CT05} by utilizing erasure codes.
The protocol is designed to transfer a message from the sender (who instantiates the protocol) to the set of receiving nodes in such a way that receivers can be certain that they all received the same version of the message. 
More formally, a protocol implements Reliable Broadcast if it meets the following conditions:
\begin{enumerate}
    \item {\bf Termination.} If an honest node instantiates RBC, all honest nodes eventually terminate.
    \item {\bf Validity.} If an honest node instantiates RBC with a message $m$, then any honest node can only output $m$.
    \item {\bf Agreement.} If an honest node outputs a message $m$, eventually all honest nodes output $m$.
\end{enumerate}

We implement four rather straightforward modifications to the original RBC protocol. 
First, before multicasting \msf{prevote} message each node checks the size of the proposed unit (line $10$) to prevent the adversary from sending too big units and hence consuming too much bandwidth of honest nodes\footnote{Since in this version of RBC nodes learn about the content of broadcast message only after multicasting it as \msf{prevote}, without this modification adversary could propose arbitrarily big chunks of meaningless data and hence inflate message complexity arbitrarily high.}.
Second, before multicasting prevote, we wait until \chdag{} of a node reaches at least one round lower then round of prevoted unit (line $11$).
This check counters an attack in which adversary initializes \chrbc{} for meaningless pieces of data of rounds which do not exist in the \chdag{} yet, hence wasting bandwidth of honest nodes. 
Third, after reconstructing a unit $U$, we check if $U$ is valid (line $16$), i.e., we check if it has enough parents, if one of its parents was created by $U$'s creator, and if the data field of $U$ contains all required components.
Finally, before multicasting $\msf{commit}$ message for unit $U$, nodes wait until they receive all parents of $U$ via RBC (line $17$) to ensure that it will be possible to attach $U$ to \chdag{}\footnote{Note that this condition is not possible to check before multicasting \msf{prevote} since nodes are not able to recover $U$ at that stage of the protocol.}.

\renewcommand{\algorithmcfname}{\chrbc{}}
\begin{algorithm-hbox}
    \SetKwFunction{FMain}{$\msf{Main}$}
  \SetKwProg{Fn}{ }{:}{}
  \tcc{(Protocol for node $\cP{}_i$ with sender $\cP{}_{s}$ broadcasting unit $U$ of round $r$)}
    \SetKwProg{Upon}{upon}{ do}{end}
    \If{$\cP{}_i=\cP{}_{s}$}
        {
        $\{s_j\}_{j\in N}\!\!\leftarrow\!$ shares of $(f\!\!+\!1,N)$-erasure coding of $U$\\
        $h\leftarrow$ Merkle tree root of $\{s_j\}_{j\in N}$\\
        \For{$j\in N$}
            {
            $b_i\leftarrow$  Merkle branch of $s_j$\\
            {\bf send} $\msf{propose}(h,b_j,s_j)$ to $\cP_j$\\
            }
        }
    \Upon{receiving $\msf{propose}(h,b_j,s_j)$ from $\cP_{s}$}
        {
        \If{$\msf{received\_propose}(\cP_i,r)$}{{\bf terminate}}
        \If{$\msf{check\_size}(s_j)$}
            {
            {\bf wait} until $\cD_i$ reaches round $r-1$\\
            {\bf multicast} $\msf{prevote}(h,b_j,s_j)$
            }
        $\msf{received\_propose}(\cP_i,r) = \msf{True}$
        }    
    \Upon{receiving $2f+1$ valid $\msf{prevote}(h,\cdot,\cdot)$}
        {
        {\bf reconstruct} $U$ from $s_j$\\
        \lIf{$U$ is not valid}
            {
            {\bf terminate}
            }
        {\bf wait} until all $U$'s parents are locally output by $\cP_i$\\
        {\bf interpolate} all $\{s'_j\}_{j\in N}$ from $f+1$ shares\\
        $h'\leftarrow$ Merkle tree root of $\{s'_j\}_{j\in N}$\\
        \If{$h\!=\!h'$ and $\msf{commit}(\cP_{s},r,\cdot)$ has not been send}
            {
            {\bf multicast} $\msf{commit}(\cP_{s},r,h)$\\
            }
        }
    \Upon{receiving $f+1$ $\msf{commit}(\cP_{s},r,h)$}
        {
        \If{$\msf{commit}(\cP_{s},r,\cdot)$ has not been send}
            {
            {\bf multicast} $\msf{commit}(\cP_{s},r,h)$
            }
        }
    \Upon{receiving $2f+1$ $\msf{commit}(\cP_{s},r,h)$}
        {
        {\bf output} $U$ decoded from $s_j$'s
        }
    \BlankLine
  \SetKwFunction{FMain}{$\msf{check\_size}$}
  \SetKwProg{Fn}{ }{:}{}
  \Fn{\FMain{$s_j$}}
    {
    $T\leftarrow$ number of transactions in $U$ inferred from $s_j$-th size\\
    \tcc{$C_B$ is a global parameter of the protocol}
    \If{$T>C_B \cdot$($i$-th batch size in round $r$)}
        {
        {\bf output} \msf{False}
        }
    \lElse
        {
        {\bf output} \msf{True}
        }
    }
\caption{ }
\end{algorithm-hbox}

The following lemma summarizes the most important properties of \chrbc{} protocol. Note that {\bf Latency} is a strictly stronger condition then {\bf Termination}, and similarly {\bf Fast Agreement} is stronger then {\bf Agreement}.
Thus, in particular, this lemma proves that \chrbc{} implements Reliable Broadcast.

\begin{lemma}\label{lemma:RBC}(reliable broadcast)
For any given round $r$, all \chrbc{} runs instantiated for units of that round have the following properties:
\begin{enumerate}
    \item {\bf Latency.} \chrbc{} instantiated by an honest node terminates within three \asynchr{}s for all honest nodes.
    \item {\bf Validity} If an honest node outputs a unit $U$ in some instance of \chrbc{}, $U$ had to be proposed in that instance.
    \item {\bf Fast Agreement.} If an honest node outputs a unit $U$ at \asynchr{} $r_a$, all honest nodes output $U$ before \asynchr{} $r_a+2$ ends,
    \item {\bf Message Complexity.} In total, all the \chrbc{} instances have communication complexity\footnote{We remark that $N^2 \log(N)$ can be improved to $N^2$ in the communication complexity bound by employing RSA accumulators instead of Merkle Trees in the RBC protocol.} at most $O(t_r+N^2\log N)$ per node, where $t_r$ is the total number of transactions input in honest units of round $r$.
\end{enumerate}
\end{lemma}

\begin{proof} 
{\bf Latency.} Let $\cP_i$ be an honest node instantiating \chrbc{} for unit $U_i$ during the \asynchr{} $r$ so that the \msf{propose} messages are sent during $r$ as well. 
Each honest node during the atomic step in which it receives the \msf{propose} message multicasts the corresponding \msf{prevote}.
Since proposals were sent during \asynchr{} $r$, \asynchr{} $r+1$ cannot end until all of them are delivered, all honest nodes multicast prevotes during \asynchr{} $r+1$. 
Similarly, all prevotes sent by honest nodes need to reach their honest recipients during \asynchr{} $r+2$, hence honest nodes multicast their commits during \asynchr{} $r+2$\footnote{A cautious reader probably wonders about the {\bf wait} in line $16$. It will not cause an execution of \chrbc{} instantiated by an honest node to halt since before sending proposals during \asynchr{} $r_a$ a node must have had all $U$'s parents output locally by other instances of \chrbc{}, hence each other node will have them output locally by the end of \asynchr{} $r_a+2$, by condition {\bf Fast Agreement}}.
Finally, the commits of all honest nodes reach their destination before the end of \asynchr{} $r+3$ causing all honest nodes to locally terminate at that time, what concludes the proof. 
\smallskip

\noindent {\bf Validity.} The first honest node to produce the \msf{commit} message for a given unit $U$ has to do it after receiving $2f+1$ \msf{precommit} messages for $U$, out of which at least $f+1$ had to be multicast by other honest nodes.
Since honest nodes are sending prevotes only for shares they received as proposals, $U$ had to be proposed by node initiating \chrbc{}.
\smallskip 

\noindent {\bf Fast Agreement.} Assume that an honest node $\cP_i$ has output a unit $U$ at \asynchr{} $r_a$ during some instance of \chrbc{}.
Since it can happen only after receiving $2f+1$ commit messages for $U$, at least $f+1$ honest nodes had to multicast \msf{commit} for $U$.
Each honest node engages in \chrbc{} for only one version of each node's unit for a given \dagr{}, so it is not possible for any other node to output a unit different than $U$ for that instance of \chrbc{}.

On the other hand, $\cP_i$ must have received at least $f+1$ commit messages for $U$, which have been multicast before \asynchr{} $r_a$ ended.
By definition of \asynchr{}, all of these commits must have reached their honest recipients before \asynchr{} $r_a+1$ ended.
Consequently, after gathering $f+1$ commits for $U$, all honest multicasted their commits before \asynchr{} $r_a+1$ ended.
Thus, each honest node received commits from every other honest node before the end of \asynchr{} $r_a+2$, and terminated outputting $U$.
\smallskip

\noindent {\bf Message Complexity.} Each honest node engages in at most $N$  \chrbc{} instances for each \dagr{} $r$, since it accepts only one proposal from each node.
Let now $\cP_i$ be an honest node with buffer of size $B_r$.
Due to the condition in line $10$, $\cP_i$ engages only in \chrbc{} instances broadcasting units of size at most $C_B\frac{B_r}{N}$, where $C_B$ is the constant from the {\it uniform buffer distribution} assumption.

A single instance of \chrbc{} proposing a unit $U$ has communication complexity of $O(\mathrm{size}(U)+N\log N)$, see~\cite{CT05} (here the communication complexity is computed {\it per honest node}).
Now, because of the {\it unifform buffer distribution} property we know that every valid unit of round $r$ includes between $C_B^{-1}\frac{ B_r}{N}$ and $C_B \frac{B_r}{N}$ transactions (with $C_B=\Theta(1)$) thus if $U$ is such a unit, we have
$$\mathrm{size}(U) = O\inparen{\frac{B_r}{N} + N},$$
\noindent 
where the $+N$ comes from the fact that such a unit contains $O(N)$ parent hashes.
Consequently the total communication complexity in round $r$ is upper bounded by

\begin{align*}
    \sum_{i=1}^N \big(\mathrm{size}(U[i;r])+N\log N\big) &=N\cdot O\inparen{\frac{B_r}{N} + N + N \log N}\\
    &= O(B_r+N^2 \log N)\\
    &= O(t_r+N^2 \log N).
\end{align*}
\noindent 
Where the last equality follows because the total number of transactions in honest units is $t_r = \nfrac{2N}{3} \cdot \Theta \inparen{\nfrac{B_r}{N}} = \Theta \inparen{B_r}$.

\end{proof}

\noindent {\bf General RBC with proofs of termination.} Note that the above scheme differs from classical RBC only in four mentioned \chdag{}-related aspects, hence after deleting additional checks it can be used to broadcast arbitrary pieces of data. 
Additionally, one simple modification can be done, which provides a proof that RBC has terminated for each node which had received the output locally.
Namely, together with \msf{commit} message, each node can send share of a threshold signature of hash $h$ with threshold $2f+1$.
Then, each node that gathered $2f+1$ is able to construct the threshold signature and use it as a proof of the fact that RBC had finalized at its end.

\section{Asynchronous Model}\label{sec:async_round}

Slightly informally, the difference between synchronous an asynchronous network is typically explained as the absence of a global bound $\Delta$ on all message delays in asynchronous setting.
In this paper, we use the following formalization adopted from Canetti and Rabin~\cite{CR93}, where the adversary has a full control over the network delays, but is committed to eventually deliver each message to its recipient. 

\begin{definition2}[Asynchronous Network Model]
The execution of the protocol proceeds in consecutive \emph{atomic steps}. 
In each atomic step the adversary chooses a single node $\cP_i$ which is then allowed to perform the following actions:
\begin{itemize}
    \item read a subset $S$, chosen by the adversary, among all its pending messages. Messages in $S$ are considered \emph{delivered},
    \item perform a polynomially bounded computation,
    \item send messages to other nodes; they stay pending until delivered. 
\end{itemize}
\end{definition2}
\noindent 
Thus, in every atomic step one of the nodes may read messages, perform computations, and send messages. 
While such a model may look surprisingly single-threaded, note that an order in which nodes act is chosen by an adversary, hence intuitively corresponds to the worst possible arrangements of network delays.
The following definition by~\cite{CR93} serves as a main measure of time for the purpose of estimating latency of our protocols. 

\begin{definition2}[Asynchronous Round]
Let $l_0$ be an atomic step in which the last node is chosen by adversary to perform an action, i.e., after $l_0$ each node acted at least once.
The asynchronous rounds of the protocol and the subsequent $l$s are defined as follows:
\begin{itemize}
    \item All atomic steps performed after $l_{i-1}$ until (and including) $l_i$ have asynchronous round $i$,
    \item $l_i$ is the atomic step in which the last message sent during asynchronous round $l_{i-1}$ is delivered.
\end{itemize}
\end{definition2}

\section{Fork bomb attack}\label{sec:fork-bomb}

\begin{figure}
    \centering
    \includegraphics[scale=0.62]{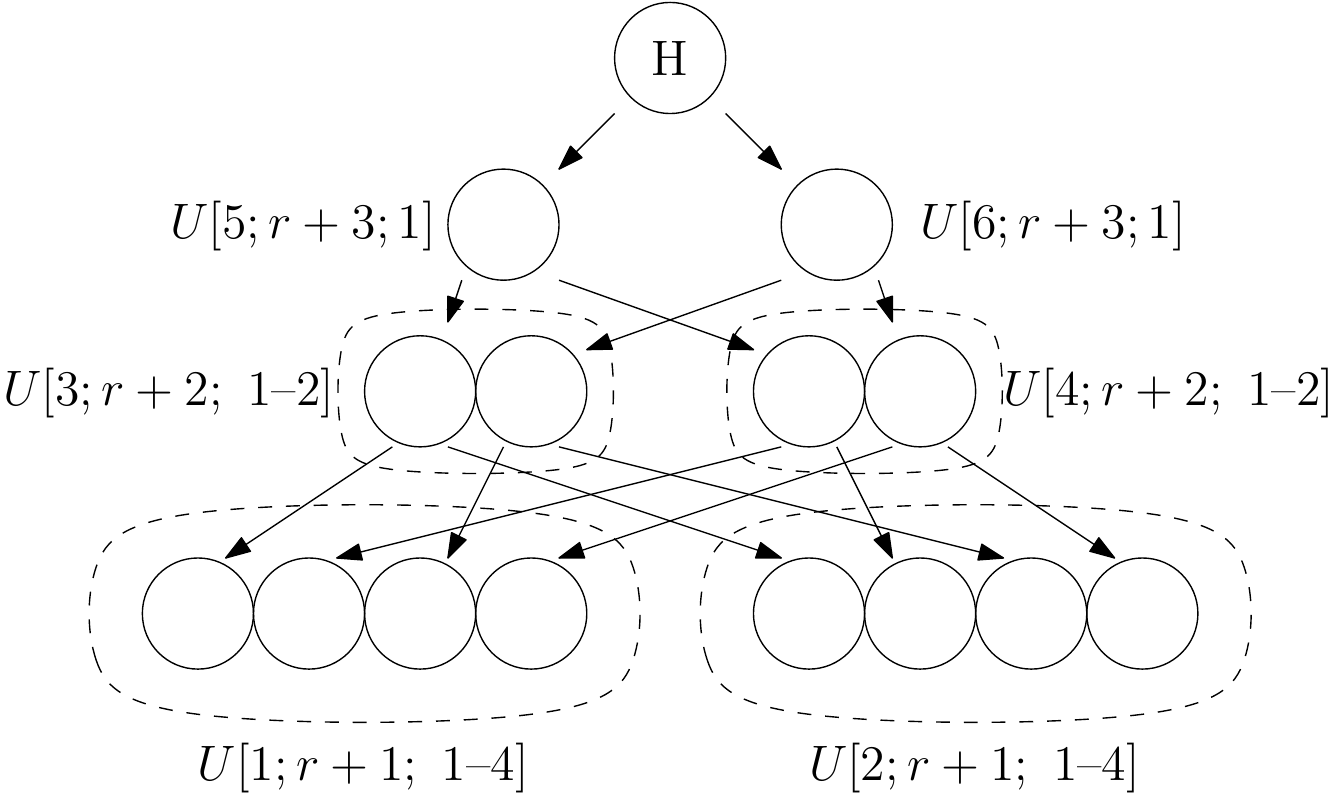}
    \caption{A "Fork Bomb" for $K=3$. Only the top unit (marked $H$) is created by an honest node. Every collection of units within a dashed oval depicts forks produced by a certain dishonest node.}
    \label{fig:bomb}
\end{figure}

In this section we describe a resource-exhaustion attack scenario on DAG-based protocols.
Generally speaking, this attack applies to protocols that allows adding nodes (which in our case are called units -- we stick to this name for the purpose of explaining the attack) to the DAG performing only ``local validation".
By this we mean that upon receiving a unit, the node checks only that its parents are correct and that the unit satisfies some protocol-specific rules that are checked based only on the data included in this unit and on its parents (possibly recursively), and adds this unit to the DAG if the checks are successful.
In particular, such a protocol has to allow adding forked units to the DAG, as it can happen that two distinct honest nodes built upon two distinct variants of a fork, hence there is no way to roll it back.
At this point we also remark that a mechanism of banning nodes that are proved to be forking (or malicious in some other way) is not sufficient to prevent this attack.

We note that while the \quickalef{} protocol (without the alert system) satisfies the above condition, the \alef{} protocol and the \quickalef{} extended with the alert system, do not.
This is because \alef{} uses reliable broadcast to disseminate units (which gives a non-local validation mechanism) and similarly the alert system for \quickalef{} adds new non-local checks before adding units created by forkers.
On the other hand, protocols such as Hashgraph~\cite{Baird16} or Blockmania~\cite{DH18} satisfy this condition and thus are affected, if not equipped with an appropriate defense mechanism.

The main idea of the attack is for a group of malicious nodes to create an exponentially large number of forking units, which all must be accepted by honest nodes if the protocol relies only on locally validating units.
The goal is to create so many valid forked units that honest nodes which end up downloading them (and they are forced to do so since these units are all valid) are likely to crash or at least slow down significantly because of exhausting RAM or disc memory.
Even though this leads to revealing the attacking nodes, there might be no way to 
punish them, if most of the nodes are unable to proceed with the protocol.

\noindent 
{\bf Description of the attack.}
We describe the attack using the language and rules specific to the \quickalef{} protocol (without the alert system).
For simplicity, when describing parents of units we only look at their dishonest parents and assume that apart from that, they also connect to all honest units of an appropriate round.
At the end of this section we also give some remarks on how one can adjust this description to any other DAG-based protocol.

Technically, the attack is performed as follows: for $K$ consecutive rounds all malicious nodes create forking units, but abstain from sending them to any honest node,
after thesee $K$ rounds, they start broadcasting forks from this huge set.
Assume that the malicious nodes participating in the attack are numbered $1, 2, \ldots, 2K-1, 2K$. Recall that $U[j;r]$ denotes the round-$r$ unit created by node $j$. We extend this notation to include forks -- let $U[j;r;i]$ denote the $i$-th fork, hence $U[j;r]=U[j;r;1]$.

Suppose it is round $r$ when the attack starts, it will then finish in round $r+K$.
We refer to Figure~\ref{fig:bomb} for an example of this construction for $K=3$.
A formal description follows: fix a $k \in \{1,2, \ldots, K\}$, the goal of nodes $2k-1$ and $2k$ is to fork at round $r+k$, more specifically:
\begin{enumerate}
    \item For the initial $k-1$ rounds (starting from round $r$) these nodes create regular, non-forking units.
    \item At round $r+k$ the nodes create $2^{K-k}$ forks each, which are sent only to other malicious nodes if $k<K$, and broadcast to every other node if $k=K$. Note that nodes $2K-1$ and $2K$ in fact do not fork, hence their units must be eventually accepted by honest nodes.
    \item Units $U[1;r+1;i]$ and $U[2;r+1;i]$ can have any round-$r$ units as parents, for $i=1,\ldots,2^{G-1}$.
    \item For $k>1$ unit $U[2k-1;r+k;i]$ has $U[2k-3;r+k-1;2i-1]$ and $U[2k-2;r+k-1;2i-1]$ as parents, for $i=1,\ldots,2^{K-k}$.
    \item For $k>1$ unit $U[2k;r+k;i]$ has $U[2k-3;r+k-1;2i]$ and $U[2k-2;r+k-1;2i]$ as parents, for $i=1,\ldots,2^{K-k}$.
\end{enumerate}

Validity of all the forks is easily checked by induction: all units $U[1;r+1;i]$ and $U[2;r+1;i]$ are valid, because there are no forks below them (although they are invalid in the sense that they are forked, this fact always becomes apparent only after some other valid nodes has built upon them). Now assume that units $U[2k-3;r+k-1;2i-1]$ and $U[2k-2;r+k-1;2i-1]$ are valid. Then the unit $U[2k-1;r+k;i]$ is valid, because there are no forks created by nodes $2k-3$ or $2k-2$ below it. Analogously, the unit $U[2k;r+k;i]$ is valid.

Since the nodes $2K-1$ and $2K$ cannot be banned, then units $U[2K-1;r+K;1]$ and $U[2K;r+K;1]$ with their lower cones must be accepted by all honest nodes, and both of them contain at least $2^{K+1}-2$ units in total, counting only $K$ most recent rounds.

To conclude let us give a couple of remarks regarding this attack:
\begin{enumerate}
    \item The attack is rather practical, as it does not require any strong assumptions about the adversary (e.g. no control over message delays is required). It is enough for the adversary to coordinate its pool of malicious nodes.

    \item The attack allows any number of parents for a unit. If the protocol requires exactly two parents per unit, with the previous unit created by the same node (as in Hashgraph) then one round of the attack can be realized in two consecutive rounds.
    
    \item If a larger number of parents is required, then the adversary can connect in addition to any subset of honest units. This is possible because honest nodes can always advance the protocol without the help of malicious nodes, and the number of required parents cannot be higher than the number of honest nodes.
\end{enumerate}

\end{document}